\documentclass[pra]{revtex4}
\usepackage{graphics,amsmath,amsfonts,amscd,revsymb,latexsym,
enumerate,multirow,epsfig}
\usepackage{amsmath}
\usepackage{amsfonts}
\usepackage{amssymb}
\usepackage[caption=false]{subfig}
\usepackage{color}
\usepackage{graphicx}
\providecommand{\U}[1]{\protect\rule{.1in}{.1in}}
\providecommand{\U}[1]{\protect\rule{.1in}{.1in}}
\newtheorem{theorem}{Theorem}

\newtheorem{definition}{Definition}

\newtheorem{Lemma}[theorem]{Lemma}

\newenvironment{proof}[1][Proof]{\noindent\textbf{#1.} }{\ \rule{0.5em}{0.5em}}

\newcommand{\pr}[1]{#1\mbox{$^\prime$} }
\newcommand{\ket}[1]{\left|#1\right\rangle}
\newcommand{\bra}[1]{\left\langle#1\right|}

\begin{document}
\preprint{ }
\title[ ]{Proof of efficient, parallelized, universal adiabatic quantum computation}
\author{Ari Mizel}
\email{ari@arimizel.com}
\affiliation{Laboratory for Physical Sciences, 8050 Greenmead Drive, College Park,
Maryland, USA\ 20740}
%\keywords{quantum information}
%\pacs{}

\begin{abstract}
We give a careful proof that a parallelized version of adiabatic quantum computation can efficiently simulate universal gate model quantum computation.  The proof specifies an explicit parameter-dependent Hamiltonian $H(\lambda)$ that is based on ground state quantum computation \cite{Mizel01}.  We treat both a 1-dimensional configuration in which qubits on a line undergo nearest neighbor 2-qubit gates and an all-to-all configuration in which every qubit undergoes 2-qubit gates with every other qubit in the system.
\end{abstract}
\volumeyear{2015}
\volumenumber{ }
\issuenumber{ }
\eid{ }
\date{\today}
%\received{}

%\revised{}

%\accepted{}

%\published{}

%\startpage{1}
%\endpage{ }
\maketitle

\section{Introduction}

Investigating the physics of computation, researchers studied computing in a thermodynamically reversible fashion \cite{Bennett79,Toffoli81,Fredkin82,Bennett82} and studied computing with a quantum mechanical system \cite{Benioff80,Manin80,Feynman81}.  In 1985, Feynman proposed a model \cite{Feynman85} addressing both together -- reversible classical computation with a quantum mechanical apparatus.  This paper formulated a ``ballistic'' version of quantum computing in which a time-dependent quantum system evolves from input state to output state in accordance with the dynamics imposed by a time-independent Hamiltonian.  At the end of his paper, Feynman mentioned the possibility of generalizing the model to permit parallelization (``concurrent operation'').  Margolus \cite{Margolus86,Margolus90} carried out such a generalization, deriving a 1+1-dimensional form of parallelized quantum computation, retaining the ballistic model of a time-dependent quantum system evolving via a time-independent Hamiltonian.  These ballistic versions of quantum computation are well-suited for the investigation of thermodynamic reversibility and other fundamental physics questions but seem impractical if one wishes to realize an actual working device.  Among other issues, the Hamiltonians proposed by Feynman and Margolis require at least 3-body interactions.

Seminars at Berkeley in 1998 introduced the idea of using a ground state to perform universal quantum computation \cite{Mizel01}.  This proposal was motivated by the practical problem of realizing a quantum computer.  In this approach, the time-independent ground state of a time-independent Hamiltonian is used to perform a quantum computation that is specified as a program of unitary gates.  This proposal rediscovered some technical features of the ballistic approach but introduced universal quantum computation in the ground state, rather than a time-dependent state.  The hope was that a system residing in the ground state would enjoy inherent robustness against decoherence.  In contrast with the earlier ballistic work, the ground state quantum computation (GSQC) model required only 2-body interactions in the time-independent Hamiltonian and also allowed parallelized quantum computation without Margolis's restriction to a 1+1-dimensional configuration.  There were early sketches of designs based on this model \cite{Mizel01b} and more theoretical work \cite{Mizel02,Mizel04,Mizel14}.  At around the same time, very different motivations led Kitaev to independently consider universal quantum computing with ground states \cite{Kitaev02}.  Kitaev was interested in computationally complexity theory in the quantum context, and he carefully studied the eigenstates of Feynman's model.  This work was not primarily focused on questions of physical implementation; indeed, Kitaev did not hesitate to incorporate 5-body interactions in his Hamiltonian.

The initial idea behind \cite{Mizel01} was to cool a system into the ground state of the GSQC Hamiltonian.  In 1999, a second GSQC paper \cite{Mizel02} was submitted that further scrutinized the original GSQC plan to reach the ground state by cooling.    The quantum adiabatic algorithm gained the attention of quantum computing researchers around the same time as a means of solving optimization problems \cite{Kadowaki98,Farhi00}.   As a result, by the time \cite{Mizel02} was published in 2002, a note had been inserted that GSQC could employ adiabatic evolution (``one approach might be to turn on gradually the tunneling matrix elements in the Hamiltonian'') to reach the ground state.

In 2004, another group of Berkeley-associated researchers studied this idea of performing universal quantum computation by adiabatic tuning of a ground state.  Taking Kitaev's work  \cite{Kitaev02} on the Feynman model as a point of departure, rather than the GSQC formalism \cite{Mizel01}, they carefully bounded the gap under adiabatic tuning of the Hamiltonian to show that universal adiabatic quantum computing could efficiently simulate the gate model in a noise-free setting \cite{Aharonov07}.  Their work was embraced by the quantum computer science community, and it generated substantial follow-up research.  In 2007, in response to this interest, \cite{Mizel07} made the adiabatic argument of \cite{Mizel02} more explicit for the GSQC Hamiltonian as opposed to the Feynman construction \cite{Feynman85}.  Members of the quantum computer science community have critiqued this work based upon its level of rigor and completeness \cite{Aharonov07} and upon its correctness \cite{Childs14,Breuckmann14,Gosset15}.  Recently, a new proof \cite{Breuckmann14} has appeared for a restricted 1+1-dimensional model \cite{Margolus90}.  This work is rigorous, in accordance with the expectations of the quantum computer science community, but it is restricted to a 1+1-dimensional model, includes no 1-qubit gates, and requires the number of qubits in the computation to exceed the depth of the computation -- a vexatious constraint.  In this paper, we revisit the proof of universal adiabatic quantum computation using GSQC \cite{Mizel07} and supply proofs that are rigorous, include 1-qubit gates, permit lengthy computations in which the depth exceeds the number of qubits, and address a configuration of qubits with all-to-all connectivity as well as nearest-neighbor connectivity in 1+1 dimensions.  The proof techniques that we introduce may be of broader interest.

This paper is organized as follows.  In section II, we review the GSQC Hamiltonian.  We invoke the adiabatic theorem to upper bound its adiabatic evolution time in terms of a lower bound on its gap.  Section III derives an expression for this lower bound on the GSQC Hamiltonian gap in terms of 2 quantities that are bounded for a circuit consistent with a 1-dimensional configuration of qubits.  The argument is extended to a circuit consistent with an all-to-all configuration in which every qubit interacts with every other qubit.  Section IV proves that probability of measuring the output of the computation is at least 1/3.  We conclude in section V.

\section{GSQC Hamiltonian and Computation Time}

An adiabatic quantum computation is specified in terms of a parameter-dependent Hamiltonian $H(\lambda)$.  One imagines preparing the system in the ground state of $H(0)$ and adiabatically increasing the parameter $\lambda$ from $0$ to $1$ to carry the system into the ground state of $H(1)$.  The initial Hamiltonian $H(0)$ is designed to have an easily-prepared ground state while the final Hamiltonian $H(1)$ is designed so that its ground state contains the solution to some computational problem.

For universal adiabatic quantum computation, the computational problem in question is any given gate model quantum algorithm.  It is natural to define the gate model algorithm in terms of a program of unitary gates acting on $M$ qubits during $N$ time steps.  The number of time steps $N$ is also known as the circuit depth, while the number of qubits $M$ is sometimes termed the circuit width.   For each time step $i \in [1,\dots,N]$ and each qubit $A \in [1,\dots,M]$, the algorithm specifies a 1-qubit gate $U_{A,i}  \in \mathbf{U}(2)$ acting on $A$ or a 2-qubit gate $U_{A,i;B,i}  \in \mathbf{U}(4)$ acting on $A$ and some partner qubit  $B$.  For universal adiabatic quantum computation, our task is to design a Hamiltonian $H(\lambda)$ such that the ground state of $H(1)$ contains the results of this gate model algorithm.

To address this task, begin by considering a trivial algorithm with $M=1$, consisting only of 1-qubit gates $U_{1,i}$.  The solution to the gate model algorithm is $U_{1,N} \dots U_{1,1} \left|0\right\rangle$.  There is a readily identified Hamiltonian of which this is the ground state, $-U_{1,N} \dots U_{1,1} \left|0\right\rangle \left<0\right| U_{1,1}^\dagger \dots U_{1,N}^\dagger$, but this Hamiltonian explicitly involves the matrix product $U_{1,N} \dots U_{1,1}$.  Implementing such a Hamiltonian would require a classical precomputation of $U_{1,N} \dots U_{1,1}$, obviating any need to execute the quantum algorithm.  We want $H(1)$ to depend upon the individual unitary gates (i.e. the algorithm itself) but not explicitly upon their product (i.e. the solution to the algorithm). 

We need to apply unitary gates sequentially within the ground state of $H(1)$.  We do so by promoting the classical time step variable $i$ to a quantum mechanical coordinate in Hilbert space.  In his construction for ballistic quantum computation \cite{Feynman85}, Feynman introduced a clock particle with Hilbert space $\left\{\left| i \right\rangle \mid i=0,\dots,N \right\}$.  We employ a different approach \cite{Mizel01} that avoids $3$-body interactions in our Hamiltonian.  We expand the Hilbert space of our computational particle itself from $\left\{\left|0\right\rangle,\left|1\right\rangle\right\}$ to $\left\{\left|0_i\right\rangle,\left|1_i\right\rangle \mid i = 0,\dots, N\right\}$.  It proves convenient to introduce fermionic creation operators $c^\dagger_{A,i,0}$ and $c^\dagger_{A,i,1}$ such that $\left|0_i\right\rangle$  and $\left|1_i\right\rangle$ of qubit $A$ are $c^\dagger_{A,i,0} \left| \mathrm{vac} \right\rangle$ and $c^\dagger_{A,i,1} \left| \mathrm{vac} \right\rangle$, respectively.  We also define the row vector $C^\dagger_{A,i} = \left[ c^\dagger_{A,i,0} \,\,\,\,\, c^\dagger_{A,i,1}\right]$.  In terms of these new operators, the solution  $U_{1,N} \dots U_{1,1} \left|0\right\rangle$ of our trivial $M=1$ qubit gate model algorithm becomes $C^\dagger_{1,N} U_{1,N} \dots U_{1,1} \left[\begin{array}{c} 1 \\ 0 \end{array} \right] \left| \mathrm{vac} \right\rangle$.  Rather than making this the ground state of $H(1)$, we consider the (unnormalized) state $C^\dagger_{1,N} U_{1,N} \dots U_{1,1} \left[\begin{array}{c} 1 \\ 0 \end{array} \right]\left| \mathrm{vac} \right\rangle+C^\dagger_{1,N-1} U_{1,N-1} \dots U_{1,1} \left[\begin{array}{c} 1 \\ 0 \end{array} \right]\left| \mathrm{vac} \right\rangle+ \dots + C^\dagger_{1,0}\left[\begin{array}{c} 1 \\ 0 \end{array} \right]  \left| \mathrm{vac} \right\rangle$.  This ``history" state is comprised of a sum of states corresponding to all time steps $i=0,\dots,N$ of the computation.

The history state is the ground state of the Hamiltonian 
$H(1)=h_1^{1}(\mathrm{INIT})+\sum_{i=1}^{N}h^{i}_{1}(U_{1,i})$ where 
\begin{equation}
h^{i}_A(U) \equiv \mathcal{E}\left[\left.\left.C_{A,i}^{\dagger }-C_{A,i-1}^{\dagger}U^{\dagger}\right] \right[ C_{A,i}-UC_{A,i-1}\right].
\label{eq:onequbitgate}
\end{equation}
Here,
\begin{equation}
h^i_{A}(\mathrm{INIT}) = \mathcal{E} c^\dagger_{A,i,1} c_{A,i,1}
\label{eq:init}
\end{equation}
``initializes" qubit $A$ by imposing an energy penalty on the state $C^\dagger_{1,N} U_{1,N} \dots U_{1,1} \left[\begin{array}{c} 0 \\ 1 \end{array} \right]\left| \mathrm{vac} \right\rangle+C^\dagger_{1,N-1} U_{1,N-1} \dots U_{1,1} \left[\begin{array}{c} 0 \\ 1 \end{array} \right]\left| \mathrm{vac} \right\rangle+ \dots + C^\dagger_{1,0}\left[\begin{array}{c} 0 \\ 1 \end{array} \right]  \left| \mathrm{vac} \right\rangle$, an alternate history state corresponding to a different input.  For a general $U$, the term (\ref{eq:onequbitgate}) causes qubit $A$ to transition from $\left|0_{i-1}\right\rangle$ to a superposition of $\left|0_i\right\rangle$ and $\left|1_i\right\rangle$ while $\left|1_{i-1}\right\rangle$ transitions to an orthogonal superposition of $\left|0_i\right\rangle$ and $\left|1_i\right\rangle$.  In the case in which $U$ is just the $2 \times 2$ identity matrix $\mathcal{I}$, the term (\ref{eq:onequbitgate}) causes qubit $A$ to transition from $\left|0_{i-1}\right\rangle$ to  $\left|0_i\right\rangle$ and from $\left|1_{i-1}\right\rangle$ to  $\left|1_i\right\rangle$.  Fig. \ref{fig:apparatus1} depicts a schematic physical realization of the $H(1)$.

\begin{figure}[htp]
\includegraphics[width=2.5in]{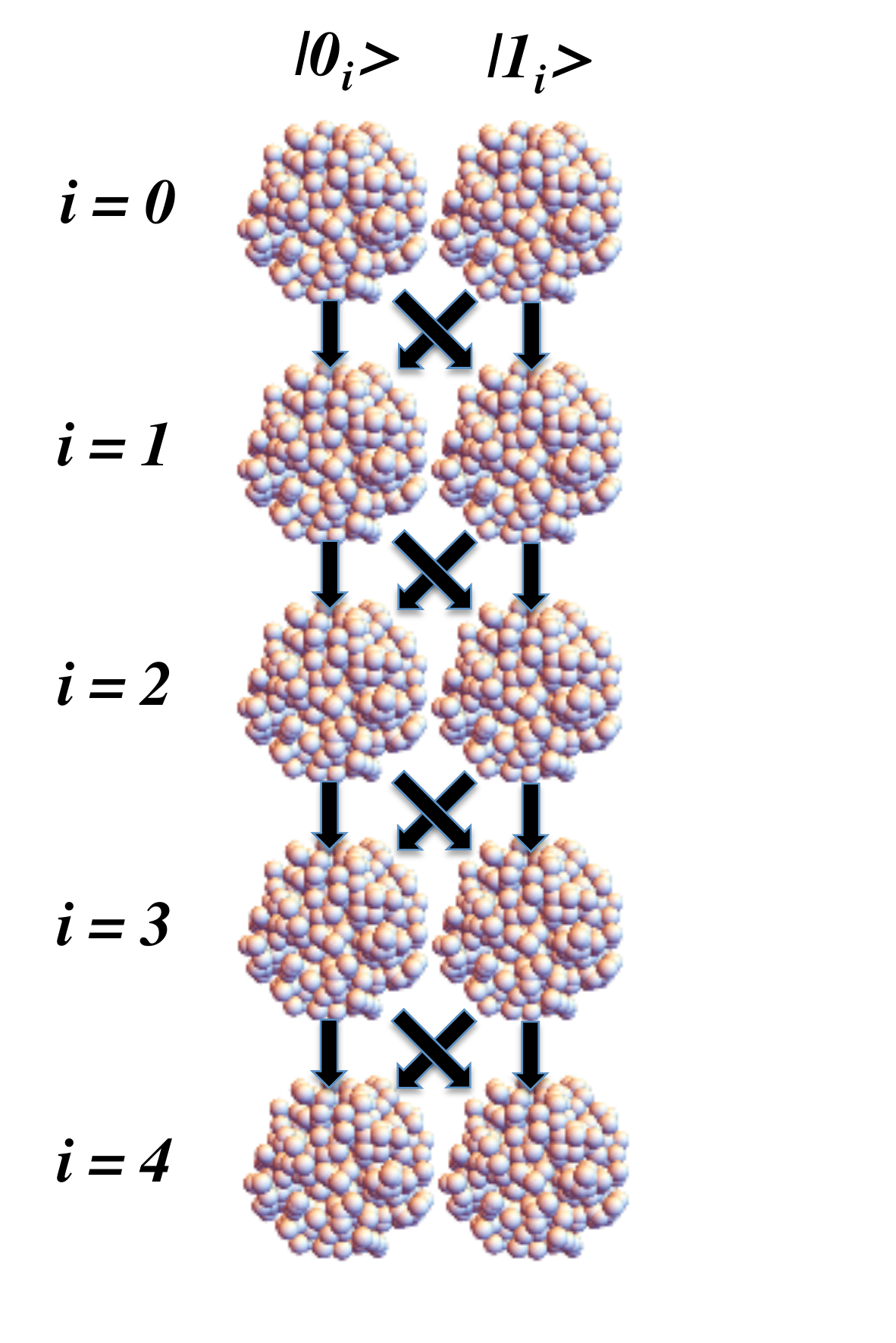}%
\caption{Schematic physical apparatus for realizing $H(1)$ in case of a single qubit.  The  $2 \times 5$ array of nanocrystals is assumed to contain a single itinerant electron that is free to migrate around the array.  Each nanocrystal is imagined to house a single available orbital.  Black arrows indicate tunnelling matrix elements from orbital to orbital from the beginning of the program to the end.  (Since the Hamiltonian is Hermitian, the electron can tunnel along or opposite the arrows.) }
\label{fig:apparatus1}
\end{figure}

To reliably extract the computational output from the history state, it is possible to rescale \cite{Mizel01} the coefficient of its term $C^\dagger_{1,N} U_{1,N} \dots U_{1,1} \left[\begin{array}{c} 1 \\ 0 \end{array} \right]\left| \mathrm{vac} \right\rangle$ or to append identity gates to the end of the calculation \cite{Feynman85,Aharonov07}.  The first approach is probably more suitable if one is interested in physically constructing a computer, since it involves less hardware.  The second approach is usually simpler for proofs.  

One can generalize this framework to multiqubit gate model computations.  We assume that each set of qubit states $C_{A,i}$ is occupied by a single particle: the system is an eigenstate of $\sum_{i=0}^N C_{A,i}^{\dagger}C_{A,i}$, with eigenvalue $+1$, for all $A=1$ to $M$.  For a trivial multiqubit algorithm consisting of only 1-qubit gates $U_{A,i}$, the solution to the gate model algorithm is $ \otimes_{A=1}^M \left( \Pi _{i=1}^N  U_{A,i} \left|0\right\rangle \right)$.  The corresponding history state is
\[
\otimes_{A=1}^M \left( C^\dagger_{A,N} U_{1,N} \dots U_{1,1} \left[\begin{array}{c} 1 \\ 0 \end{array} \right] +C^\dagger_{A,N-1} U_{1,N-1} \dots U_{1,1} \left[\begin{array}{c} 1 \\ 0 \end{array} \right]+ \dots + C^\dagger_{A,0}\left[\begin{array}{c} 1 \\ 0 \end{array} \right]\right)  \left| \mathrm{vac} \right\rangle;
\]
it is the ground state of 
\begin{equation}
H(1)=\sum_{A=1}^{M} h^1_{A}(\mathrm{INIT}) + \sum_{A=1}^{M} \sum_{i=1}^{N}h^{i}_{A}(U_{A,i}).
\label{eq:independentqubitHamiltonian}
\end{equation}
Fig. \ref{fig:apparatus}(a) generalizes Fig. \ref{fig:apparatus1} to the many-qubit case of Hamiltonian (\ref{eq:independentqubitHamiltonian}).

\begin{figure}[htp]
\subfloat[Schematic physical apparatus for realizing the Hamiltonian (\ref{eq:independentqubitHamiltonian}) in the case of $M=2$ qubits.]{%
  \includegraphics[width=2.5in]{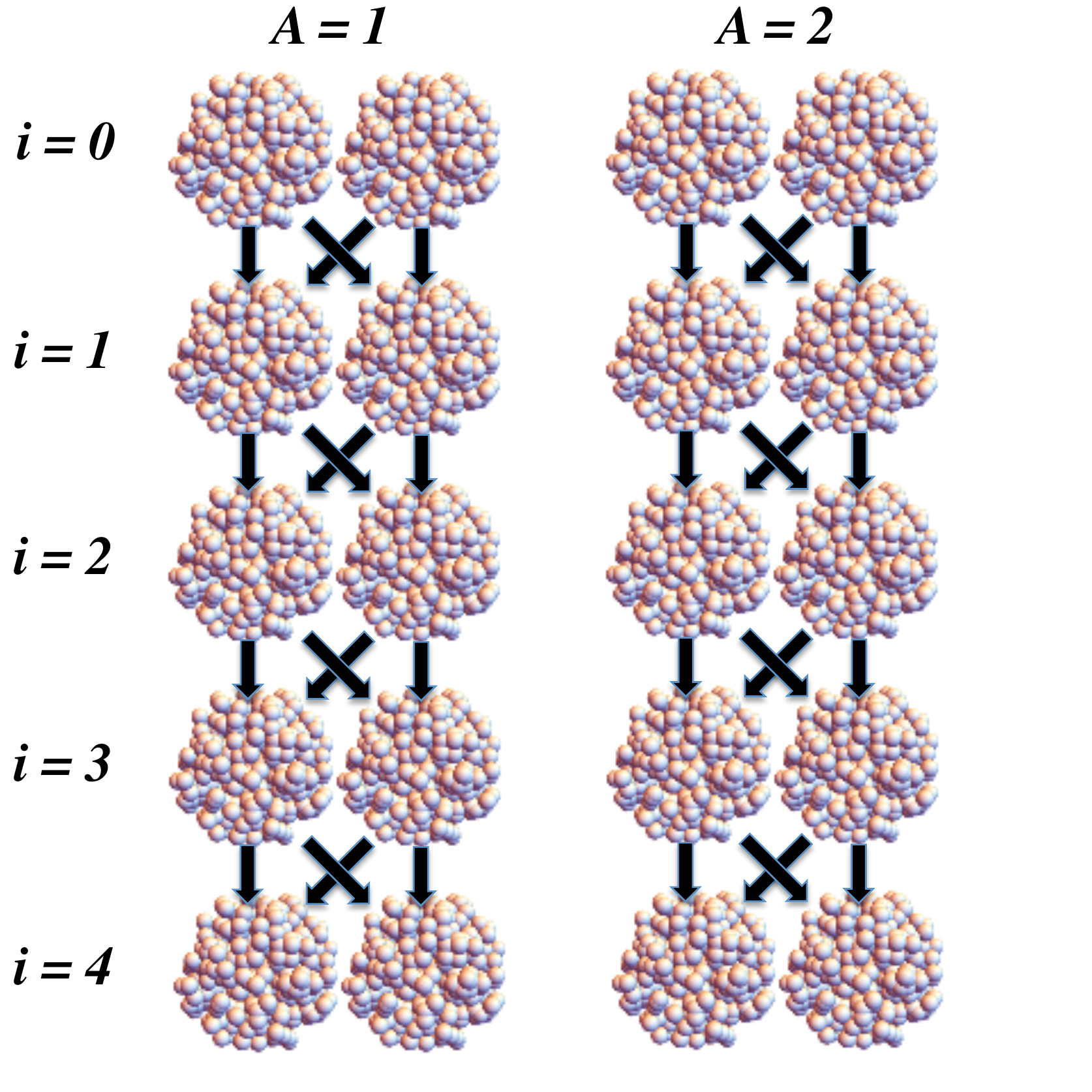}%
}

\subfloat[Isolated single time step of (a) in which each qubit undergoes a 1-qubit gate.]{%
  \includegraphics[width=2.5in]{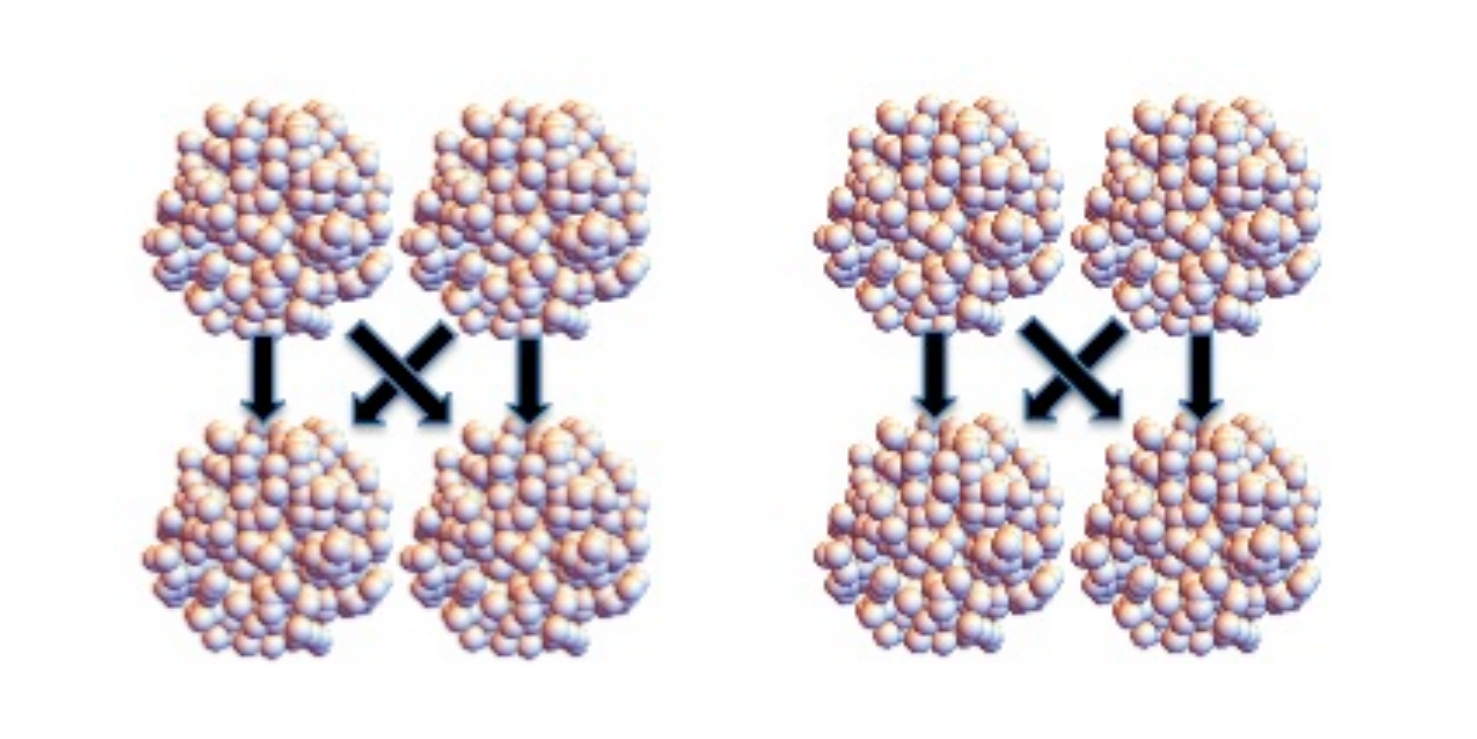}%
}
\subfloat[Isolated single time step, suitable for replacing (b), in which 2 qubits undergo a 2-qubit gate.]{%
  \includegraphics[width=2.5in]{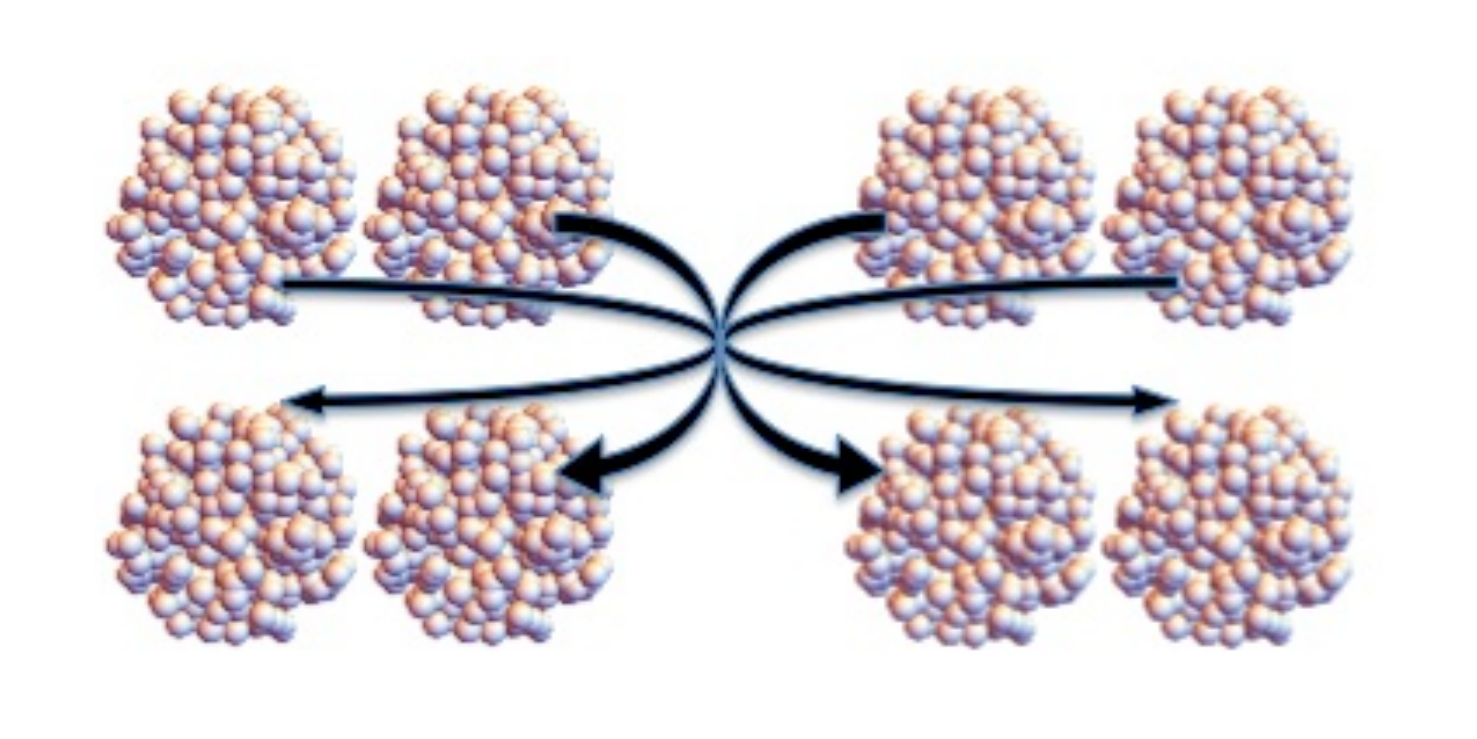}%
}
\caption{Schematic physical apparatus for realizing Hamiltonian (\ref{eq:independentqubitHamiltonian}) in the case $M=2$, $N=4$.}
\label{fig:apparatus}
\end{figure}

When a 2-qubit gate is included in the multiqubit gate model computation, one replaces $U_{A,j} \otimes U_{B,j}$ in $\otimes _{A=1}^M \left(\Pi _{i=1}^N  U_{A,i} \left|0\right\rangle\right)$ with $U_{A,j;B,j} \in \mathbf{U}(4)$.  In the history state, we similarly replace $U_{A,j} \otimes U_{B,j}$ with $U_{A,j;B,j}$.  There is a complication here since the history state initially contains some problematic terms in which $U_{A,j}$ appears but not $U_{B,j}$ and some problematic terms in which $U_{B,j}$ appears but not $U_{A,j}$.  We resolve the complication by deleting ``time-invalid'' terms from the history state that contain $U_{A,i_A} \otimes U_{B,i_B}$ with $i_A$ and $i_B$ on opposite sides of $j$, such that $(i_A - j+1/2)(i_B-j+1/2) < 0$.  The resulting history state is the ground state of $H(1)$ provided that we replace $h^{j}_A(U_{A,j})+h^{j}_B(U_{B,j})$ with $h_{A,B}^{j,j}(U_{A,j;B,j}) + h_{A,B}^{j,j}({\mathrm{P}})$.  Here,
\begin{eqnarray}
h_{A,B}^{i,j}(U)&  \equiv & \mathcal{E}\left[ C_{A,i}^{\dagger } \otimes C_{B,j}^{\dagger } - C_{A,i-1}^{\dagger } \otimes C_{B,j-1}^{\dagger }  U^{\dagger }\right] \nonumber\\
& &  \left[  C_{A,i} \otimes C_{B,j} - U C_{A,i-1} \otimes C_{B,j-1}  \right], \label{eq:twoqubitgate}
\end{eqnarray}
for any 2-qubit unitary operator $U$, while
\begin{eqnarray}
h_{A,B}^{i,j}({\mathrm{P}})& =& \mathcal{E}\!\!\!\!\!\sum_{k<i,l\geq j}\!\!\!\!C_{A,k}^{\dagger}C_{A,k} \otimes C_{B,l}^{\dagger }C_{B,l} \nonumber \\
& & +C_{A,l}^{\dagger }C_{A,l} \otimes C_{B,k}^{\dagger}C_{B,k} 
\label{eq:penalty}
\end{eqnarray}
imposes an energy penalty on the time-invalid terms that we deleted from the history state.  If $ h_{A,B}^{j,j}({\mathrm{P}})$ were absent from our Hamiltonian, the time-invalid terms would become additional zero-energy ground states: a ground state supported by a time-invalid region in which qubit $A$ resides at time steps $i_A < j$ and qubit $B$ resides at time steps $i_B \ge j$ and a ground state supported by a time-invalid region in which qubit $A$ resides at time steps $i_A \ge  j$ and qubit $B$ resides at time steps $i_B < j$ \endnote{As noted by \cite{Breuckmann14}, the penalty Hamiltonian was omitted from \cite{Mizel02}.  This was corrected in \cite{Mizel07}.  That revision, with the others we make here, should address any shortcomings in \cite{Mizel02}.}.  Fig. \ref{fig:apparatus}(c) gives a schematic representation of the 2-qubit gate Hamiltonian $h_{A,B}^{j,j}(U_{A,j;B,j}) + h_{A,B}^{j,j}({\mathrm{P}})$ that replaces the two 1-qubit Hamiltonians $h^{j}_A(U_{A,j})+h^{j}_B(U_{B,j})$ sketched in Fig. \ref{fig:apparatus}(b).

To use the GSQC Hamiltonian for computation, we would like to ensure that system occupies the ground state of $H(1)$.  The original GSQC proposal \cite{Mizel01} suggested that cooling could drive the system down to the ground state.  Later, the idea of computation by adiabatic evolution became widely known, and \cite{Mizel02} pointed out that it could be applied to GSQC.  To adiabatically carry the system into the ground state, we  introduce an adiabatic parameter $\lambda$ into the Hamiltonian.  We define $H(\lambda)$ so that $H(0)$ has a relatively simple ground state of $H(0)$ and carry it into the ground state of $H(1)$ by adiabatic evolution from $\lambda = 0$ to $1$.  The following definition introduces $\lambda$ dependence into the GSQC Hamiltonian in a suitable fashion.   It also includes a minor generalization that makes it possible to introduce and retire qubits at intermediate time-steps.

\begin{definition}[GSQC Hamiltonian]
Consider a gate model quantum algorithm defined by circuit width $M$, circuit depth $N$, 1-qubit gates $\mathcal{G}_1 \equiv \{U_{A,i}\}$, and 2-qubit gates $\mathcal{G}_2 \equiv \{U_{A,j;B,j}\}$.  Suppose that the first gate acting on qubit $A$ occurs at its original time step $o_A\ge 1$, the last gate acting on qubit $A$ occurs at its final time step $f_A \le N$, and an element of either $\mathcal{G}_1$ or $\mathcal{G}_2$ acts on qubit $A$ at every time step between $o_A$ and $f_A$.  Suppose further that the first gate acting on every qubit is a 1-qubit identity gate $\mathcal{I}$.    The corresponding ground state quantum computation (GSQC) Hamiltonian with $M$ independent $\lambda$ parameters is
\begin{eqnarray}
\lefteqn{H^{\text{independent}}(\lambda_1,\dots,\lambda_M)   =\sum_{A=2}^{M} \mathcal{E} (1-\lambda_{A-1}^3)  \sum_{i=o_A }^{f_A}C^\dagger_{A,i} C_{A,i}+ \sum_{A=1}^{M} h^{o_A-1}_{A}(\mathrm{INIT})+  h^{o_A}_{A}(\lambda_A \mathcal{I}) } \nonumber\\
&& + \sum_{i=1}^N  \left(\sum_{U_{A,i} \in \mathcal{G}_1} \left(1- \delta_{i,o_A} \right) h^{i}_{A}(U_{A,i}) + \sum_{U_{A,i;B,i} \in \mathcal{G}_2}  h_{A,B}^{i,i}(U_{A,i;B,i})  + h_{A,B}^{i,i}({\mathrm{P}})\right).
\label{eq:Hlambdaindependent}
\end{eqnarray}
This Hamiltonian is assumed to act within a basis of eigenstates of $\sum_{i=o_A-1}^{f_A} C_{A,i}^{\dagger}C_{A,i}$, with $+1$ eigenvalue, for $A \in [1,\dots,M]$.  

In terms of (\ref{eq:Hlambdaindependent}), a GSQC Hamiltonian with a single $\lambda$ parameter is defined using functions $\lambda_A(\lambda)$
\begin{equation}
H(\lambda) = H^{\text{independent}}(\lambda_1(\lambda),\lambda_2(\lambda),\dots,\lambda_M(\lambda)).
\label{eq:Hlambda}
\end{equation}
The function $\lambda_A(\lambda)$ has the property that $H(0) = H^{\text{independent}}(0,\dots,0)$ and  $\lambda_{A+1}$ increases from 0 to 1 after $\lambda_{A}$ has already finished increasing from 0 to 1.  For concreteness, set $\lambda_A(\lambda)= F(M\lambda-A)$ where $F(x) = (10x^3-15 x^4 + 6 x^5) \Theta(x) \Theta(1-x) + \Theta(x-1)$, in terms of the Heaviside function $\Theta(x)$.  This function begins at $F(x) = 0$ for $x \le 0$ and increases to reach $F(x) = 1$ for $x \ge 1$ with well-behaved first and second derivatives.
\label{def:GSQC}
\end{definition}

For $H(\lambda)$ to provide an appropriate adiabatic path to $H(1)$, it is desirable for $H(\lambda)$ to have a non-degenerate ground state at all values of $\lambda$.  This is demonstrated in the following theorem.
\begin{theorem}[Ground state of GSQC Hamiltonian]
The GSQC Hamiltonian of Def. \ref{def:GSQC} has a non-degenerate, zero energy ground state $|\psi_0 (\lambda) \rangle$.
\label{thm:groundstate}  
\end{theorem}
\begin{proof}
First note that the GSQC Hamiltonian is a sum of positive semi-definitive terms.  It therefore cannot have any negative eigenvalues.  We will show that it does have a zero eigenvalue, which must be the ground state energy, and that the associated eigenstate is unique.

The GSQC Hamiltonian is assumed to act within a specific basis, as described in Def. \ref{def:GSQC}.  We can write an arbitrary state in this basis in the form 
\begin{equation}
\left[ \sum_{i_1,\dots,i_M} \sum_{b_1,\dots,b_M} \psi_{i_1,\dots,i_M}(b_1,\dots,b_M) \otimes_{A=1}^M c^\dagger_{A,i_A,b_A}   \right]\left| \mathrm{vac} \right\rangle
\label{eq:startform}
\end{equation}
where $\psi_{i_1,\dots,i_M}(b_1,\dots,b_M)$ is a complex number that vanishes if any $i_A$ is outside the range $o_A-1 \le i_A \le f_A$.

We construct the zero energy ground state as follows.  Suppose that a given term $h^{i}_{A}(U_{A,i})$ appears in (\ref{eq:Hlambda}).  If such a term is to annihilate (\ref{eq:startform}), the equation $\psi_{i_1,\dots,i_A=i,\dots,i_M} (b_1,\dots,b_A,\dots,b_M) = \sum_{\beta_A = 0}^1 [U_{A,i}]_{b_A;\beta_A} \psi_{i_1,\dots,i_A=i-1,\dots,i_M}(b_1,\dots,\beta_A,\dots,b_M)$ must hold, where $[U_{A,i}]_{b_A;\beta_A}$ is a matrix element of $U_{A,i}$.   (When $h^{i}_{A}(\lambda_A \mathcal{I} )$ appears in (\ref{eq:Hlambda}), the right hand side of the equation acquires an extra factor of $\lambda_A$: $\psi_{i_1,\dots,i_A=i,\dots,i_M} (b_1,\dots,b_A,\dots,b_M) = \lambda_A \sum_{\beta_A = 0}^1 \delta_{b_A;\beta_A} \psi_{i_1,\dots,i_A=i-1,\dots,i_M}(b_1,\dots,\beta_A,\dots,b_M)$.)  Suppose that a given term $h_{A,B}^{i,i}(U_{A,i;B,i}) + h_{A,B}^{i,i}({\mathrm{P}})$ appears in (\ref{eq:Hlambda}).  If such a term is to annihilate (\ref{eq:startform}), the equation $\psi_{i_1,\dots,i_A=i,\dots,i_B=i,\dots,i_M}(b_1,\dots,b_A,\dots,b_B,\dots,b_M) = \sum_{\beta_A,\beta_B=0}^1 [U_{A,i;B,i}]_{b_A,b_B;\beta_A,\beta_B}\psi_{i_1,\dots,i_A=i-1,\dots,i_B=i-1\dots,i_M}(b_1,\dots,\beta_A,\dots,\beta_B,\dots,b_M)$ must hold.  Moreover, we must have $\psi_{i_1,\dots,i_A\ge i,\dots,i_B<i,\dots,i_M}(b_1,\dots,b_A,\dots,b_B,\dots,b_M)=\psi_{i_1,\dots,i_A<i,\dots,i_B\ge i,\dots,i_M}(b_1,\dots,b_A,\dots,b_B,\dots,b_M)=0$ if $h_{A,B}^{i,i}({\mathrm{P}})$ is to annihilate (\ref{eq:startform}).  From these equations, we can uniquely deduce $\psi_{i_1,\dots,i_A,\dots,i_M}(b_1,\dots,b_M)$ from $\psi_{i_1,\dots,i_A-1,\dots,i_M}(b_1,\dots,b_M)$.  Thus, starting with $\psi_{i_1=o_1-1,\dots,i_M=o_M-1}(b_1,\dots,b_M)$, we can uniquely deduce every $\psi_{i_1,\dots,i_M}(b_1,\dots,b_M)$, thereby constructing a zero energy state.

The state is annihilated by $\sum_{A=1}^{M} h^{o_A-1}_{A}(\mathrm{INIT})$, so $\psi_{i_1=o_1-1,\dots,i_M=o_M-1}(b_1,\dots,b_M)$ must vanish unless $b_1 = \dots = b_M = 0$.   The starting state $\psi_{i_1=o_1-1,\dots,i_M=o_M-1}(b_1,\dots,b_M)$ is therefore unique up to an overall constant.  Since every $\psi_{i_1,\dots,i_A,\dots,i_M}(b_1,\dots,,b_M)$ uniquely follows from the starting state, we conclude that $\sum_{A=1}^{M} h^{o_A-1}_{A}(\mathrm{INIT})+  h^{o_A}_{A}(\lambda_A \mathcal{I}) + \sum_{i=1}^N  \left(\sum_{U_{A,i} \in \mathcal{G}_1} \left(1- \delta_{i,o_A} \right) h^{i}_{A}(U_{A,i}) + \sum_{U_{A,i;B,i} \in \mathcal{G}_2}  h_{A,B}^{i,i}(U_{A,i;B,i})  + h_{A,B}^{i,i}({\mathrm{P}})\right)$ has a unique, zero energy ground state.

There is an extra term $\sum_{A=2}^{M} \sum_{i=o_A }^{f_A}\mathcal{E} (1-\lambda_{A-1}^3)  C^\dagger_{A,i} C_{A,i}$ in the definition of $H(\lambda)$, but this term also annihilates the unique, zero energy ground state.  This is because, in the ground state, qubit $A$ is localized at time step $o_A-1$ as long as $\lambda_A=0$; thus, $\sum_{i=o_A }^{f_A}\mathcal{E} (1-\lambda_{A-1}^3)  C^\dagger_{A,i} C_{A,i}$ annihilates the ground state as long as $\lambda_A=0$.  Moreover, we have defined $H(\lambda)$ so that $\lambda_A$ deviates from $0$ only after $\lambda_{A-1}$ becomes fixed at $1$ rendering $\sum_{i=o_A }^{f_A}\mathcal{E} (1-\lambda_{A-1}^3)  C^\dagger_{A,i} C_{A,i} = 0$.
\end{proof}

Note that the ground state of $H(0)$ is simply the product state $\otimes_{A=1}^M c^\dagger_{A,0,0} \ket{\mathrm{vac}}$, which is presumably straightforward to prepare.  The strategy of adiabatic quantum computation is to prepare the ground state of $H(1)$ by slowly carrying $\lambda$ from 0 to 1.  To assess the efficiency of this strategy, we bound the required amount of time.  The adiabatic theorem provides an estimate of this time given the energy gap between the ground state of $H(\lambda)$ and its first-excited state.  Myriad forms of the adiabatic theorem have appeared in the literature; a review appears in \cite{Albash18}.  We invoke one recent form from \cite{Jansen07}. 

\begin{theorem}[Adiabatic theorem \cite{Jansen07}]
Suppose that the spectrum of $H(\lambda)$  restricted to $P(\lambda)$  consists of $m(\lambda)$  eigenvalues (each possibly degenerate, crossing permitted) separated by a gap $g(\lambda)$  from the rest of the spectrum of $H(\lambda)$, and $H$, $\dot{H}$ , and $\ddot{H}$ are bounded operators. (This assumption is always fulfilled in finite-dimensional spaces.) Then, for  $\left|\psi \right\rangle \in P(0)$  and any $\lambda \in \left[0,1\right]$,
\[
\left \langle \psi \right| {U^\dagger}_T(\lambda) (1 - P(\lambda))  U_T(\lambda)  \left|\psi \right\rangle \le A(\lambda)^2, \|P_T(\lambda)-P(\lambda)\| \le A(\lambda),
\]
where
$A(\lambda) \le \frac{1}{T} \left[ \frac{m\| \dot{H}  \|}{g^2} |_{\lambda=0} + \frac{m\| \dot{H}  \|}{g^2} |_{\lambda}  +  \int_0^\lambda \left(\frac{m \|\ddot{H}\|}{g^2} + 7 m \sqrt{m} \frac{\| \dot{H} \|^2}{g^3}\right) \right]$.
\label{thm:adiabatic}
\end{theorem}

To apply this theorem to our $H(\lambda)$,  we note that $m(\lambda) = 1$.  Using the triangle inequality for the norm, we deduce $\| \dot{H}  \| \le |2\lambda_A \dot{\lambda_A}| +| \dot{\lambda_A}|+|3\lambda_A^2 \dot{\lambda_A}| \le 6(15/8) M\mathcal{E} \le 12 M\mathcal{E}$, and $\| \ddot{H}  \| \le |2 \dot{\lambda_A}^2 + 2 \lambda_A \ddot{\lambda_A}| + |\ddot{\lambda_A}| + |6 \lambda_A \dot{\lambda_A}^2+3 \lambda_A^2 \ddot{\lambda_A}| \le 8 |\dot{\lambda_A}|^2 + 6| \ddot{\lambda_A}| \le (8(15/8)^2 + 6(10/\sqrt{3})) M^2\mathcal{E} \le 63M^2\mathcal{E}$ using the expression for $F(x)$ in Def. \ref{def:GSQC}.   Therefore, the amount of excitation of the ground state during adiabatic evolution is bounded by the quantity $\left((24M+63M^2)\mathcal{E}/g_{min}^2 + 7(144 M^2)\mathcal{E}^2/g_{min}^3\right)/T < 1095M^2\mathcal{E}^2/Tg_{min}^3$. Thus, the amount of excitation will be small if $T \gg 1095M^2\mathcal{E}^2/g_{min}^3$.  To use this expression, we need to derive a lower bound on the energy gap $g_{min}$.  Note that this expression is quite pessimistic, and a time $T \gg 1/g_{min}$ may be sufficient \cite{Jansen07}.  However, to ensure rigor, we employ the pessimistic expression in this paper.

\begin{figure}[htp]
\subfloat[Schematic physical apparatus for realizing Hamiltonian (\ref{eq:independentqubitHamiltonian}) for $M=2$ qubits.  Just the $\left|0_{i-1}\right\rangle$ to $\left|0_{i}\right\rangle$ transitions are shown; similar figures including $\left|1_{i-1}\right\rangle$ to $\left|1_{i}\right\rangle$ transmissions are omitted.]{%
  \includegraphics[width=2.5in]{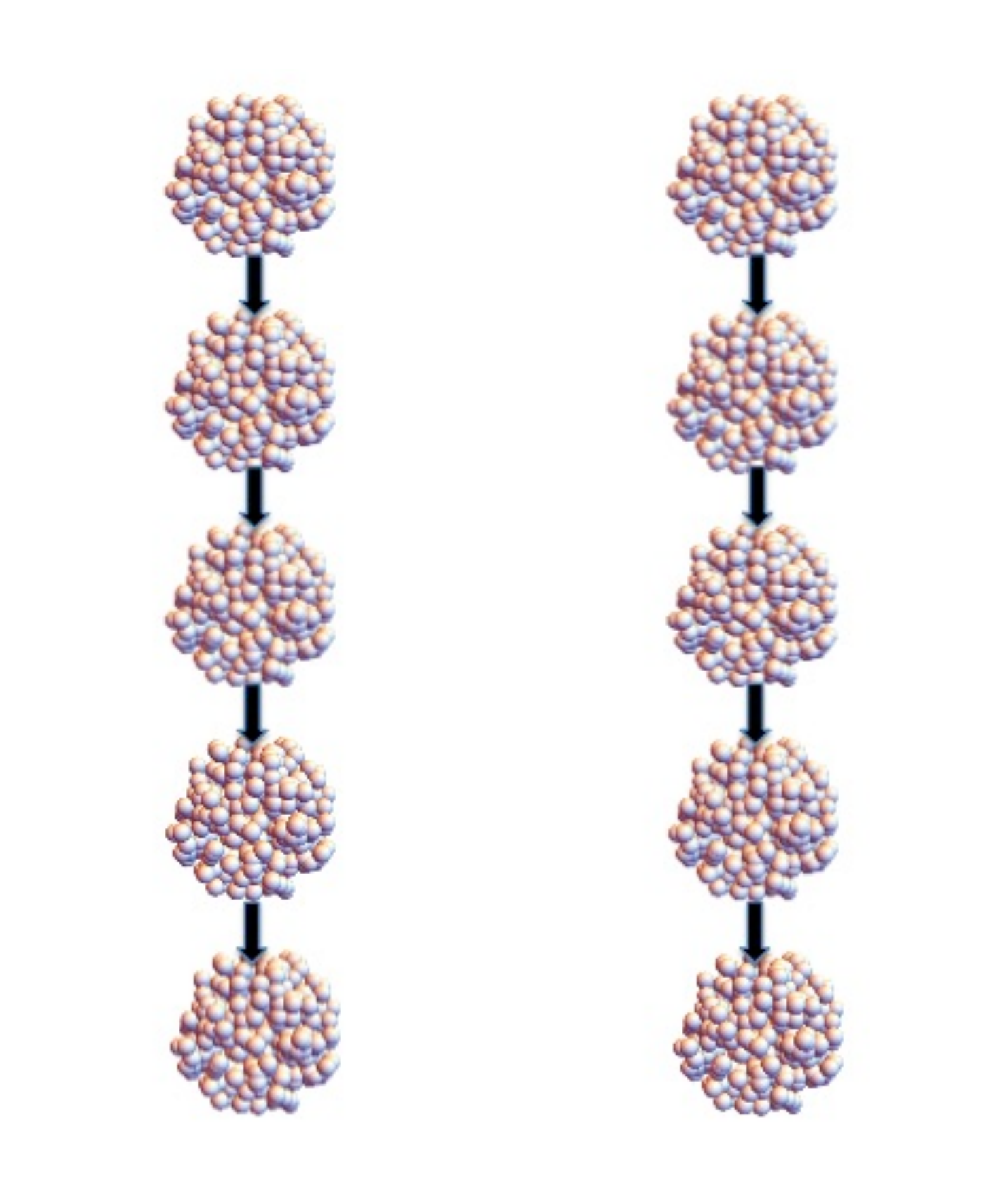}%
}

\subfloat[Isolated single time step of (a) in which each qubit undergoes a 1-qubit gate.  Just the $\left|0_{i-1}\right\rangle$ to $\left|0_{i}\right\rangle$ transitions are shown.]{%
  \includegraphics[width=2.5in]{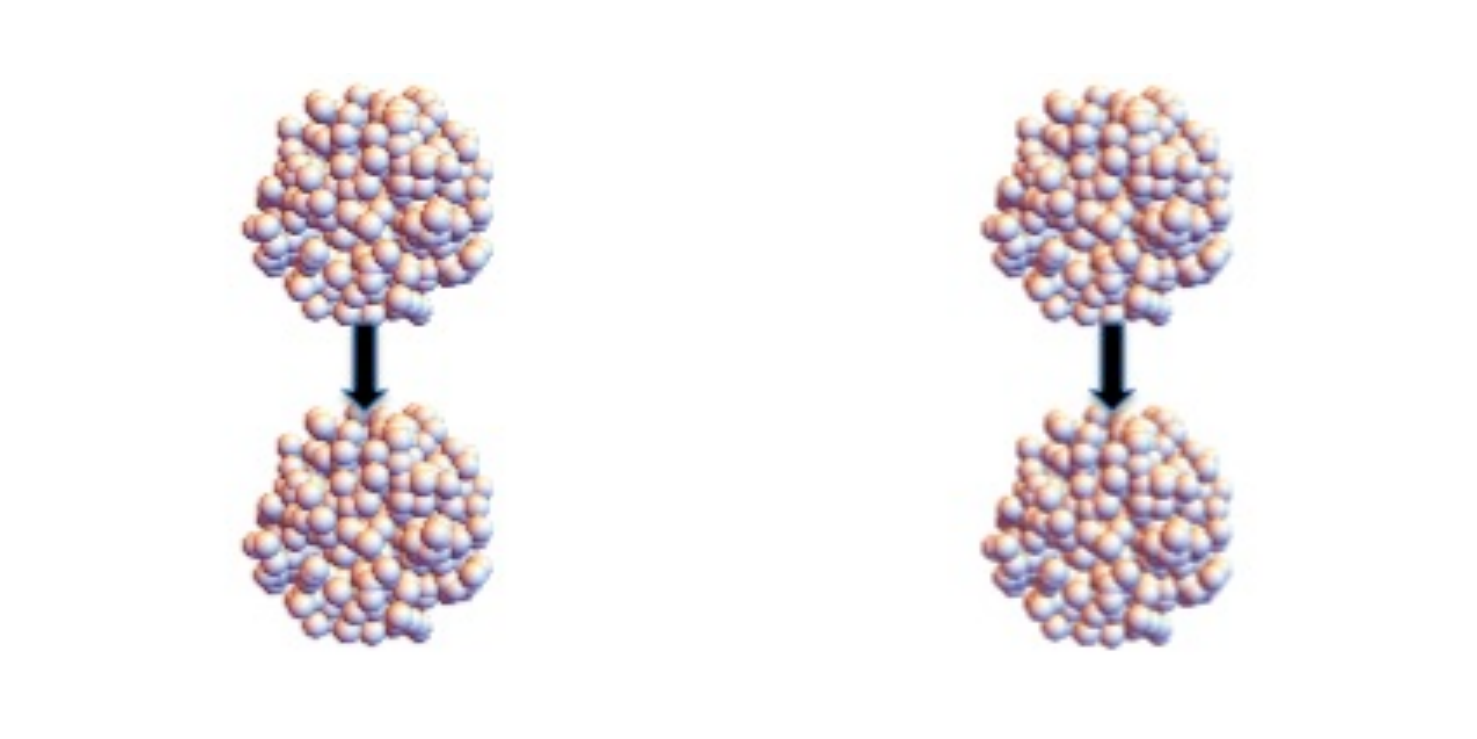}%
}
\subfloat[Isolated single time step, suitable for replacing (b), in which 2 qubits undergo a 2-qubit gate with $U = \mathcal{I} \otimes \mathcal{I}$.  Just the $\left|0_{i-1}\right\rangle \otimes \left|0_{j-1}\right\rangle $ to $\left|0_{i}\right\rangle \otimes \left|0_{j}\right\rangle$ transition is shown. ]{%
  \includegraphics[width=2.5in]{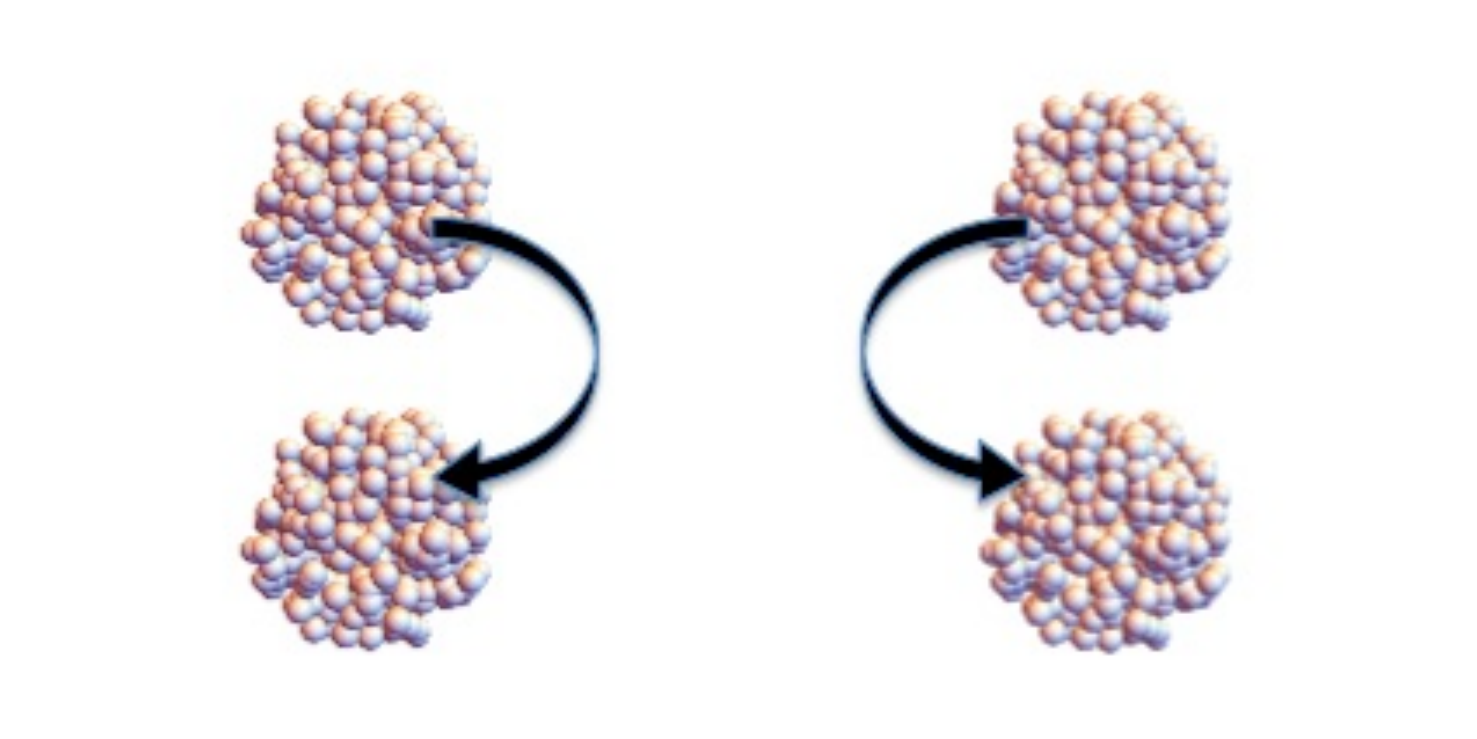}%
}
\caption{Schematic physical apparatus for realizing Hamiltonian (\ref{eq:independentqubitHamiltonian}) in the case that $U = \mathcal{I}$ at every time step.}
\label{fig:apparatusidentity}
\end{figure}

\section{Gap of $H(\lambda)$}

The following lemma simplifies our analysis of $g_{min}$ by allowing us to substitute identity gates for the unitary gates $\mathcal{G}_1 \equiv \{U_{A,i}\}$, and $\mathcal{G}_2 \equiv \{U_{A,j;B,j}\}$ in $H(\lambda)$.

\begin{Lemma}[Identity gates]

Starting with $H(\lambda)$, define a new Hamiltonian $\pr{H}(\lambda)$ in which every 1-qubit gate $U_{A,j}$ is replaced with an identity gate $\mathcal{I}$  and every 2-qubit gate $U_{A,j;B,j}$ is replaced with $\mathcal{I} \otimes \mathcal{I}$.  This $\pr{H}(\lambda)$ can be obtained from $H(\lambda)$ by a unitary transformation that preserves the energy gap.
\label{lemma:identitygates}
\end{Lemma}

\begin{proof}
We prove this lemma by exhibiting the unitary transformation
\begin{eqnarray*}
\lefteqn{\mathcal{U}  =  \sum_{i_1=o_1-1}^{f_1} \dots \sum_{i_M=o_M-1}^{f_M} \otimes_{A=1}^M C_{A,i_A}^\dagger V(i_1,\dots,i_M) \otimes _{B=1}^M C_{B,i_B}} \\
&= & \sum_{i_1=o_1-1}^{f_1} \dots \sum_{i_M=o_M-1}^{f_M}\sum_{b_1=0}^{1} \dots \sum_{b_M=0}^{1}\sum_{b^\prime_1=0}^{1} \dots \sum_{b^\prime_M=0}^{1} \\
&& \hspace{-0.2in}V(i_1,\dots,i_M)_{b_1,\dots,b_M;b^\prime_1,\dots,b^\prime_M}\otimes_{A=1}^M c_{A,i_A,b_A}^\dagger\otimes _{B=1}^M c_{B,i_B,b^\prime_B}
\end{eqnarray*}
Here, $V(i_1,\dots,i_M)$ is a $2^M$ by $2^M$ unitary matrix that acts on the $C_{B,i_B}$ column-vectors, leaving the qubit indices $B$ and time steps $i_B$ fixed but changing the bit values $b_B$.

We show that $\mathcal{U}$ is unitary in the basis described in Def. \ref{def:GSQC}: eigenstates of $\sum_{i=o_A-1}^{f_A} C_{A,i}^{\dagger}C_{A,i}$, with $+1$ eigenvalue, for $A=1,\dots,M$.  To demonstrate unitarity, note that $c_{B,i_B,b^\prime_B} \otimes_{A=1}^M c_{A,j_A,\beta_A}\ket{\psi} =  0$ for all states $\ket{\psi}$  in the basis.  Indeed, $0=(\sum_{i=o_A-1}^{f_A} C_{A,i}^{\dagger}C_{A,i} - 1)\ket{\psi}$, implying $c_{A,j,b}c_{A,j^\prime,b^\prime}(\sum_{i=o_A-1}^{f_A} C_{A,i}^{\dagger}C_{A,i} - 1)\ket{\psi} = 0 = (\sum_{i=o_A-1}^{f_A} C_{A,i}^{\dagger}C_{A,i} +1) c_{A,j,b}c_{A,j^\prime,b^\prime}\ket{\psi}$.  Since $(\sum_{i=o_A-1}^{f_A} C_{A,i}^{\dagger}C_{A,i} +1)$ is a positive operator, it follows that $c_{A,j,b}c_{A,j^\prime,b^\prime}\ket{\psi} = 0$ and thus $c_{B,i_B,b^\prime_B} \otimes_{A=1}^M c_{A,j_A,\beta_A}\ket{\psi} =  0$.  With this fact in mind, compute
\begin{eqnarray*}
\lefteqn{\mathcal{U}^\dagger \mathcal{U}\ket{\psi}  = \sum_{i_1=o_1-1}^{f_1} \dots  \sum_{\beta^\prime_M=0}^{1} V(i_1,\dots,i_M)^{\dagger}_{b_1,\dots,b_M;b^\prime_1,\dots,b^\prime_M}} \\
&& V(j_1,\dots,j_M)_{\beta^\prime_1,\dots,\beta^\prime_M;\beta_1,\dots,\beta_M} \otimes_{A=1}^M c_{A,i_A,b_A}^\dagger \\
&& \otimes _{B=1}^M (\delta_{i_B,j_B}\delta_{b^\prime_B,\beta^\prime_B} - c^\dagger_{B,j_B,\beta^\prime_B}c_{B,i_B,b^\prime_B})\\
&&\otimes_{A=1}^M c_{A,j_A,\beta_A}\ket{\psi}
\end{eqnarray*}
The only contribution from $\otimes _{B=1}^M (\delta_{i_B,j_B}\delta_{b^\prime_B,\beta^\prime_B} - c^\dagger_{B,j_B,\beta^\prime_B}c_{B,i_B,b^\prime_B})$ comes from $\otimes _{B=1}^M \delta_{i_B,j_B}\delta_{b^\prime_B,\beta^\prime_B}$ since $c_{B,i_B,b^\prime_B} \otimes_{A=1}^M c_{A,j_A,\beta_A}\ket{\psi} =  0$.  Thus, 
\begin{eqnarray*}
\lefteqn{\mathcal{U}^\dagger\mathcal{U} \ket{\psi}  = \sum_{i_1=o_1-1}^{f_1} \dots \sum_{b^\prime_M=0}^{1}V(i_1,\dots,i_M)^{\dagger} _{b_1,\dots,b_M;b^\prime_1,\dots,b^\prime_M}} \\
&& \hspace{-0.2in}V(i_1,\dots,i_M)_{b^\prime_1,\dots,b^\prime_M;\beta_1,\dots,\beta_M} \otimes_{A=1}^M c_{A,i_A,b_A}^\dagger c_{A,i_A,\beta_A}\ket{\psi} \\
& = &\sum_{i_1=o_1-1}^{f_1} \dots \sum_{i_M=o_M-1}^{f_M} \otimes_{A=1}^M C_{A,i_A}^\dagger C_{A,i_A}\ket{\psi}  = \ket{\psi}.
\end{eqnarray*}

If the gate model quantum algorithm contains only one qubit gates, then take $V(i_1,\dots,i_M) \equiv \otimes _{B=1}^M (U_{B,i_B}\cdots U_{B,1})$.  If a two-qubit gate $U_{A,j;B,j}$ appears instead of $U_{A,j}$ and $U_{B,j}$  in the gate model algorithm, replace  $U_{A,j}\otimes U_{B,j}$ in this definition with $U_{A,j;B,j}$  when $i_A, i_B \ge j$ and otherwise set $U_{A,j} = U_{B,j} = \mathcal{I}$.  Direct calculation shows that $\pr{H}(\lambda)\ket{\psi}=\mathcal{U}^\dagger H(\lambda)\mathcal{U}\ket{\psi}$.
\end{proof}

Note that $\pr{H}(\lambda)$ does not change the bit value of any qubit: it has no matrix element between $c_{A,i,b}$ and $c_{A,i,1-b}$ for any $A$, $i$, or $b$.  Fig. \ref{fig:apparatusidentity}, to be compared with Fig. \ref{fig:apparatus}, shows this simplification graphically.  Thus, $\pr{H}(\lambda)$ is block-diagonal in the basis of bit values, and we will find it convenient to study it block-by-block.  Note, in particular, that its ground state resides in the block in which all bits have value 0 and $\sum_{A=1}^{M} h^{o_A-1}_{A}(\mathrm{INIT})$ vanishes.

Our lower bound on the gap will make use of the following generic lemma about the sums of Hamiltonians \cite{Mizel07}.

\begin{Lemma}[Eigenvalue Bound for Sum of Hamiltonians]
Let $h$ be a Hamiltonian with a ground state space ${\mathcal S}_0$ of states of energy $e_0$ and a first excited energy $e_1$.  Let $\delta h$ be a positive semi-definite Hamiltonian with maximum eigenvalue $\| \delta h \|$.  Then a lower bound for the eigenvalues of $h+\delta h$ is
\[
\min_{\left|\psi_0\right\rangle \in {\mathcal S}_0} e_0 + \frac{(e_1-e_0) \left\langle\psi_0| \delta h |\psi_0\right\rangle}{(e_1-e_0) + \left\langle\psi_0| \delta h |\psi_0\right\rangle + \| \delta h \|}.
\]
\label{lemma:sumofhamiltonians}
\end{Lemma}
\begin{proof}
Without loss of generality, we can write the ground state of $h+\delta h$ in the form $a \left|\psi_0\right\rangle + b \left|\delta \psi \right\rangle$, where $\left|\psi_0\right\rangle$ is a normalized state in ${\mathcal S}_0$, $\left|\delta \psi \right\rangle$ is a normalized state orthogonal to ${\mathcal S}_0$,  and $|a|^2+|b|^2=1$.  Expressing $h + \delta h-e_0$ in the basis of $\left|\psi_0\right\rangle$, $\left|\delta \psi \right\rangle$ gives the $2 \times 2$ matrix
\[
 \left[\begin{array}{cc}
\left\langle \psi_0 |\delta h| \psi_0 \right\rangle & \left\langle \psi_0 |\delta h| \delta \psi \right\rangle \\ \left\langle \delta \psi |\delta h| \psi_0 \right\rangle &  \left\langle \delta \psi |h + \delta h -e_0| \delta \psi \right\rangle
\end{array}\right].
\]
If the eigenvalues of this matrix are $\varepsilon_0 \le \varepsilon_1$, then, since $h-e_0$ and $\delta h$ are positive semi-definite, we know that $\varepsilon_0$ and $\varepsilon_1$ are non-negative.  It follows that $\varepsilon_0 \ge \varepsilon_0 \varepsilon_1/(\varepsilon_0+\varepsilon_1)$, which is the ratio of the determinant to the trace. 

The determinant of this matrix is at least $\left\langle \delta \psi |h-e_0| \delta\psi \right\rangle\left\langle \psi_0 |\delta h| \psi_0 \right\rangle$ since $\left\langle \psi_0 |\delta h| \psi_0 \right\rangle \left\langle \delta \psi |\delta h| \delta \psi \right\rangle - |\left\langle \delta \psi |\delta h| \psi_0 \right\rangle|^2 \ge 0$.  Dividing by the trace of this matrix gives $
\left\langle \psi_0 |\delta h| \psi_0 \right\rangle/[1+(\left\langle \psi_0 |\delta h| \psi_0 \right\rangle+ \left\langle \delta \psi |\delta h| \delta \psi \right\rangle)/ \left\langle \delta \psi |h -e_0| \delta \psi \right\rangle]$.  This is minimized by setting $\left\langle \delta \psi |\delta h| \delta \psi \right\rangle$ to its maximum value $\|\delta h\|$ and $\left\langle \delta \psi |h -e_0| \delta \psi \right\rangle$ to its minimum value $e_1-e_0$, which gives 
\[
e_0 + \frac{(e_1-e_0) \left\langle\psi_0| \delta h |\psi_0\right\rangle}{(e_1-e_0) + \left\langle\psi_0| \delta h |\psi_0\right\rangle + \| \delta h \|}.
\]
We minimize over choices of $\left|\psi_0\right\rangle$ to obtain the bound that appears in the statement of the lemma.
\end{proof}

To lower bound the gap of $H(\lambda)$, we now cast the circuit of Def. \ref{def:GSQC} into a more specific form.  Our first case of interest is a one-dimensional configuration in which the qubits are positioned on a line and 2-qubit gates are performed between nearest neighbors.   It is reminiscent of the configuration considered in \cite{Breuckmann14} but permits computations with $N \gg M$ and includes 1-qubit gates.  A sketch of the configuration appears in Fig. \ref{fig:nanocrystalarray1dwith1qubitgates}.

\begin{definition}[One-dimensional configuration]

A GSQC Hamiltonian $H(\lambda)$ of Def. \ref{def:GSQC} has a one-dimensional configuration if it is derived from the following type of circuit.  The circuit acts on an odd number of qubits $M$ and is generated from a underlying algorithm of even depth $n \ge M-1$.  To the end of the circuit are added an extra $n+M-1$ time steps of 1-qubit identity gates acting on every qubit, bringing the depth to $2n+M-1$.  The 1-qubit gates and the 2-qubit gates are then segregated into separate time steps such that no 2-qubit gate acts at an odd time step.  At time step $2j+1$, for $j \in [1,\dots,(2n+M-1)]$, every qubit $A$ undergoes a 1-qubit gate $U_{A,2j+1}$.  At time step $2j+2$, for $j \in [1,\dots,(2n+M-1)]$, there are 2-qubit gates $U_{A,2j+2;B,2j+2}$ that follow a specific pattern consistent with a 1-dimensional layout of qubits.  In particular, when $j$ is odd, $A=2$ and $B=3$ undergo a 2-qubit gate, qubits $A=4$ and $B=5$ undergo a 2-qubit gate, qubits $A=6$ and $B=7$ undergo a 2-qubit gate, and so on.  Qubit 1 is the only qubit to undergo a 1-qubit gate $U_{1,2j+2} =\mathcal{I}$ at time step $2j+2$ when $j$ is odd.  When $j$ is even, the pattern of gates switches, so that qubits $A=1$ and $B=2$ undergo a 2-qubit gate, qubits $A=3$ and $B=4$ undergo a 2-qubit gate, qubits $A=5$ and $B=6$ undergo a 2-qubit gate, and so on.  Qubit $M$ is the only qubit to undergo a 1-qubit gate $U_{M,2j+2} = \mathcal{I}$ at time step $2j+2$ when $j$ is even.  The pattern of 2-qubit gates is summarized in Table \ref{table:2qbgatesin1d}.  The segregation into 1-qubit gate time steps and 2-qubit gate time steps doubles the depth of the circuit to $2(2n+M-1)$, from original time step $3$ to final time step $2(2n+M)$.

In addition to the gates above, every qubit undergoes a 1-qubit identity gate $U_{2(2n+M)+1} =\mathcal{I}$.  Each of the qubits with even $A$ begins with two additional 1-qubit identity gates $U_{A,1} = U_{A,2} = \mathcal{I}$, so that $o_A = 1$ and $f_A = 2(2n+M)+1$.  Each of the qubits with odd $A<M$ finishes with two additional 1-qubit identity gates $U_{A,2(2n+M)+2} = U_{A,2(2n+M)+3} = \mathcal{I}$ so that $o_A = 3$ and $f_A = 2(2n+M)+3$.  Qubit $A=M$ has $o_A =3$ and $f_A = 2(2n+M)+1$.   The net depth of the circuit is $N = 2(2n+M)+3 \ge 6M-1$.  Only identity gates act during the final $2n+2M =  (N-3)/2+M$ time steps of the circuit.
\label{def:1d}
\end{definition}
The Hamiltonian of Def. \ref{def:1d} is depicted in Fig. \ref{fig:nanocrystalarray1dwith1qubitgates} for the case $M=5$, $N=9$.  (The values $M=5$, $N=9$ are chosen to keep the figure a manageable size despite violating the condition $N\ge 6 M-1$.)  The figure is simplified, analogously to Fig. \ref{fig:apparatusidentity}, under the assumption that all gates are identity gates (Lem. \ref{lemma:identitygates}).

\begin{center}
\begin{table}
\begin{tabular}{|c|c|c|c|c|c|c|c|}
\hline
\multirow{2}{*}{$U_{A,4;B,4}$}&A & M & M-2 & M-4 & $\dots$ & 5 & 3 \\
\cline{2-8}
& B & M-1 & M-3 & M-5 & $\dots$ & 4 & 2 \\
\hline
\multirow{2}{*}{$U_{A,6;B,6}$}&A & 1 & 3 & 5 & $\dots$ & M-4 & M-2 \\
\cline{2-8}
& B & 2 & 4 & 6 & $\dots$ & M-3 & M-1 \\
\hline
\multirow{2}{*}{$U_{A,8;B,8}$}&A & M & M-2 & M-4 & $\dots$ & 5 & 3 \\
\cline{2-8}
& B & M-1 & M-3 & M-5 & $\dots$ & 4 & 2 \\
\hline
\multirow{2}{*}{$U_{A,10;B,10}$}&A & 1 & 3 & 5 & $\dots$ & M-4 & M-2 \\
\cline{2-8}
& B & 2 & 4 & 6 & $\dots$ & M-3 & M-1 \\
\hline
\multirow{2}{*}{$U_{A,12;B,12}$}&A & M & M-2 & M-4 & $\dots$ & 5 & 3 \\
\cline{2-8}
& B & M-1 & M-3 & M-5 & $\dots$ & 4 & 2 \\
\hline
\multirow{2}{*}{$U_{A,14;B,14}$}&A & 1 & 3 & 5 & $\dots$ & M-4 & M-2 \\
\cline{2-8}
& B & 2 & 4 & 6 & $\dots$ & M-3 & M-1 \\
\hline
\end{tabular}
\caption{Pattern of 2-qubit gates in a 1-dimensional qubit configuration.}
\label{table:2qbgatesin1d}
\end{table}
\end{center}

\begin{figure}
\includegraphics[width=2.5in]{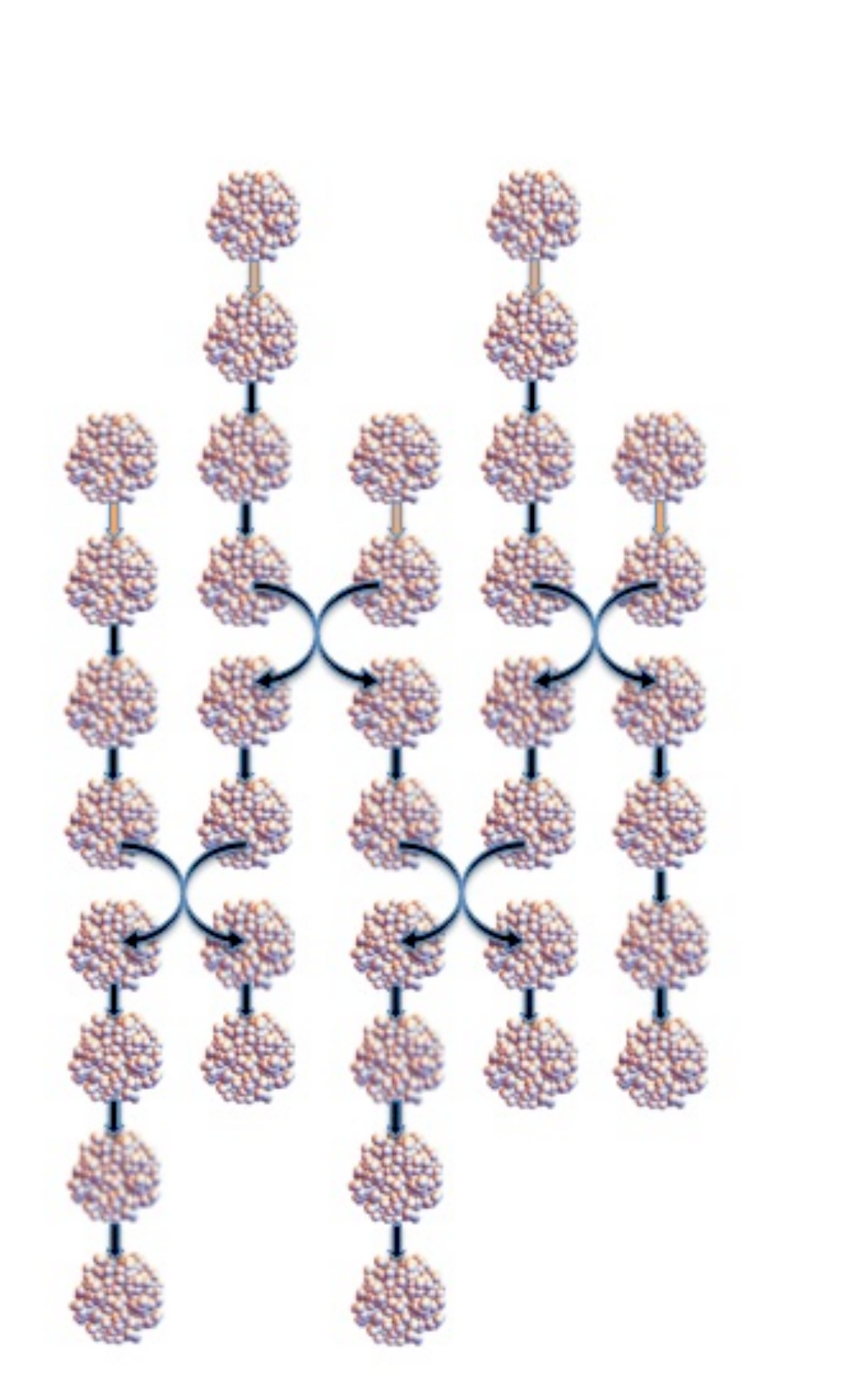}
\caption{Schematic physical apparatus for realizing the one-dimensional configuration of Def. \ref{def:1d}.  The Hamiltonian is assumed to have been simplified using Lem. \ref{lemma:identitygates}, so a single itinerant electron migrates in each array of 8 nanocrystals.  Black lines denote tunnelling matrix elements between orbitals housed on different nanocrystals.  Yellow lines denote $\lambda$-dependent tunnelling matrix elements.}
\label{fig:nanocrystalarray1dwith1qubitgates}
\end{figure}

\begin{theorem}[Gap of $H(\lambda)$]
Apply Lem. \ref{lemma:identitygates} to the Hamiltonian $H(\lambda)$ of Def. \ref{def:1d}, and let  $\left| \psi_0 (\lambda)\right\rangle$ be the resulting ground state.  Let $A$ be the integer satisfying $\frac{A-1}{M} \le  \lambda < \frac{A}{M} $ for $\lambda < 1$ and let $A=M$ for $\lambda =1$.  (Thus, $A = \lfloor M \lambda \rfloor+1$ for $\lambda <1$.)  Let $e_1(\lambda)$ be the lowest non-zero eigenvalue of $H(\lambda) -  h^{o_A}_A(\lambda_A \mathcal{I}) $ in the sub-block in which all bits have bit value zero.  Then $H(\lambda)$ has a gap $g(\lambda)$ bounded below by
\[
g_{min} = \min_{\lambda} \frac{1}{4}e_1(\lambda) \bra{ \psi_0 (\lambda)} c^\dagger_{A,o_A-1,0}c_{A,o_A-1,0}  \ket{ \psi_0 (\lambda)}.
\]
\label{thm:Hlambdagap}
\end{theorem}
\begin{proof}
Note that $H(\lambda)$ commutes with $C^\dagger_{B,o_B-1}C_{B,o_B-1}$ for all $B > A$, so it can be diagonalized into $2^{M-A}$ blocks with $C^\dagger_{B,o_B-1}C_{B,o_B-1}$ fixed at either 0 or 1 for each $B > A$.  To find the first excited state, we can limit our attention to the block of $H(\lambda)$ in which all qubits $B > A =  \lfloor M \lambda \rfloor +1$ are quiescent, with $C^\dagger_{B,o_B-1,0}C_{B,o_B-1,0}$ having eigenvalue $1$.

To see this, consider a block in which $C^\dagger_{B,o_B-1}C_{B,o_B-1}$ has eigenvalue $0$ for some particular $B > A+1$.  Note that $H(\lambda)$ contains the term $ \sum_{i=o_B}^{f_B} \mathcal{E}(1-\lambda_{B-1}^3) C^\dagger_{B,i}C_{B,i}$.  Since $B>A+1$, we have $\lambda_{B-1}=\lambda_B=0$, so this sum becomes $\sum_{i=o_B}^{f_A} \mathcal{E} C^\dagger_{B,i}C_{B,i}$.  This pushes the energies of all states in the block up to $\mathcal{E}$ or greater.

This argument needs a little adjustment when we consider the block in which $C^\dagger_{B,o_B-1,0}C_{B,o_B-1,0}$ has eigenvalue $1$ for all quiescent qubits $B > A+1$ but eigenvalue $0$ for a single itinerant qubit $B=A+1$.  The energy of every state in the block is at least $\mathcal{E}(1-\lambda_A^3)$ because of the contribution of $\sum_{i=o_B}^{f_B} \mathcal{E}(1-\lambda_{B-1}^3) C^\dagger_{B,i}C_{B,i} =  \mathcal{E}(1-\lambda_{A}^3)  \sum_{i=o_B}^{f_B}C^\dagger_{B,i}C_{B,i}  = \mathcal{E}(1-\lambda_{A}^3) $ to $H(\lambda)$.  Thus, if $\lambda_A \le 3/4$, every state in the block has energy $(1-(3/4)^3) \mathcal{E}$ or greater.  But what happens when the energy $\mathcal{E}(1-\lambda_A^3)$ becomes small as $\lambda_A$ approaches 1?  In the basis of Def. \ref{def:GSQC}, $\mathcal{E}(1-\lambda_{A}^3) \ge \mathcal{E}(1-\lambda_{A}^3)(C^\dagger_{A,o_A-1}C_{A,o_A-1}+C^\dagger_{A,o_A}C_{A,o_A})$.  When $3/4 < \lambda_A < 1$, $\mathcal{E}(1-\lambda_{A}^3)(C^\dagger_{A,o_A-1}C_{A,o_A-1}+C^\dagger_{A,o_A}C_{A,o_A}) + h_A^{o_A}(\lambda_A \mathcal{I}) > h_A^{ o_A}(\mathcal{I}) = (h_A^{ o_A}(\lambda_A\mathcal{I}) + \sum_{i=o_{A+1}}^{f_{A+1}} \mathcal{E}(1-\lambda_{A}^3) C^\dagger_{A+1,i}C_{A+1,i}|_{\lambda_A=1}$,  so it follows that the smallest energy in the block when $3/4<\lambda_A<1$ is lower bounded by the smallest energy in the block when $\lambda_A =1$.  Thus, we do not have to consider the block when $3/4 < \lambda_A <1$ as long as we study the block when $\lambda_A=1$.  Note, however, that $\lambda_A$ reaches 1 precisely when $\lambda$ reaches the value $\lambda^\prime = A/M$, which triggers the value of $A$ to increment to $A^\prime = \lfloor M \lambda^\prime \rfloor+1 = A+1$.   Because $A$ has incremented to $A^\prime$, the itinerant qubit $B=A+1$ is now qubit $B=A^\prime$, and it follows that our block at $\lambda_A=1$ can be studied by studying the block of $H(\lambda^\prime)$ in which all qubits $B > A^\prime$ are quiescent with $C^\dagger_{B,o_B-1,0}C_{B,o_B-1,0}$ having eigenvalue $1$.

The conclusion is that we can focus, for $0<\lambda<1$, on the block of $H(\lambda)$ in which all qubits $B > A$ are quiescent with $C^\dagger_{B,o_B-1,0}C_{B,o_B-1,0}$ having eigenvalue $1$.  The application of Lem. \ref{lemma:identitygates} allows us to treat all unitary gates in $H(\lambda)$ as identity gates so our block can be further diagonalized into $2^M$ sub-blocks in the basis of bit values.  All of its $2^M$ sub-blocks are identical except for the contribution of the positive definite initialization Hamiltonian $\sum_{A=1}^{M} h^{o_A-1}_{A}(\mathrm{INIT})$.  The (normalized) ground state $\left| \psi_0 (\lambda)\right \rangle$ of Thm. \ref{thm:groundstate} resides in the sub-block in which the initialization Hamiltonian vanishes identically -- the sub-block in which all the bit values are zero.  In part (a) of the following, we apply Lem. \ref{lemma:sumofhamiltonians} to bound the first excited state of $H(\lambda)$ in this sub-block.  In part (b), we apply Lem. \ref{lemma:sumofhamiltonians} to bound the ground state energy of $H(\lambda)$ in the other $2^M-1$ sub-blocks.

(a) Set $\delta h =  h^{o_A}_A(\lambda_A \mathcal{I})$, $h = H(\lambda) -  \delta h$, and focus on the sub-block in which all bits have bit value zero.  The sub-block contains the ground state of $H(\lambda)$, $\left| \psi_0 (\lambda)\right\rangle$.  Because of the subtraction of $\delta h$, $h$ has an additional zero-energy, orthogonal ground state besides $\left| \psi_0 (\lambda)\right\rangle$; we assign it the label $\left| \psi_0^\perp (\lambda)\right\rangle$.  Lem. \ref{lemma:sumofhamiltonians} gives the following lower bound on the lowest energy of $H(\lambda)$  in the sub-block, in the space orthogonal to $\left| \psi_0 (\lambda)\right\rangle$:
\[
\frac{e_1(\lambda) \bra{\psi_0^\perp(\lambda)} h^{o_A}_A(\lambda_A \mathcal{I}) \ket{\psi_0^\perp(\lambda)}}{e_1(\lambda) + \bra{\psi_0^\perp(\lambda)} h^{o_A}_A(\lambda_A \mathcal{I}) \ket{\psi_0^\perp(\lambda)} + \lambda_A^2+1 }.
\]
To evaluate $\bra{\psi_0^\perp(\lambda)} h^{o_A}_A(\lambda_A \mathcal{I}) \ket{\psi_0^\perp(\lambda)}$, note that $\ket{\psi_0(\lambda)} = \kappa_0\left[ \lambda_A(\lambda) + (1-\lambda_A(\lambda)) c^\dagger_{A,o_A-1,0}c_{A,o_A-1,0}\right] \ket{\psi_0(A/M)}$
since $\ket{\psi_0(A/M)}$ corresponds to the case $\lambda=A/M$, $\lambda_A =  1$.  The orthogonal ground state is $\ket{\psi_0^\perp(\lambda)} = \kappa_0^\perp \left[ 1 + (\mu-1) c^\dagger_{A,o_A-1,0}c_{A,o_A-1,0}\right] \ket{\psi_0(A/M)}$.  Orthogonality implies $\mu = \lambda_A(1 - 1/\bra{\psi_0(A/M)}c^\dagger_{A,o_A-1,0}c_{A,o_A-1,0}\ket{\psi_0(A/M)})$ while the normalization constants are $\kappa_0 = 1/\sqrt{\lambda_A^2+(1-\lambda_A^2) \bra{\psi_0(A/M)}c^\dagger_{A,o_A-1,0}c_{A,o_A-1,0}\ket{\psi_0(A/M)}}$ and $\kappa_0^\perp = 1/\sqrt{1+(\mu^2-1) \bra{\psi_0(A/M)}c^\dagger_{A,o_A-1,0}c_{A,o_A-1,0}\ket{\psi_0(A/M)}}$.

We find that $\bra{\psi_0^\perp(\lambda)} h^{o_A}_A(\lambda_A \mathcal{I}) \ket{\psi_0^\perp(\lambda)} = \lambda_A ^2+ \bra{\psi_0(A/M)}c^\dagger_{A,o_A-1,0}c_{A,o_A-1,0}\ket{\psi_0(A/M)}/(1-\bra{\psi_0(A/M)}c^\dagger_{A,o_A-1,0}c_{A,o_A-1,0}\ket{\psi_0(A/M)})$.  Inserting these expressions into our lower bound, and simplifying under the assumption $e_1(\lambda)<1$, we find a bound $e_1(\lambda)\bra{\psi_0(A/M)}c^\dagger_{A,o_A-1,0}c_{A,o_A-1,0}\ket{\psi_0(A/M)}/2$. 

(b) Consider the other $2^M-1$ sub-blocks in which 1 or more qubits have bit value 1. The initialization Hamiltonian makes a contribution $\mathcal{E} c^\dagger_{B,o_B-1,1}c_{B,o_B-1,1}$ to the sub-block for each qubit with bit value 1.  This positive semi-definite contribution is smallest in sub-blocks for which a single qubit $B$ has bit value equal to 1; if we go to a sub-block that includes another positive semi-definite term, we can only push the energy of its ground state higher.  Thus, we are justified in focusing on the case in which the initialization Hamiltonian contributes $\mathcal{E} c^\dagger_{B,o_B-1,1}c_{B,o_B-1,1}$ to the block for a single qubit $B$.

Set $\delta h =  h^{o_A}_A(\lambda_A \mathcal{I}) +h^{o_B-1}_{B}(\mathrm{INIT})$ and $h = H(\lambda) -  \delta h$.  Lem. \ref{lemma:sumofhamiltonians} gives the following lower bound on the lowest energy of $H(\lambda)$ in this block of interest:
\[
\min_{\ket{\psi} = \cos \theta \ket{\psi_0(\lambda)} +\sin \theta \ket{\psi_0^\perp(\lambda)}} \frac{e_1(\lambda) \bra{\psi} \delta h \ket{\psi}}{e_1(\lambda) +  \bra{\psi} \delta h \ket{\psi} +  \lambda_A^2+2 }.
\]
We minimize over $\theta$ to find $\bra{\psi} \delta h \ket{\psi}$ and obtain the smallest result when $A=B$.  Assuming that $e_1(\lambda)<1$, we derive a lower bound  $e_1(\lambda) \bra{\psi_0(A/M)}c^\dagger_{A,o_A-1,0}c_{A,o_A-1,0}\ket{\psi_0(A/M)}/4$, which is smaller than the bound in (a).
\end{proof}

We now bound $g_{min}$ of Thm. \ref{thm:Hlambdagap} by bounding the quantities $e_1(\lambda)$ and $\bra{ \psi_0 (\lambda)} c^\dagger_{A,o_A-1,0}c_{A,o_A-1,0}  \ket{ \psi_0 (\lambda)}$. 

\begin{Lemma}
%[Bounds on  $\langle  {\psi_0}^{\text{trunc}} (\bar{\cal A},\{0,0,\dots,0\}) | \sum_{\bar{A}\in \bar{\cal A}} c^\dagger_{\bar{A},o_{\bar{A}},0}c_{\bar{A},o_{\bar{A}},0} | {\psi_0}^{\text{trunc}}(\bar{\cal A},\{0,0,\dots,0\})\rangle$ and $\text{min}_{A} \langle  \psi_0 (\lambda)| c^\dagger_{A,o_A-1,0}c_{A,o_A-1,0} | \psi_0 (\lambda)\rangle$]
Let $H(\lambda)$ be a GSQC Hamiltonian (Def. \ref{def:GSQC}) of a circuit that conforms with the one-dimensional configuration of Def. \ref{def:1d}.  Then  $\bra{ \psi_0 (\lambda)} c^\dagger_{A,o_A-1,0}c_{A,o_A-1,0}  \ket{ \psi_0 (\lambda)} \ge 1/N$, where $A$ is the integer satisfying $\frac{A-1}{M} \le  \lambda < \frac{A}{M} $ (with $A=M$ if $\lambda =1$).
\label{lemma:1overN}
\end{Lemma}
\begin{proof}
Because it has zero energy, the ground state must have the same probability density on all time-valid choices of $\{i_1,\dots,i_M\}$.  The claim of the lemma therefore comports with the naive expectation that each qubit occupies each time step with probability at least $\sim 1/N$.  A careful proof follows.

Suppose that $A < M$.  Then, in the ground state $\ket{ \psi_0 (\lambda)}$, qubits $A+1,\dots,M$ are localized at time steps $o_{A+1}-1,\dots,o_M-1$ respectively.    Because qubit $A+1$ resides inertly at position $o_{A+1}-1$, there are only $4$ valid time step values $o_A - 1,\dots,o_A+2$ available to qubit $A$ before it finds itself unable to traverse a 2-qubit gate with qubit $A+1$.  For each valid time step, there are $4$ valid time step values available to qubit $A-1$.  For each of those time step values, there are $4$ valid time step values available to qubit $A-2$.  This pattern continues down to qubit $1$.  Introducing a normalization constant $\kappa_0$, we have $\bra{ \psi_0 (\lambda)} c^\dagger_{A,o_A-1,0}c_{A,o_A-1,0}  \ket{ \psi_0 (\lambda)} = 4^{A-1} \kappa_0^2$, while $\bra{ \psi_0 (\lambda)} c^\dagger_{A,o_A,0}c_{A,o_A,0}  \ket{ \psi_0 (\lambda)}=\bra{ \psi_0 (\lambda)} c^\dagger_{A,o_A+1,0}c_{A,o_A+1,0}  \ket{ \psi_0 (\lambda)}=\bra{ \psi_0 (\lambda)} c^\dagger_{A,o_A+2,0}c_{A,o_A+2,0}  \ket{ \psi_0 (\lambda)} = \lambda_A^2 4^{A-1} \kappa_0^2$.  Since the sum of these matrix elements must be 1, we can solve for $\kappa_0$ and conclude that $\bra{ \psi_0 (\lambda)} c^\dagger_{A,o_A-1,0}c_{A,o_A-1,0}  \ket{ \psi_0 (\lambda)} = 1/(1+3\lambda_A^2) \ge 1/4 > 1/N$.

Next suppose that $A=M$.  In this case, in the ground state $\ket{ \psi_0 (\lambda)}$, qubit $A$ visits all $N-3$ time steps from $o_M-1=2$ to $f_M = N-2$.  Again, $\bra{ \psi_0 (\lambda)} c^\dagger_{A,o_A-1,0}c_{A,o_A-1,0}  \ket{ \psi_0 (\lambda)} = 4^{A-1} \kappa_0^2$, while $\bra{ \psi_0 (\lambda)} c^\dagger_{A,o_A,0}c_{A,o_A,0}  \ket{ \psi_0 (\lambda)} = \dots = \bra{ \psi_0 (\lambda)} c^\dagger_{A,f_A,0}c_{A,f_A,0}  \ket{ \psi_0 (\lambda)} = \lambda_A^2 4^{A-1} \kappa_0^2$.  We conclude that $\bra{ \psi_0 (\lambda)} c^\dagger_{A,o_A-1,0}c_{A,o_A-1,0}  \ket{ \psi_0 (\lambda)} = 1/(1+(N-4) \lambda_A^2) \ge 1/(N-3) > 1/N$. 
\end{proof}

To bound $g_{min}$ using Thm. \ref{thm:Hlambdagap}, our remaining task is to lower bound the quantity $e_1(\lambda)$.  We assume that Lem. \ref{lemma:identitygates} has been applied to $H(\lambda)$, and our attention is focused on the sub-block of $H(\lambda) -  h^{o_A}_A(\lambda_A \mathcal{I})$ with all zero bit values.  The operator $c^\dagger_{A,o_A-1,0}c_{A,o_A-1,0}$ commutes with this Hamiltonian.  We will argue below that that the smallest value of $e_1(\lambda)$ occurs when $\lambda = 1$ and $c^\dagger_{A,o_A-1,0}c_{A,o_A-1,0} = c^\dagger_{M,o_M-1,0}c_{M,o_M-1,0}$ has eigenvalue 0.  To analyze this case, it is useful to associate a graph with $H(1) -  h^{o_M}_M(\mathcal{I})$ as follows.  The vertices of the graph lie in one-to-one correspondence with time-valid points ${\bf i} = (i_1,i_2,\dots,i_M )$ for which $i_M \ge o_M$ and all penalty Hamiltonians of the form (\ref{eq:penalty}) vanish \endnote{To compute $e_1(\lambda)$, we need not consider those very high-energy states which the penalty Hamiltonian does not annihilate.}.  The set of vertices is called ${\cal V}$.  There is an edge connecting any two vertices between which there is a non-zero Hamiltonian matrix element.  All edges have equal weight because every such matrix element has the same value $-\mathcal{E}$ for $H(1) -  h^{o_M}_M( \mathcal{I})$.  Two vertices ${\bf i} $ and ${\bf i} ^\prime$ connected by an edge are called adjacent; we write ${\bf i}  \sim {\bf i} ^\prime$ to express adjacency.

Any state vector $\left| \psi \right\rangle$ in the Hilbert space of interest can be described in terms of the function $\psi(i_1,\dots,i_M) = \left\langle \mathrm{vac} \right| \otimes_{A=1}^M c^\dagger_{A,i_A,0} \left| \psi \right\rangle$.  If $\left| \psi \right\rangle$ is annihilated by the penalty Hamiltonian (\ref{eq:penalty}), then $\psi(i_1,\dots,i_M)$ will be supported on the vertices of the graph.  The expectation value of $H(1) -  h^{o_M}_M( \mathcal{I})$ in the state is given by

\begin{equation}
E_{\psi} =  \left\langle \psi \right| H(1) -  h^{o_M}_M( \mathcal{I}) \left| \psi \right\rangle = \mathcal{E}\sum_{{\bf i}  \sim {\bf i} ^\prime} \left| \psi({\bf i} ) - \psi({\bf i} ^\prime)\right|^2/\sum_{{\bf i}  \in {\cal V}} \left| \psi({\bf i} )\right|^2.
\label{eq:E}
\end{equation}
The state with smallest $E_{\psi}$ is $\psi_0({\bf i} ) = 1/\sqrt{\|{\cal V}\|}$, a constant function with $E_{\psi_0}= 0$.  We are interested in lower bounding the smallest energy $E_{\psi}$ of any state orthogonal to $\psi_0({\bf i} )$ .  The following lemma, closely related to a result in \cite{Mohar91}, is our workhorse.

\begin{Lemma}[Paths]

Consider a graph with vertices ${\cal V}$.  Let $\|{\cal V}\|$ be even.  Let $\phi({\bf i})$ be a real valued function, orthogonal to $\psi_0({\bf i})$ so that $\sum_{{\bf i} \in {\cal V}} \phi({\bf i}) = 0$.  We wish to lower bound $E_{\phi}$, the expectation value (\ref{eq:E}), for $\phi({\bf i})$.

Suppose that there is a positive integer $T$ large enough to permit the existence of a path $P({\bf i},t): {\cal V} \otimes \left[0,1,\dots , T\right] \rightarrow {\cal V}$ with the following properties:
\begin{enumerate}
\item[(I)] The path must begin as the identity: $P({\bf i},0) = {\bf i}$.
\item[(II)] The path must advance by at most an edge at a time: $P({\bf i},t) = P({\bf i},t+1)$ or $P({\bf i},t)\sim P({\bf i},t+1)$.
\item[(III)] The path must produce a one-to-one mapping to vertices at the final time $T$: $P({\bf i},T)=P({\bf i}^\prime,T) \Leftrightarrow {\bf i} = {\bf i}^\prime$.
\item[(IV)] Set $\Phi$ equal to the median value of $\phi({\bf i})$ and let $\psi({\bf i}) = \phi({\bf i}) - \Phi$.  Because $\Phi$ is the median value of $\phi({\bf i})$, we can partition $ {\cal V}$ into a set ${\cal P}$ in which $\psi({\bf i})$ is non-negative and a disjoint set ${\cal N}$ in which $\psi({\bf i})$ is non-positive, with $\|{\cal P}\| = \|{\cal N}\| = \|{\cal V}\|/2$. The path must allow a second partition of ${\cal V}$ into pairs of adjacent vertices ${\bf i} \sim {\bf i}^\prime$ such that either $P({\bf i},T) \in {\cal P}  \land P({\bf i}^\prime,T) \in {\cal N}$ or $P({\bf i},T) \in {\cal N} \land P({\bf i}^\prime,T) \in {\cal P}$.
\end{enumerate}

Suppose that $B = max_{{\bf i} \sim {\bf i}^\prime} B_{{\bf i},{\bf i}^\prime}$, where, for any ${\bf i} \sim {\bf i}^\prime$, $B_{{\bf i},{\bf i}^\prime}$ gives the number of instances in which $P$ employs the edge between ${\bf i}$ and ${\bf i}^\prime$.  Explicitly, $B_{{\bf i},{\bf i}^\prime} = \sum_{{\bf i}^{\prime\prime} \in {\cal V}} \sum_{t=0}^{T-1} \left( \delta_{P({\bf i}^{\prime\prime},t),{\bf i}} \delta_{P({\bf i}^{\prime\prime},t+1),{\bf i}^\prime}+ \delta_{P({\bf i}^{\prime\prime},t+1),{\bf i}} \delta_{P({\bf i}^{\prime\prime},t),{\bf i}^\prime}\right)$ where $\delta_{{\bf i},{\bf i}^{\prime\prime}} = 1 \Leftrightarrow {\bf i} = {\bf i}^{\prime\prime}$.  Then, $E_{\phi} \ge 2\mathcal{E}/(2T+1)(2B+1)$.
\label{lemma:path}
\end{Lemma}

\begin{proof}

Note that 
\[
E_{\psi} = \frac{\mathcal{E}\sum_{{\bf i} \sim {\bf i}^\prime} \left| \psi({\bf i}) - \psi({\bf i}^\prime)\right|^2}{\sum_{{\bf i} \in {\cal V}} \left| \psi({\bf i})\right|^2} = \frac{\mathcal{E}\sum_{{\bf i} \sim {\bf i}^\prime} \left| \phi({\bf i}) - \phi({\bf i}^\prime)\right|^2}{\sum_{{\bf i} \in {\cal V}} \left| \phi({\bf i})\right|^2 - \Phi (\phi({\bf i}) +\phi^*({\bf i})) +  \left|\Phi\right|^2} = \frac{\mathcal{E}\sum_{{\bf i} \sim {\bf i}^\prime} \left| \phi({\bf i}) - \phi({\bf i}^\prime)\right|^2}{\sum_{{\bf i} \in {\cal V}} \left| \phi({\bf i})\right|^2 +  \left|\Phi\right|^2} \le E_{\phi}.
\]
Thus,a lower bound for $E_{\psi}$ also lower bounds $E_{\phi}$.

Note that (III) implies
\[
2 \sum_{{\bf i} \in {\cal V}} \left| \psi({\bf i})\right|^2 =   2 \sum_{{\bf i} \in {\cal V}} \left| \psi(P({\bf i},T))\right|^2  =   \sum_{{\bf i} \in {\cal V}} \left| \psi(P({\bf i},T))\right|^2 +  \left| \psi(P({\bf i}^\prime({\bf i}) ,T))\right|^2
\]
where ${\bf i}^\prime({\bf i})$ is the partner of ${\bf i}$ in the second partition described in (IV).  It follows from (IV) that $\psi(P({\bf i},T)) \psi(P({\bf i}^\prime({\bf i}),T)) \le 0$, so
\[
2 \sum_{{\bf i} \in {\cal V}} \left| \psi({\bf i})\right|^2 \le \sum_{{\bf i} \in {\cal V}} \left| \psi(P({\bf i},T)) -  \psi(P({\bf i}^\prime({\bf i}),T))\right|^2.
\]
Each term on the right hand side expands into a telescoping sum
\begin{eqnarray*}
\lefteqn{\left| \psi(P({\bf i},T)) -  \psi(P({\bf i}^\prime,T))\right|^2}\\
&= &\left| \psi(P({\bf i},T)) -  \psi(P({\bf i},T-1)) + \dots +   \psi(P({\bf i},1)) -  \psi(P({\bf i},0)) \right. \\
& & \hspace{1.0in} + \psi(P({\bf i},0)) -  \psi(P({\bf i}^\prime,0)) \\
& & + \left. \psi(P({\bf i}^\prime,0)) -  \psi(P({\bf i}^\prime,1)) + \dots + \psi(P({\bf i}^\prime,T-1)) -  \psi(P({\bf i}^\prime,T)) \right|^2.\\
& \le & (2T+1) \left( \left| \psi(P({\bf i},T)) -  \psi(P({\bf i},T-1))\right|^2 + \dots +  \left| \psi(P({\bf i},1)) -  \psi(P({\bf i},0))\right|^2 \right.\\
& & \hspace{1.0in} +  \left| \psi(P({\bf i},0)) -  \psi(P({\bf i}^\prime,0))\right|^2 \\
& & \left. + \left| \psi(P({\bf i}^\prime,T)) -  \psi(P({\bf i}^\prime,T-1))\right|^2 + \dots +  \left| \psi(P({\bf i}^\prime,1)) -  \psi(P({\bf i}^\prime,0))\right|^2 \right)
\end{eqnarray*}
where we have suppressed the dependence of ${\bf i}^\prime({\bf i})$ to lighten the notation.

Using (I), we can write $\left| \psi(P({\bf i},0)) -  \psi(P({\bf i}^\prime,0))\right|^2 = \left| \psi({\bf i}) - \psi({\bf i}^\prime)\right|^2$, so
\begin{eqnarray*}
2 \sum_{{\bf i} \in {\cal V}} \left| \psi({\bf i})\right|^2 &  \le & (2T+1) \left( \sum_{{\bf i} \in {\cal V}} \sum_{t=0}^{T-1} \left| \psi(P({\bf i},t)) -  \psi(P({\bf i},t+1))\right|^2  \right. \\
&& \left. \hspace{0.75in} +  \sum_{{\bf i} \in {\cal V}} \sum_{t=0}^{T-1} \left| \psi(P({\bf i}^{\prime}({\bf i}),t)) -  \psi(P({\bf i}^{\prime}({\bf i}),t+1))\right|^2 + \sum_{{\bf i} \in {\cal V}} \left| \psi({\bf i}) - \psi({\bf i}^\prime)\right|^2 \right) \\
& = & (2T+1) \left(2 \sum_{{\bf i}^{\prime\prime} \in {\cal V}} \sum_{t=0}^{T-1} \left| \psi(P({\bf i}^{\prime\prime},t)) -  \psi(P({\bf i}^{\prime\prime},t+1))\right|^2 + \sum_{{\bf i} \in {\cal V}} \left| \psi({\bf i}) - \psi({\bf i}^\prime)\right|^2 \right).
\end{eqnarray*}
Rewriting the first sum and adding terms to the second sum, we use (II) to obtain
\begin{eqnarray*}
2 \sum_{{\bf i} \in {\cal V}} \left| \psi({\bf i})\right|^2  & \le & (2T+1) \left(2 \sum_{{\bf i}^{\prime\prime} \in {\cal V}} \sum_{t=0}^{T-1} \sum_{{\bf i} \sim  {\bf i}^\prime} \left(\delta_{P({\bf i}^{\prime\prime},t),{\bf i}} \delta_{P({\bf i}^{\prime\prime},t+1),{\bf i}^\prime}+ \delta_{P({\bf i}^{\prime\prime},t+1),{\bf i}} \delta_{P({\bf i}^{\prime\prime},t),{\bf i}^\prime}\right) \left| \psi({\bf i}) -  \psi({\bf i}^\prime)\right|^2 \right. \\
&&\left. \hspace{0.75in} + \sum_{{\bf i} \sim {\bf i}^\prime} \left| \psi({\bf i}) - \psi({\bf i}^\prime)\right|^2 \right)  \\
& = & (2T+1)  \sum_{{\bf i} \sim  {\bf i}^\prime}  (2B_{{\bf i},{\bf i}^\prime}+1) \left| \psi({\bf i}) -  \psi({\bf i}^\prime)\right|^2\\
& \le & (2T+1) (2B+1)   \sum_{{\bf i} \sim  {\bf i}^\prime} \left| \psi({\bf i}) -  \psi({\bf i}^\prime)\right|^2.
\end{eqnarray*}
Dividing the initial expression and the final expression by $(2T+1)(2B+1) \sum_{{\bf i} \in {\cal V}} \left| \psi({\bf i})\right|^2$ yields the desired inequality.
\end{proof}

We have written Lem. \ref{lemma:path} so that it only applies to graphs that have $\|{\cal V}\|$ even, and moreover only to graphs that can be partitioned into pairs of adjacent vertices.  These assumptions are almost certainly inessential, but they simplify the proof of the Lemma.

In the Appendix, we apply this lemma to examples of increasing complexity. The proof for our graph of interest, the one-dimensional configuration (Def. \ref{def:1d}) depicted in Fig. \ref{fig:nanocrystalarray1dwith1qubitgates}, references these examples.

\begin{theorem}
For a circuit with a one-dimensional configuration of qubits (Def. \ref{def:1d}) the minimum gap is bounded by
\[
g_{min} \ge  \mathcal{E}/(12M(N+2M-4)+1)(4(N+4M-4)+1)(2N).
\]
\label{thm:one-dimensional}
\end{theorem}

\begin{proof}
We first argue that $e_1(\lambda)$ is smallest when $\lambda = 1$.  If $\lambda<(M-1)/M$, note that $e_1(\lambda)$ is associated with an excited state of $H(\lambda) -  h^{o_A}_A(\lambda_A \mathcal{I})$ for which all qubits $B > A$ are fixed at $c^\dagger_{B,o_B-1,0}c_{B,o_B-1,0}$.  In this state, there are only $4$ valid time step values $o_A - 1,\dots,o_A+2$ available to qubit $A$ before it finds itself unable to traverse the 2-qubit gate with qubit $A+1$.  Otherwise, as far as the energy of excited this state is concerned, we may as well have a computer with $A$ qubits rather than $M$ qubits: qubits $B>A$ play no role in determining the energy of the excited state.  We can lower bound the energy of the excited state using the same technique that we use below to lower bound $e_1(1)$, except that $M$ is replaced with $A$ and $f_A$ is replaced with $o_A+2$.  Since these replacements increase the lower bound, we can focus on $e_1(1)$ and need not consider the case $\lambda < (M-1)/M$ further.  And, noting that $H(\lambda) -  h^{o_M}_M(\lambda_M \mathcal{I})$ is independent of $\lambda$ for $(M-1)/M \le \lambda \le 1$, we can just fix $\lambda =1$.

Once we fix $\lambda =1$, we note that $H(1) -  h^{o_M}_M( \mathcal{I})$ commutes with $c^\dagger_{M,o_M-1,0}c_{M,o_M-1,0}$.  If we work among states with $c^\dagger_{M,o_M-1,0}c_{M,o_M-1,0}$ equal to $1$, then qubit $M-1$ is confined to a 4 time-steps $o_M - 1,\dots,o_M+2$.  We can lower bound the energy of the excited state using the same technique that we use below, except that $M$ is replaced with $M-1$ and $f_A$ is replaced with $o_A+2$.  Thus, we instead work among states with $c^\dagger_{M,o_M-1,0}c_{M,o_M-1,0}$ equal to $0$.

Now, we bound $e_1(1)$ by closely following the strategy of Ex. 4 of the Appendix, with a small modification to accommodate the extra time steps with 1-qubit gates.  The vertices of our graph can be labelled $(k_1,\dots,k_M)$.   The one-dimensional configuration (Def. \ref{def:1d}) has been chosen so that, at a fixed value of $k_{\alpha}$,  $k_{\alpha-1}$ is constrained to 1 of 4 consecutive values.  We relabel these 4 values using new variables $i_{\alpha-1} \in \{0,1,2,3\}$ that lie in correspondence with the 4 consecutive values of $k_{\alpha-1}$.  For $\alpha = M$, $i_M \equiv k_M \in \{3, \dots, N-2 \}$.

The graph has the following edges in terms of the relabelled coordinates $i_\alpha$.  The 2-qubit gates provide edges between $(i_1,\dots,i_{\alpha-1}=3,i_\alpha=1,\dots,i_M)$ and $(i_1,\dots,i_{\alpha-1}-3=0,i_\alpha+1=2,\dots,i_M)$.  The 1-qubit gates provide edges between $(i_1,\dots,i_{\alpha}=0,\dots,i_M)$ and $(i_1,\dots,i_{\alpha}+1=1,\dots,i_M)$ and between $(i_1,\dots,i_{\alpha}=2,\dots,i_M)$ and $(i_1,\dots,i_{\alpha}+1=3,\dots,i_M)$.  There are additional 1-qubit edges between $(i_1,\dots,i_{\alpha},\dots,i_M)$ and $(i_1,\dots,i_{\alpha}+1,\dots,i_M)$  for $\alpha=1$ or $\alpha = M$ and at the beginning and end of the computation for other $i_\alpha$.

We  define $j_M(i_1,\dots,i_M)$ as in Ex. 3 and 4, ensuring $\sum_{i_1 = 0}^{3} \dots \sum_{i_{M-1} = 0}^{3} sign\,\, \psi(i_1,\dots,i_{M-1},j_M(i_1,\dots,i_M)) = 0$ for each fixed $i_M \in \{3,\dots, N-2\}$.  As in Ex. 4, we define $P^{(M)}(i_1,\dots,i_{M},t) = Q^{(M)}((i_1,\dots,i_{M}),t\text{ sign }(j_M(i_1,\dots,i_{M}) - i_M))$ for $0 \le t \le 6M \vert j_M(i_1,\dots,i_M) - i_M\vert $ and $P^{(M)}((i_1,\dots,i_{M}),t) = (i_1,\dots,j_M(i_1,\dots,i_M))$ for $6M \vert j_M(i_1,\dots,i_M) - i_M\vert \le t \le T^{(M)}$.  

We define $Q^{(M)}$ so that it establishes a path from $(i_1,\dots,i_{M-1},i_{M})$ to $(i_1,\dots,i_{M-1},i_{M}+1)$.  If  $i_M \text{ mod } 4 \equiv 1$, $i_M \text{ mod } 4 \equiv 2$,  $i_M \text{ mod } 4 \equiv 3$, then there is a direct edge from $(i_1,\dots,i_{M})$ to $(i_1,\dots,i_{M}+1)$, so we can set $Q^{(M)}((i_1,\dots,i_{M}),t) = (i_1,\dots,i_{M}+1)$ for $1 \le t \le 6M$.  Otherwise, we advance $i_M$ using the following downward sweep, upward sweep procedure analogous to that of Ex. 4.

(1) If $i_{M-1} \le 1$, set $\alpha = M-1$ and perform a ``downward sweep'' of $\alpha$ as follows.   If $(i_1(\tau),\dots,i_{\alpha}(\tau),\dots,i_M(\tau)) \equiv Q^{(M)}((i_1,\dots,i_{M}),\tau)$ has $i_{\alpha}(\tau) \le 1$, determine if there is a direct edge from $(i_1(\tau),\dots,i_{\alpha}(\tau),\dots,i_M(\tau))$ to $(i_1(\tau),\dots,i_{\alpha}(\tau)+1,\dots,i_M(\tau))$ and then an edge to $(i_1(\tau),\dots,i_{\alpha}(\tau)+2,\dots,i_M(\tau))$ due to two 1-qubit gates.  If both edges are available, set $Q^{(M)}((i_1,\dots,i_M),\tau+1) = (i_1(\tau),\dots,i_{\alpha}(\tau)+1,\dots,i_M(\tau))$, set $Q^{(M)}((i_1,\dots,i_M),\tau+2) = (i_1(\tau),\dots,i_{\alpha}(\tau)+2,\dots,i_M(\tau))$, increment $\tau$ by 2, and halt the downward sweep of $\alpha$.  If $i_{\alpha}(\tau) \le 1$ but no such direct connection to $i_{\alpha} (\tau)+ 2$ is available, examine $i_{\alpha -1}(\tau)$.  If $i_{\alpha-1}(\tau) \ge 2$, then we subtract 2 from $i_{\alpha-1}(\tau)$ and add 2 to $i_{\alpha}(\tau)$ using the paths shown in Table \ref{table:QMdef}.  Decrement $\alpha$ and continue the downward sweep.  If the downward sweep reaches $\alpha=1$, we must have $i_1(\tau) \le 1$. Thus, there will always be an edge to $(i_1(\tau)+1,\dots,i_M)$ and then an edge to $(i_1(\tau)+2,\dots,i_M)$ due to two 1-qubit gates, and the downward sweep will inevitably halt at some $\alpha = A$.

We then begin an ``upward sweep," setting $\alpha = A+1$.  If $i_{\alpha-1}(\tau) \ge 2$ and $i_{\alpha}(\tau) \le 1$, then we subtract 2 from $i_{\alpha-1}(\tau)$ and add 2 to $i_{\alpha}$ using Table \ref{table:QMdef}.  We increment $\alpha$ and continue the upward sweep.  When we reach $\alpha = M$, we make the final advance of $i_M$ to $i_M+1$ using Table \ref{table:finaladvance}.

The downward sweep and upward sweep involve at most $M$ values of $\alpha$ each, and each change in $\alpha$ requires at most $3$ edges as given in Table \ref{table:QMdef}.  Thus, $\tau$ is at most $6M$ when both sweeps are complete.  If $\tau < 6M$, fix $Q^{(M)}((i_1,\dots,i_{M}),\tau) = (i_1,i_2,i_3,\dots,i_{M-2},i_{M-1},i_{M}+1)$ for all remaining $\tau$ up to $6M$.  This completes the definition of $Q^{(M)}$ when $i_{M-1} \le 1$.

(2) If $i_{M-1} \ge 2$, begin with Table \ref{table:finaladvance}.   Then, follow the downward sweep and upward sweep procedures described in case (1).  Just omit the final advance -- do not use Table \ref{table:finaladvance} a second time.

As in Ex. 4, we stitch multiple $Q^{(M)}$ together to complete the definition.  We can then move on to $P^{(M-1)}$ and $Q^{(M-1)}$ and so on until we have a complete definition of $P$.  Using reasoning analogous to that of Ex. 4, we find that $T = 6M(N-2)+6(M-1)(4)+6(M-2)(4)+6(1)(4) = 6M(N+2M-4)$ and $B \le 2N+2(4+\dots+4) = 2(N+4M-4)$; we then apply Lem. \ref{lemma:path} to obtain $e_1(1) = E_{\phi} \ge 2\mathcal{E}/(12M(N+2M-4)+1)(4(N+4M-4)+1)$.  Combining this result with earlier results gives the desired bound on $g_{min}$.
\end{proof}

Note that $g_{min}$ is greater than order $1/N^4$, assuming $M \le N$, consistent with the claims of \cite{Mizel02,Mizel14}.  Thm. \ref{thm:adiabatic} implies that an adiabatic evolution time $T \gg 1095M^2\mathcal{E}^2/g_{min}^3 \approx 4 \times 10^{12} M^5 N^{9}/\mathcal{E}$ is sufficient for computation with $H(\lambda)$.

\begin{center}
\begin{table}
\begin{tabular}{||c|c||c|c||c|c||c|c||}
\hline
$i_{\alpha-1}(\tau)$&$i_{\alpha}(\tau)$&$i_{\alpha-1}(\tau+1)$&$i_{\alpha}(\tau+1)$&$i_{\alpha-1}(\tau+2)$&$i_{\alpha}(\tau+2)$&$i_{\alpha-1}(\tau+3)$&$i_{\alpha}(\tau+3)$ \\
\hline
2&0&2&1&3&1&0&2\\
\hline
3&0&3&1&0&2&1&2\\
\hline
2&1&3&1&0&2&0&3\\
\hline
3&1&0&2&1&2&1&3\\
\hline
\end{tabular}
\caption{Path of $Q^{(M)}$ from $(i_1(\tau),\dots,i_{\alpha-1}(\tau),i_{\alpha}(\tau),\dots,i_{M}(\tau))$ to $(i_1(\tau+3)=i_1(\tau),\dots,i_{\alpha-1}(\tau+3) = i_{\alpha-1}(\tau)-2,i_{\alpha}(\tau+3)= i_{\alpha}(\tau)+2,\dots,i_{M}(\tau+3)=i_{M}(\tau))$.  This 3-step path plays the role that the 1-step path from $(i_1(\tau),\dots,i_{\alpha-1}(\tau),i_{\alpha}(\tau),\dots,i_{M}(\tau))$ to $(i_1(\tau+1)=i_1(\tau),\dots,i_{\alpha-1}(\tau+1) = i_{\alpha-1}(\tau)-1,i_{\alpha}(\tau+1)= i_{\alpha}(\tau)+1,\dots,i_{M}(\tau+1)=i_{M}(\tau))$ played in Ex. 4.  Here, 3 steps are necessary instead of just 1 because of the presence of 1-qubit gates.}
\label{table:QMdef}
\end{table}
\end{center}

\begin{center}
\begin{table}
\begin{tabular}{||c|c||c|c||c|c||}
\hline
$i_{M-1}(\tau)$&$i_{M}(\tau)$& $i_{M-1}(\tau+1)$&$i_{M}(\tau+1)$ & $i_{M-1}(\tau+2)$&$i_{M}(\tau+2)$ \\
\hline
2&$i_{M}(\tau)$ & 3 & $i_{M}(\tau)$ & 0 & $i_{M}(\tau)+1$\\
\hline
3&$i_{M}(\tau)$ & 0 &  $i_{M}(\tau)+1$ & 1 &  $i_{M}(\tau)+1$\\
\hline
\end{tabular}
\caption{Path of $Q^{(M)}$ from $(i_1(\tau),\dots,i_{M-1}(\tau),i_{M}(\tau))$ to $(i_1(\tau+2) = i_1(\tau),\dots,i_{M-1}(\tau+2)=i_{M-1}(\tau)-2 ,i_{M}(\tau+2)=i_{M}(\tau)+1)$.}
\label{table:finaladvance}
\end{table}
\end{center}

Now, we bound the gap for a circuit with all-to-all qubit connectivity.  

\begin{center}
\begin{table}
\begin{tabular}{|c|c|c|c|c|c|c|c|}
\hline
\multirow{2}{*}{$U_{A,4;B,4}$}& A & M & M-3 & M-5 & $\dots$ & 4 & 2 \\
\cline{2-8}
& B & M-1 & M-2 & M-4 & $\dots$ & 5 & 3 \\
\hline
\multirow{2}{*}{$U_{A,6;B,6}$}&A & 1 & 3 & 5 & $\dots$ & M-4 & M-2 \\
\cline{2-8}
& B & 2 & 4 & 6 & $\dots$ & M-3 & M-1 \\
\hline
\multirow{2}{*}{$U_{A,8;B,8}$}& A & M & M-1 & M-3 & $\dots$ & 6 & 4 \\
\cline{2-8}
& B & M-2 & M-4 & M-6 & $\dots$ & 3 & 2 \\
\hline
\multirow{2}{*}{$U_{A,10;B,10}$}&A & 1 & 2 & 3 & $\dots$ & M-6 & M-4 \\
\cline{2-8}
&B & 4 & 6 & 8 & $\dots$ & M-1 & M-2 \\
\hline
\multirow{2}{*}{$U_{A,12;B,12}$}&A & M & M-2 & M-1 & $\dots$ & 8 & 6 \\
\cline{2-8}
&B & M-4 & M-6 & M-8 & $\dots$ & 2 & 4\\
\hline
\multirow{2}{*}{$U_{A,14;B,14}$}& A & 1 & 4 & 2 & $\dots$ & M-8 & M-6 \\
\cline{2-8}
& B & 6 & 8 & 10 & $\dots$ & M-2 & M-4 \\
\hline
\end{tabular}
\caption{Pattern of 2-qubit gates at time step $2j+2$ in a qubit configuration with all-to-all connectivity.}
\label{table:roundrobin}
\end{table}
\end{center}

\begin{definition}[All-to-all configuration]  Consider the circuit used to define $H(\lambda)$ in Def. \ref{def:GSQC}.  Suppose that the 1-dimensional configuration of Def. \ref{def:1d} is altered so that the 2-qubit gates $U_{A,2j+2;B,2j+2}$ of $H(\lambda)$ give opportunities for every pair of qubits to interact via 2-qubit gates.  All pairs are treated democratically, consistent with an all-to-all physical connectivity of qubits.  In particular, the 2-qubit gates $U_{A,2j+2;B,2j+2}$ conform to the ``round-robin tournament'' pattern, presented in Table \ref{table:roundrobin}.    When $2j+2 = 4$, qubits $A=2$ and $B=3$ undergo a 2-qubit gate $U_{A,4;B,4}$, as do qubits $A=4$ and $B=5$, qubits $A=6$ and $B=7$, and so on out to $A=M-1$ and $B=M$.  Qubit $1$ undergoes a 1-qubit gate.  When $2j+2 = 6$, qubits $A=1$ and $B=2$ undergo a 2-qubit gate $U_{A,6;B,6}$, as do qubits $A=3$ and $B=4$, qubits $A=5$ and $B=6$, and so on out to $A=M-2$ and $B=M-1$.  Qubit $M$ undergoes a 1-qubit gate.  The pattern of gates at $2j+2 = 8$ is obtained by keeping $A = M$ fixed and shifting the other $U_{A,4;B,4}$ entries around the table in a clockwise fashion.  The pattern of gates at $2j+3 = 10$ is obtained by keeping $A = 1$ fixed and shifting the other $U_{A,6;B,6}$ entries around the table in a clockwise fashion.  In general, the pattern of gates at $2j+2 \ge 8$ is obtained by keeping the top left $A$ index fixed and shifting the other $U_{A,2j-2;B,2j-2}$ entries around the table in a clockwise fashion.

As in the 1-dimensional configuration of Def. \ref{def:1d}, there are extra 1-qubit identity gates at the time steps near $i=1$ and $i=N$.  The qubits with even $A$ have $o_A = 1$.  The qubits with odd $A$ have $o_A = 3$.  The largest time step with 2-qubit gates occurs at $i=N-3$.  Any qubit index that appears as the first index $A$ of a 2-qubit gate $U_{A,N-3;B,N-3}$ at time step $N-3$ has $f_A = N$.  The remaining qubits have final time step $f_A = N-2$.
\label{def:alltoall}
\end{definition}

The following theorem bounds the gap for this pattern of 2-qubit gates using Thm. \ref{thm:one-dimensional}.

\begin{theorem}
For a circuit with the all-to-all configuration of gates specified in Def. \ref{def:alltoall}, $g_{min}$ obeys the lower bound given in Thm. \ref{thm:one-dimensional} for a 1-dimensional configuration of gates.
\end{theorem}

\begin{proof}
We refer to our Hamiltonian, describing a configuration of qubits with all-to-all connectivity, as $H^{\text{all-to-all}}(\lambda)$.  We refer to the Hamiltonian of Thm. \ref{thm:one-dimensional}, describing a 1-dimensional configuration of qubits, as $H^{\text{1d}}(\lambda)$.  We assume that Lem. \ref{lemma:identitygates} has been applied to both $H^{\text{1d}}(\lambda)$ and $H^{\text{all-to-all}}(\lambda)$, so that only identity gates appear in each Hamiltonian.  We will show that the resulting two Hamiltonians are related by a unitary transformation in their penalty-free low-energy subspace, from which the proof immediately follows.

Define the unitary operator
\begin{eqnarray}
\mathcal{W}_{A,k;B,k} & = & \sum_{i_A< k, i_B < k} C^\dagger_{A,i_A} C_{A,i_A} C^\dagger_{B,i_B} C_{B,i_B}  +  \sum_{i_A< k, i_B \ge k} C^\dagger_{A,i_A} C_{A,i_A} C^\dagger_{B,i_B} C_{B,i_B} \nonumber \\
&& + \sum_{i_A \ge k, i_B< k} C^\dagger_{A,i_A} C_{A,i_A} C^\dagger_{B,i_B} C_{B,i_B}  + \sum_{i_A \ge k, i_B \ge k} C^\dagger_{A,i_A} C_{A,i_B} C^\dagger_{B,i_B} C_{B,i_A} 
\label{eq:swap}
\end{eqnarray}
This operator swaps qubits $A$ and $B$ if both reside at time step $k$ or later and has no affect on any other qubits.  Within the basis of eigenstates specified in Def. \ref{def:GSQC}, it leaves invariant any 1-qubit or 2-qubit gate at time step $k-1$ or earlier as well as any associated penalty Hamiltonian.  It leaves $h_{A,B}^{k,k}(\mathcal{I} \otimes \mathcal{I})+h_A^{k+1}(\mathcal{I}) + h_B^{k+1}(\mathcal{I})$ invariant.  It transforms any 2-qubit gate between $A$ and $C$ or $B$ and $C$ at a time step $k+2$ or later: $h_{A,C}^{k+2,k+2}(\mathcal{I} \otimes \mathcal{I}) = \mathcal{W}^\dagger_{A,k;B,k}  h_{B,C}^{k+2,k+2}(\mathcal{I} \otimes \mathcal{I})\mathcal{W}_{A,k;B,k}$.  Because $\mathcal{W}_{A,k;B,k}$ only swaps $A$ and $B$ if they are both at time step $k$ or later, it does not perfectly transform the associated penalty Hamiltonian $h_{B,C}^{k+2,k+2}(P)$ into $h_{A,C}^{k+2,k+2}(P)$.  However, note that the penalty-free subspace of $\mathcal{W}^\dagger_{A,k;B,k}(h_{A,B}^{k,k}(P) + h_{B,C}^{k+2,k+2}(P))\mathcal{W}_{A,k;B,k}$ is identical with the penalty-free subspace of $h_{A,B}^{k,k}(P) + h_{A,C}^{k+2,k+2}(P)$.   Indeed, if $i_A, i_B  < k$ or $i_A, i_B \ge k$, $\mathcal{W}^\dagger_{A,k;B,k}(h_{A,B}^{k,k}(P) + h_{B,C}^{k+1,k+1}(P))\mathcal{W}_{A,k;B,k}$ is identical with $h_{A,B}^{k,k}(P) + h_{A,C}^{k+1,k+1}(P)$.  On the other hand, if $i_A < k$, $i_B  \ge k$ or  $i_A \ge k$, $i_B  < k$, then $\mathcal{W}^\dagger_{A,k;B,k}h_{A,B}^{k,k}(P)\mathcal{W}_{A,k;B,k}$ and $h_{A,B}^{k,k}(P)$ both impose a penalty, so we are outside the penalty-free subspace.  It is sufficient to focus on the penalty-free subspace for our proof because $g_{min}$ is associated with a state that lies in this subspace.

Suppose that we start with $H^{\text{1d}}(\lambda)$ and apply the transformation (\ref{eq:swap}) to every pair in Table \ref{table:2qbgatesin1d} at time step $2j+2 = 6$, except the leftmost pair $(1,2)$.  Explicitly, we let $\mathcal{W}_k\equiv \mathcal{W}_{M-2,k;M-1,k} \dots \mathcal{W}_{5,k;6,k} \mathcal{W}_{3,k;4,k}$ and compute ${\mathcal{W}^\dagger}_6 H^{\text{1d}}(\lambda) \mathcal{W}_6 $.  The pattern of pairs from time step $2j+2 \ge 8$ in Table \ref{table:2qbgatesin1d} that characterizes $H^{\text{1d}}(\lambda)$ changes to the pattern shown in Table \ref{table:2qbgatesin1dafterswap6} once we use $h_{A,B}^{k,k}(\mathcal{I} \otimes \mathcal{I}) = h_{B,A}^{k,k}(\mathcal{I} \otimes \mathcal{I})$ to swap the label in the $A$ row and the label in the $B$ row.   The pairing of qubits at time step $2j+2 = 8$ in ${\mathcal{W}^\dagger}_6 H^{\text{1d}}(\lambda) \mathcal{W}_6 $ is now identical with that of $H^{\text{all-to-all}}(\lambda)$.

\begin{center}
\begin{table}
\begin{tabular}{|c|c|c|c|c|c|c|c|}
\hline
\multirow{2}{*}{$U_{A,4;B,4}$}& A & M & M-3 & M-5 & $\dots$ & 4 & 2 \\
\cline{2-8}
& B & M-1 & M-2 & M-4 & $\dots$ & 5 & 3 \\
\hline
\multirow{2}{*}{$U_{A,6;B,6}$}&A & 1 & 3 & 5 & $\dots$ & M-4 & M-2 \\
\cline{2-8}
& B & 2 & 4 & 6 & $\dots$ & M-3 & M-1 \\
\hline
\multirow{2}{*}{$U_{A,8;B,8}$}&A & M & M-1 & M-3 & $\dots$ & 6 & 4 \\
\cline{2-8}
& B & M-2 & M-4 & M-6 & $\dots$ & 3 & 2 \\
\hline
\multirow{2}{*}{$U_{A,10;B,10}$}&A & 1 & 4 & 6 & $\dots$ & M-3 & M-1 \\
\cline{2-8}
& B & 2 & 3 & 5 & $\dots$ & M-4 & M-2 \\
\hline
\multirow{2}{*}{$U_{A,12;B,12}$}&A & M & M-1 & M-3 & $\dots$ & 6 & 4 \\
\cline{2-8}
& B & M-2 & M-4 & M-6 & $\dots$ & 3 & 2 \\
\hline
\multirow{2}{*}{$U_{A,14;B,14}$}&A & 1 & 4 & 6 & $\dots$ & M-3 & M-1 \\
\cline{2-8}
& B & 2 & 3 & 5 & $\dots$ & M-4 & M-2 \\
\hline
\end{tabular}
\caption{Pattern of 2-qubit gates in a 1-dimensional qubit configuration after application of $\mathcal{W}_{M-2,6;M-1,6} \dots \mathcal{W}_{5,6;6,6} \mathcal{W}_{3,6;4,6}$ and some swapping of labels $A$ and $B$.}
\label{table:2qbgatesin1dafterswap6}
\end{table}
\end{center}

Suppose that prior to applying the transformation $\mathcal{W}_6$, we had applied the transformation (\ref{eq:swap}) to every pair in Table \ref{table:2qbgatesin1d} at time step $2j+2 = 8$, except the leftmost pair $(M,M-1)$.  That is, we define $\mathcal{V}_k\equiv \mathcal{W}_{M-3,k;M-2,k} \dots \mathcal{W}_{4,k;5,k} \mathcal{W}_{2,k;3,k}$ and compute ${\mathcal{W}^\dagger}_6{\mathcal{V}^\dagger}_8 H^{\text{1d}}(\lambda) \mathcal{V}_8  \mathcal{W}_6$.  The pattern in Table \ref{table:2qbgatesin1dafterswap6} now changes to that in Table \ref{table:2qbgatesin1dafterswap8}, which now agrees with the pairing in Table \ref{table:roundrobin} up to time step $2j+2 = 10$.

\begin{center}
\begin{table}
\begin{tabular}{|c|c|c|c|c|c|c|c|}
\hline
\multirow{2}{*}{$U_{A,4;B,4}$}& A & M & M-3 & M-5 & $\dots$ & 4 & 2 \\
\cline{2-8}
& B & M-1 & M-2 & M-4 & $\dots$ & 5 & 3 \\
\hline
\multirow{2}{*}{$U_{A,6;B,6}$}&A & 1 & 3 & 5 & $\dots$ & M-4 & M-2 \\
\cline{2-8}
& B & 2 & 4 & 6 & $\dots$ & M-3 & M-1 \\
\hline
\multirow{2}{*}{$U_{A,8;B,8}$}&A & M & M-1 & M-3 & $\dots$ & 6 & 4 \\
\cline{2-8}
& B & M-2 & M-4 & M-6 & $\dots$ & 3 & 2 \\
\hline
\multirow{2}{*}{$U_{A,10;B,10}$}&A & 1 & 2 & 3 & $\dots$ & M-6 & M-4 \\
\cline{2-8}
& B & 4 & 6 & 8 & $\dots$ & M-1 & M-2 \\
\hline
\multirow{2}{*}{$U_{A,12;B,12}$}&A & M & M-4 & M-6 & $\dots$ & 3 & 2 \\
\cline{2-8}
& B & M-2 & M-1 & M-3 & $\dots$ & 6 & 4 \\
\hline
\multirow{2}{*}{$U_{A,14;B,14}$}&A & 1 & 2 & 3 & $\dots$ & M-6 & M-4 \\
\cline{2-8}
& B & 4 & 6 & 8 & $\dots$ & M-1 & M-2 \\
\hline
\end{tabular}
\caption{Pattern of 2-qubit gates in a 1-dimensional qubit configuration after application of $\mathcal{W}_{M-3,8;M-2,8} \dots W_{4,8;5,8} \mathcal{W}_{2,8;3,8}$ then $\mathcal{W}_{M-2,6;M-1,6} \dots \mathcal{W}_{5,6;6,6} \mathcal{W}_{3,6;4,6}$.}
\label{table:2qbgatesin1dafterswap8}
\end{table}
\end{center}

More generally, the table associated with 
\begin{equation}
{\mathcal{W}^\dagger}_6{\mathcal{V}^\dagger}_8 {\mathcal{W}^\dagger}_{10} {\mathcal{V}^\dagger}_{12} \dots {\mathcal{W}^\dagger}_{N-5} {\mathcal{V}^\dagger}_{N-3} H^{\text{1d}}(\lambda) \mathcal{V}_{N-3}  \mathcal{W}_{N-5} \dots \mathcal{V}_{12}  \mathcal{W}_{10} \mathcal{V}_8  \mathcal{W}_6.
\label{eq:WVWVHWVWV}
\end{equation}
agrees at all rows with Table \ref{table:roundrobin}.  To prove this, we note that $\mathcal{W}_6$ converted row $2j+2=8$ of Table \ref{table:2qbgatesin1d} to row $2j+2 = 8$ of Table \ref{table:roundrobin}.  We note also that $\mathcal{V}_8  \mathcal{W}_6$ effected the clockwise shift (the permutation that carries $3$ to $2$, $5$ to $3$, etc.) that turned row $2j+2=10$ of Table \ref{table:2qbgatesin1d} into row $2j+2 = 10$ of Table \ref{table:roundrobin}.  It follows that the product $\mathcal{V}_{2j+2}  \mathcal{W}_{2j} \dots \mathcal{V}_{8}  \mathcal{W}_{6}$ gives repeated clockwise shifts, changing row $2j+2$ of Table \ref{table:2qbgatesin1d} into  row $2j+2$ of Table \ref{table:roundrobin} for even $j$.  It also follows that the product $\mathcal{W}_{2j+2}  \left(\mathcal{V}_{2j}  \mathcal{W}_{2j-2} \dots \mathcal{V}_{8}  \mathcal{W}_{6} \right)$ changes row $2j+2$ of Table \ref{table:2qbgatesin1d} into the second row of Table \ref{table:roundrobin}  and then applies repeated clockwise shifts, giving $2j+2$ row of Table \ref{table:roundrobin} for odd $j$.

Thus, Hamiltonian (\ref{eq:WVWVHWVWV}) agrees with $H^{\text{all-to-all}}(\lambda)$ on their penalty-free subspaces, proving the theorem.
\end{proof}

As in the one-dimensional case, Thm. \ref{thm:adiabatic} implies that an adiabatic evolution time $T \gg 1095M^2 \mathcal{E}^2 /g_{min}^3 \approx 4\times10^{12} M^5 N^{9}/\mathcal{E}$ is sufficient for computation in the case of all-to-all connectivity.

\section{Measuring output}

After we have carried $\lambda$ to $1$, we need to measure the output of the computation.  The following theorem bounds the probability that a measurement finds all qubits in the final $(N-3)/2 + M$ time steps.  Since Def. \ref{def:1d} specifies that all gates in these time steps are identity gates, the output of the computation can be correctly obtained if all qubits found in this range, even if they are not found at their final time step $f_A$.

\begin{theorem}
Consider a circuit with the 1-dimensional configuration of gates specified in Def. \ref{def:1d} or the all-to-all configuration of gates specified in Def. \ref{def:alltoall}.  If $\lambda = 1$, the probability that a measurement finds all qubits at the last $(N-3)/2 + M$ time steps is at least $1/3$.
\end{theorem}

\begin{proof}
The argument parallels the proof of Lem. \ref{lemma:1overN}.  When $\lambda=1$, the probability that qubit $M$ is at any time step $i$ goes like $4^{M-1}$ times a normalization constant.  This is equal for each of the $N-3$ possible time steps of qubit $M$, so the probability qubit $M$ resides at any given time step is $1/(N-3)$.  In particular, the probability that it resides on the last $(N-1)/2 - (M-1)$ time steps is at least $(N-1)/2(N-3) - (M-1)/(N-3) > 1/3$ since $N \ge 6M-1$ in Def. \ref{def:1d}.

If qubit $M$ resides on the last $(N-1)/2 - (M-1)$ time steps, then the 2-qubit gates in the circuit ensure that qubit $M-1$ resides on the last $(N-1)/2 - (M-1) + 2$ time steps, qubit $M-2$ resides on the last $(N-1)/2 - (M-1) + 4$ time steps, and so on.  We conclude that all qubits $1,\dots,M$ necessarily reside on the last $(N-1)/2 - (M-1) + 2(M-1) = (N-1)/2 + (M-1) = (N-3)+M$ time steps.
\end{proof}

\section{Conclusion}

We have provided a careful analysis of parallelized, universal adiabatic quantum computation.  Taking as input an arbitrary quantum algorithm described as a circuit of unitary gates, we defined a system and its Hamiltonian $H(\lambda)$ possessing several key properties.  The ground state of $H(\lambda)$ is non-degenerate.   The ground state of $H(0)$ is a simple product state that seems straightforward to prepare.  The amount of time needed to adiabatically evolve from $\lambda=0$ to $\lambda=1$, such that the system remains almost completely in the instantaneous ground state of $H(\lambda)$, is bounded by a polynomial function of the width and the depth of the circuit of unitary gates.  Thus, the computation can be executed efficiently.  Finally, the output of the quantum algorithm can be extracted with high-probability by measuring the ground state of $H(1)$.  

We considered a case in which $H(\lambda)$ has restricted connectivity derived from a circuit operating on a one-dimensional configuration of qubits.  We also considered a case in which $H(\lambda)$ has all-to-all qubit connectivity.  

Throughout this paper, we made the assumption that the adiabatic quantum computation is proceeding in a noise-free, zero-temperature environment.  The problem of making adiabatic quantum computation fault-tolerant against thermal excitations is challenging and fascinating.  One suspects that the parallelized execution of gates integral to our $H(\lambda)$ could play a role in such fault-tolerance since parallel gates are essential for fault-tolerance in the circuit model.

\section*{Acknowledgements}

The author is very grateful for helpful remarks by Van Molino and by Stephen DiPippo.

\bibliographystyle{apsrev4-1}
\bibliography{ari_mizel}

\begin{thebibliography}{26}%
\makeatletter
\providecommand \@ifxundefined [1]{%
 \@ifx{#1\undefined}
}%
\providecommand \@ifnum [1]{%
 \ifnum #1\expandafter \@firstoftwo
 \else \expandafter \@secondoftwo
 \fi
}%
\providecommand \@ifx [1]{%
 \ifx #1\expandafter \@firstoftwo
 \else \expandafter \@secondoftwo
 \fi
}%
\providecommand \natexlab [1]{#1}%
\providecommand \enquote  [1]{``#1''}%
\providecommand \bibnamefont  [1]{#1}%
\providecommand \bibfnamefont [1]{#1}%
\providecommand \citenamefont [1]{#1}%
\providecommand \href@noop [0]{\@secondoftwo}%
\providecommand \href [0]{\begingroup \@sanitize@url \@href}%
\providecommand \@href[1]{\@@startlink{#1}\@@href}%
\providecommand \@@href[1]{\endgroup#1\@@endlink}%
\providecommand \@sanitize@url [0]{\catcode `\\12\catcode `\$12\catcode
  `\&12\catcode `\#12\catcode `\^12\catcode `\_12\catcode `\%12\relax}%
\providecommand \@@startlink[1]{}%
\providecommand \@@endlink[0]{}%
\providecommand \url  [0]{\begingroup\@sanitize@url \@url }%
\providecommand \@url [1]{\endgroup\@href {#1}{\urlprefix }}%
\providecommand \urlprefix  [0]{URL }%
\providecommand \Eprint [0]{\href }%
\providecommand \doibase [0]{http://dx.doi.org/}%
\providecommand \selectlanguage [0]{\@gobble}%
\providecommand \bibinfo  [0]{\@secondoftwo}%
\providecommand \bibfield  [0]{\@secondoftwo}%
\providecommand \translation [1]{[#1]}%
\providecommand \BibitemOpen [0]{}%
\providecommand \bibitemStop [0]{}%
\providecommand \bibitemNoStop [0]{.\EOS\space}%
\providecommand \EOS [0]{\spacefactor3000\relax}%
\providecommand \BibitemShut  [1]{\csname bibitem#1\endcsname}%
\let\auto@bib@innerbib\@empty
%</preamble>
\bibitem [{\citenamefont {Mizel}\ \emph {et~al.}(8035)\citenamefont {Mizel},
  \citenamefont {Mitchell},\ and\ \citenamefont {Cohen}}]{Mizel01}%
  \BibitemOpen
  \bibfield  {author} {\bibinfo {author} {\bibfnamefont {A.}~\bibnamefont
  {Mizel}}, \bibinfo {author} {\bibfnamefont {M.~W.}\ \bibnamefont {Mitchell}},
  \ and\ \bibinfo {author} {\bibfnamefont {M.~L.}\ \bibnamefont {Cohen}},\
  }\href@noop {} {\bibfield  {journal} {\bibinfo  {journal} {Phys. Rev. A.
  Rapid Comm.}\ }\textbf {\bibinfo {volume} {63}},\ \bibinfo {pages} {40302}
  (\bibinfo {year} {2001; quant-ph/9908035})}\BibitemShut {NoStop}%
\bibitem [{\citenamefont {Bennett}(1979)}]{Bennett79}%
  \BibitemOpen
  \bibfield  {author} {\bibinfo {author} {\bibfnamefont {C.~H.}\ \bibnamefont
  {Bennett}},\ }\href@noop {} {\bibfield  {journal} {\bibinfo  {journal} {IBM
  Journal of Research and Development}\ }\textbf {\bibinfo {volume} {6}},\
  \bibinfo {pages} {525} (\bibinfo {year} {1979})}\BibitemShut {NoStop}%
\bibitem [{\citenamefont {Toffoli}(1981)}]{Toffoli81}%
  \BibitemOpen
  \bibfield  {author} {\bibinfo {author} {\bibfnamefont {T.}~\bibnamefont
  {Toffoli}},\ }\href@noop {} {\bibfield  {journal} {\bibinfo  {journal}
  {Mathematical Systems Theory}\ }\textbf {\bibinfo {volume} {14}},\ \bibinfo
  {pages} {13} (\bibinfo {year} {1981})}\BibitemShut {NoStop}%
\bibitem [{\citenamefont {Fredkin}\ and\ \citenamefont
  {Toffoli}(1982)}]{Fredkin82}%
  \BibitemOpen
  \bibfield  {author} {\bibinfo {author} {\bibfnamefont {E.}~\bibnamefont
  {Fredkin}}\ and\ \bibinfo {author} {\bibfnamefont {T.}~\bibnamefont
  {Toffoli}},\ }\href@noop {} {\bibfield  {journal} {\bibinfo  {journal} {Int.
  J. Theor. Phys.}\ }\textbf {\bibinfo {volume} {21}},\ \bibinfo {pages} {219}
  (\bibinfo {year} {1982})}\BibitemShut {NoStop}%
\bibitem [{\citenamefont {Bennett}(1982)}]{Bennett82}%
  \BibitemOpen
  \bibfield  {author} {\bibinfo {author} {\bibfnamefont {C.~H.}\ \bibnamefont
  {Bennett}},\ }\href@noop {} {\bibfield  {journal} {\bibinfo  {journal} {Int.
  J. Theor. Phys.}\ }\textbf {\bibinfo {volume} {21}},\ \bibinfo {pages} {905}
  (\bibinfo {year} {1982})}\BibitemShut {NoStop}%
\bibitem [{\citenamefont {Benioff}(1980)}]{Benioff80}%
  \BibitemOpen
  \bibfield  {author} {\bibinfo {author} {\bibfnamefont {P.}~\bibnamefont
  {Benioff}},\ }\href@noop {} {\bibfield  {journal} {\bibinfo  {journal} {J. of
  Stat. Phys.}\ }\textbf {\bibinfo {volume} {22}},\ \bibinfo {pages} {563}
  (\bibinfo {year} {1980})}\BibitemShut {NoStop}%
\bibitem [{\citenamefont {Manin}(1980)}]{Manin80}%
  \BibitemOpen
  \bibfield  {author} {\bibinfo {author} {\bibfnamefont {Y.~I.}\ \bibnamefont
  {Manin}},\ }\href@noop {} {\bibfield  {journal} {\bibinfo  {journal}
  {Sov.Radio.}\ ,\ \bibinfo {pages} {13}} (\bibinfo {year} {1980})}\BibitemShut
  {NoStop}%
\bibitem [{\citenamefont {Feynman}(1982)}]{Feynman81}%
  \BibitemOpen
  \bibfield  {author} {\bibinfo {author} {\bibfnamefont {R.~P.}\ \bibnamefont
  {Feynman}},\ }\href@noop {} {\bibfield  {journal} {\bibinfo  {journal} {Int.
  J. of Theo. Phys.}\ }\textbf {\bibinfo {volume} {21}},\ \bibinfo {pages}
  {467} (\bibinfo {year} {1982})}\BibitemShut {NoStop}%
\bibitem [{\citenamefont {Feynman}(1985)}]{Feynman85}%
  \BibitemOpen
  \bibfield  {author} {\bibinfo {author} {\bibfnamefont {R.~P.}\ \bibnamefont
  {Feynman}},\ }\href@noop {} {\bibfield  {journal} {\bibinfo  {journal} {Opt.
  News}\ }\textbf {\bibinfo {volume} {11(2)}},\ \bibinfo {pages} {11} (\bibinfo
  {year} {1985})}\BibitemShut {NoStop}%
\bibitem [{\citenamefont {Margolus}(1986)}]{Margolus86}%
  \BibitemOpen
  \bibfield  {author} {\bibinfo {author} {\bibfnamefont {N.}~\bibnamefont
  {Margolus}},\ }in\ \href@noop {} {\emph {\bibinfo {booktitle} {Proceedings of
  New Techniques and Ideas in Quantum Measurement Theory, edited by Daniel
  Greenberger}}},\ Vol.\ \bibinfo {volume} {480}\ (\bibinfo  {publisher} {Ann.
  New York Acad. Sci.},\ \bibinfo {year} {1986})\ pp.\ \bibinfo {pages}
  {487--497}\BibitemShut {NoStop}%
\bibitem [{\citenamefont {Margolus}(1990)}]{Margolus90}%
  \BibitemOpen
  \bibfield  {author} {\bibinfo {author} {\bibfnamefont {N.}~\bibnamefont
  {Margolus}},\ }in\ \href@noop {} {\emph {\bibinfo {booktitle} {Complexity,
  Entropy, and the Physics of Information, edited by Wojciech Zurek}}}\
  (\bibinfo  {publisher} {Addison-Wesley},\ \bibinfo {year} {1990})\BibitemShut
  {NoStop}%
\bibitem [{\citenamefont {Mizel}(2001)}]{Mizel01b}%
  \BibitemOpen
  \bibfield  {author} {\bibinfo {author} {\bibfnamefont {A.}~\bibnamefont
  {Mizel}},\ }in\ \href@noop {} {\emph {\bibinfo {booktitle} {Proceedings of
  the 1st International Conferene on Experimental Implementation of Quantum
  Computation, edited by R. G. Clark}}}\ (\bibinfo  {publisher} {Rinton
  Press},\ \bibinfo {year} {2001})\ pp.\ \bibinfo {pages}
  {127--130}\BibitemShut {NoStop}%
\bibitem [{\citenamefont {Mizel}\ \emph {et~al.}(7001)\citenamefont {Mizel},
  \citenamefont {Mitchell},\ and\ \citenamefont {Cohen}}]{Mizel02}%
  \BibitemOpen
  \bibfield  {author} {\bibinfo {author} {\bibfnamefont {A.}~\bibnamefont
  {Mizel}}, \bibinfo {author} {\bibfnamefont {M.~W.}\ \bibnamefont {Mitchell}},
  \ and\ \bibinfo {author} {\bibfnamefont {M.~L.}\ \bibnamefont {Cohen}},\
  }\href@noop {} {\bibfield  {journal} {\bibinfo  {journal} {Phys. Rev. A}\
  }\textbf {\bibinfo {volume} {65}},\ \bibinfo {pages} {022315} (\bibinfo
  {year} {2002; quant-ph/0007001})}\BibitemShut {NoStop}%
\bibitem [{\citenamefont {Mizel}(2083)}]{Mizel04}%
  \BibitemOpen
  \bibfield  {author} {\bibinfo {author} {\bibfnamefont {A.}~\bibnamefont
  {Mizel}},\ }\href@noop {} {\bibfield  {journal} {\bibinfo  {journal} {Phys.
  Rev. A}\ }\textbf {\bibinfo {volume} {70}},\ \bibinfo {pages} {012304}
  (\bibinfo {year} {2004; quant-ph/0312083})}\BibitemShut {NoStop}%
\bibitem [{\citenamefont {Mizel}(2014)}]{Mizel14}%
  \BibitemOpen
  \bibfield  {author} {\bibinfo {author} {\bibfnamefont {A.}~\bibnamefont
  {Mizel}},\ }\href@noop {} {\bibfield  {journal} {\bibinfo  {journal}
  {arXiv:1403.7694}\ } (\bibinfo {year} {2014})}\BibitemShut {NoStop}%
\bibitem [{\citenamefont {Kitaev}\ \emph {et~al.}(2002)\citenamefont {Kitaev},
  \citenamefont {Shen},\ and\ \citenamefont {Vyalyi}}]{Kitaev02}%
  \BibitemOpen
  \bibfield  {author} {\bibinfo {author} {\bibfnamefont {A.~Y.}\ \bibnamefont
  {Kitaev}}, \bibinfo {author} {\bibfnamefont {A.~H.}\ \bibnamefont {Shen}}, \
  and\ \bibinfo {author} {\bibfnamefont {M.~N.}\ \bibnamefont {Vyalyi}},\
  }\href@noop {} {\emph {\bibinfo {title} {Classical and Quantum
  Computation}}}\ (\bibinfo  {publisher} {American Mathematical Society},\
  \bibinfo {year} {2002})\BibitemShut {NoStop}%
\bibitem [{\citenamefont {Kadowaki}\ and\ \citenamefont
  {Nishimori}(1998)}]{Kadowaki98}%
  \BibitemOpen
  \bibfield  {author} {\bibinfo {author} {\bibfnamefont {T.}~\bibnamefont
  {Kadowaki}}\ and\ \bibinfo {author} {\bibfnamefont {H.}~\bibnamefont
  {Nishimori}},\ }\href@noop {} {\bibfield  {journal} {\bibinfo  {journal}
  {Phys. Rev. E}\ }\textbf {\bibinfo {volume} {58}},\ \bibinfo {pages} {5355}
  (\bibinfo {year} {1998})}\BibitemShut {NoStop}%
\bibitem [{\citenamefont {Farhi}\ \emph {et~al.}(1106)\citenamefont {Farhi},
  \citenamefont {Goldstone}, \citenamefont {Gutmann},\ and\ \citenamefont
  {Sipser}}]{Farhi00}%
  \BibitemOpen
  \bibfield  {author} {\bibinfo {author} {\bibfnamefont {E.}~\bibnamefont
  {Farhi}}, \bibinfo {author} {\bibfnamefont {J.}~\bibnamefont {Goldstone}},
  \bibinfo {author} {\bibfnamefont {S.}~\bibnamefont {Gutmann}}, \ and\
  \bibinfo {author} {\bibfnamefont {M.}~\bibnamefont {Sipser}},\ }\href@noop {}
  {\  (\bibinfo {year} {2000; quant-ph/0001106})}\BibitemShut {NoStop}%
\bibitem [{\citenamefont {Aharonov}\ \emph {et~al.}(2007)\citenamefont
  {Aharonov}, \citenamefont {van Dam}, \citenamefont {Kempe}, \citenamefont
  {Landau}, \citenamefont {Lloyd},\ and\ \citenamefont {Regev}}]{Aharonov07}%
  \BibitemOpen
  \bibfield  {author} {\bibinfo {author} {\bibfnamefont {D.}~\bibnamefont
  {Aharonov}}, \bibinfo {author} {\bibfnamefont {W.}~\bibnamefont {van Dam}},
  \bibinfo {author} {\bibfnamefont {J.}~\bibnamefont {Kempe}}, \bibinfo
  {author} {\bibfnamefont {Z.}~\bibnamefont {Landau}}, \bibinfo {author}
  {\bibfnamefont {S.}~\bibnamefont {Lloyd}}, \ and\ \bibinfo {author}
  {\bibfnamefont {O.}~\bibnamefont {Regev}},\ }\href@noop {} {\bibfield
  {journal} {\bibinfo  {journal} {SIAM Journal on Computing}\ }\textbf
  {\bibinfo {volume} {37}},\ \bibinfo {pages} {166} (\bibinfo {year}
  {2007})}\BibitemShut {NoStop}%
\bibitem [{\citenamefont {Mizel}\ \emph {et~al.}(2007)\citenamefont {Mizel},
  \citenamefont {Lidar},\ and\ \citenamefont {Mitchell}}]{Mizel07}%
  \BibitemOpen
  \bibfield  {author} {\bibinfo {author} {\bibfnamefont {A.}~\bibnamefont
  {Mizel}}, \bibinfo {author} {\bibfnamefont {D.~A.}\ \bibnamefont {Lidar}}, \
  and\ \bibinfo {author} {\bibfnamefont {M.~W.}\ \bibnamefont {Mitchell}},\
  }\href@noop {} {\bibfield  {journal} {\bibinfo  {journal} {Phys. Rev. Lett.}\
  }\textbf {\bibinfo {volume} {99}},\ \bibinfo {pages} {070502} (\bibinfo
  {year} {2007})}\BibitemShut {NoStop}%
\bibitem [{\citenamefont {Childs}\ \emph {et~al.}(2014)\citenamefont {Childs},
  \citenamefont {Gosset},\ and\ \citenamefont {Webb}}]{Childs14}%
  \BibitemOpen
  \bibfield  {author} {\bibinfo {author} {\bibfnamefont {A.~M.}\ \bibnamefont
  {Childs}}, \bibinfo {author} {\bibfnamefont {D.}~\bibnamefont {Gosset}}, \
  and\ \bibinfo {author} {\bibfnamefont {Z.}~\bibnamefont {Webb}},\ }\href@noop
  {} {\bibfield  {journal} {\bibinfo  {journal} {Proceedings of the 41st
  International Colloquium on Automata, Languages, and Programming (ICALP
  2014)}\ ,\ \bibinfo {pages} {308}} (\bibinfo {year} {2014})}\BibitemShut
  {NoStop}%
\bibitem [{\citenamefont {Breuckmann}\ and\ \citenamefont
  {Terhal}(2014)}]{Breuckmann14}%
  \BibitemOpen
  \bibfield  {author} {\bibinfo {author} {\bibfnamefont {N.~P.}\ \bibnamefont
  {Breuckmann}}\ and\ \bibinfo {author} {\bibfnamefont {B.~M.}\ \bibnamefont
  {Terhal}},\ }\href@noop {} {\bibfield  {journal} {\bibinfo  {journal} {J.
  Phys. A}\ }\textbf {\bibinfo {volume} {47}},\ \bibinfo {pages} {195304}
  (\bibinfo {year} {2014})}\BibitemShut {NoStop}%
\bibitem [{\citenamefont {Gosset}\ \emph {et~al.}(2015)\citenamefont {Gosset},
  \citenamefont {Terhal},\ and\ \citenamefont {Vershynina}}]{Gosset15}%
  \BibitemOpen
  \bibfield  {author} {\bibinfo {author} {\bibfnamefont {D.}~\bibnamefont
  {Gosset}}, \bibinfo {author} {\bibfnamefont {B.~M.}\ \bibnamefont {Terhal}},
  \ and\ \bibinfo {author} {\bibfnamefont {A.}~\bibnamefont {Vershynina}},\
  }\href@noop {} {\bibfield  {journal} {\bibinfo  {journal} {Phys. Rev. Lett.}\
  }\textbf {\bibinfo {volume} {114}},\ \bibinfo {pages} {140501} (\bibinfo
  {year} {2015})}\BibitemShut {NoStop}%
\bibitem [{\citenamefont {Albash}\ and\ \citenamefont
  {Lidar}(2018)}]{Albash18}%
  \BibitemOpen
  \bibfield  {author} {\bibinfo {author} {\bibfnamefont {T.}~\bibnamefont
  {Albash}}\ and\ \bibinfo {author} {\bibfnamefont {D.~A.}\ \bibnamefont
  {Lidar}},\ }\href@noop {} {\bibfield  {journal} {\bibinfo  {journal} {Rev.
  Mod. Phys.}\ }\textbf {\bibinfo {volume} {90}},\ \bibinfo {pages} {015002}
  (\bibinfo {year} {2018})}\BibitemShut {NoStop}%
\bibitem [{\citenamefont {Jansen}\ \emph {et~al.}(2007)\citenamefont {Jansen},
  \citenamefont {Ruskai},\ and\ \citenamefont {Seiler}}]{Jansen07}%
  \BibitemOpen
  \bibfield  {author} {\bibinfo {author} {\bibfnamefont {S.}~\bibnamefont
  {Jansen}}, \bibinfo {author} {\bibfnamefont {M.-B.}\ \bibnamefont {Ruskai}},
  \ and\ \bibinfo {author} {\bibfnamefont {R.}~\bibnamefont {Seiler}},\
  }\href@noop {} {\bibfield  {journal} {\bibinfo  {journal} {J. Math. Phys.}\
  }\textbf {\bibinfo {volume} {48}},\ \bibinfo {pages} {102111} (\bibinfo
  {year} {2007})}\BibitemShut {NoStop}%
\bibitem [{\citenamefont {Mohar}(1991)}]{Mohar91}%
  \BibitemOpen
  \bibfield  {author} {\bibinfo {author} {\bibfnamefont {B.}~\bibnamefont
  {Mohar}},\ }\href@noop {} {\bibfield  {journal} {\bibinfo  {journal} {Graphs
  and Combinatorics}\ }\textbf {\bibinfo {volume} {7}},\ \bibinfo {pages} {53}
  (\bibinfo {year} {1991})}\BibitemShut {NoStop}%
\end{thebibliography}%

\section*{Appendix}

\noindent {\bf Example 1: 1-dimensional chain}

Consider a 1-dimensional chain.  The set ${\cal V}_1=\{{\bf i}_1,\dots,{\bf i}_{N_1}\}$ is comprised of an even number of vertices $\|{\cal V}_1\| = N_1$.  (Note that ${\bf i}_1$, a first vertex in ${\cal V}_1$, is distinct from $i_1$, the first component of the vertex ${\bf i}$.)  Edges join ${\bf i}_1$ to ${\bf i}_2$,  ${\bf i}_2$ to ${\bf i}_3$, $\dots$, and ${\bf i}_{N_1 - 1}$ to ${\bf i}_{N_1}$.  Given any non-vanishing $\phi({\bf i})$ orthogonal to the ground state, we seek a lower bound on $E_{\phi}$, the expectation value (\ref{eq:E}) of its energy.  We will prove a bound by applying Lem. \ref{lemma:path} to a suitable path $P({\bf i},t)$.

Using the definitions in Lem. \ref{lemma:path}, note that $\|{\cal N}\| = \|{\cal P}\| = N_1/2$ vertices.  Label the elements of ${\cal N}$ by ${\bf i}_{n_1},\dots,{\bf i}_{n_{N_1/2}}$ and the elements of ${\cal P}$ by ${\bf i}_{p_1},\dots,{\bf i}_{p_{N_1/2}}$.  Let us specify a path $P({\bf i},t)$ that carries even locations ${\bf i}_{2 \alpha}$ in the chain to elements of $\cal N$ and odd locations ${\bf i}_{2\alpha-1}$ in the chain to elements of $\cal P$.  In particular, for $\alpha=1,\dots,N_1/2$, define $P({\bf i}_{2 \alpha},t) = {\bf i}_{2 \alpha + t \text{ sign }(n_{\alpha} - 2\alpha)}$ for $0 \le t \le  \left|n_{\alpha} - 2 \alpha \right|$ and  $P({\bf i}_{2 \alpha},t) = {\bf i}_{n_\alpha}$ for $\left|n_{\alpha} - 2 \alpha \right| \le t \le N_1$.  Similarly, define $P({\bf i}_{2 \alpha-1},t) = {\bf i}_{2 \alpha-1 + t \text{ sign }(p_{\alpha} - (2\alpha-1))}$ for $0 \le t \le  \left|p_{\alpha} - (2\alpha-1) \right|$ and $P({\bf i}_{2\alpha-1},t) = {\bf i}_{p_\alpha}$ for $\left|p_{\alpha} - (2\alpha-1) \right| \le t \le N_1$.  Set $T = N_1$.    This $P({\bf i},t)$ satisfies the requirements of the Lemma.  A computation shows that $B < N_1$.  Thus, $E_{\phi} \ge 2\mathcal{E}/(2N_1+1)(2N_1+1)$.  Comparing to the exact result for the chain \cite{Mizel02},  $2\mathcal{E}(1-\cos \pi/(N_1)) \approx \mathcal{E}(\pi/N_1)^2$, we see that our bound underestimates its value by a factor of about $2 \pi^2$.

To illustrate this construction with concrete numbers, let $N_1 = 6$, $\phi({\bf i}_1) = \phi({\bf i}_2) = \phi({\bf i}_3) = -1$, and  $\phi({\bf i}_4) = \phi({\bf i}_5) = \phi({\bf i}_6) = 1$ as shown in Fig. \ref{fig:1dchain}.  Choose $\Phi = 0$, the median value of $\phi({\bf i})$.  Then ${\cal N} = \{{\bf i}_1,{\bf i}_2,{\bf i}_3\}$, so $n_1 = 1$, $n_2 = 2$, $n_3=3$, and ${\cal P} = \{{\bf i}_4,{\bf i}_5,{\bf i}_6\}$, so $p_1 = 4$, $p_2 = 5$, and $p_3=6$.  The definition of $P({\bf i},t)$ follows Table \ref{table:1d} where $T = 6$.  Note that the $P({\bf i},t)$ presented in the table satisfies all of the requirements of Lem. \ref{lemma:path}: (I) $P({\bf i},0) = {\bf i}$, (II) $P({\bf i},t)$ advances by at most one edge at a time, (III) $P({\bf i},6)$ is one-to-one.  Finally, one sees that (IV) is satisfied from the final column of the table, if we pair ${\bf i}_1\sim {\bf i}_2$, ${\bf i}_3 \sim {\bf i}_4$, and ${\bf i}_5 \sim {\bf i}_6$.   The conclusion is that $E_\phi \ge 2\mathcal{E}/(13)^2 \approx 0.012\mathcal{E}$.  The exact result is $2\mathcal{E}(1-\cos \pi/6) \approx 0.27\mathcal{E}.$

\begin{center}
\begin{table}
\begin{tabular}{|c|c|c|}
\hline
${\bf i}$ & $P({\bf i},t)$ & $\psi(P({\bf i},T))$\\
\hline
${\bf i}_1$ & $ \begin{array}{cc} {\bf i}_{1+t} & 0 \le t \le 3 \\ {\bf i}_4 & 4 \le t \le 6 \end{array}$& 1 \\
\hline
${\bf i}_2$ & $ \begin{array}{cc} {\bf i}_{2-t} & 0 \le t \le 1 \\ {\bf i}_1 & 2 \le t \le 6 \end{array}$ & -1\\
\hline
${\bf i}_3$ & $ \begin{array}{cc} {\bf i}_{3+t} & 0 \le t \le 2 \\ {\bf i}_5 & 3 \le t \le 6 \end{array}$ & 1\\
\hline
${\bf i}_4$ & $ \begin{array}{cc} {\bf i}_{4-t} & 0 \le t \le 2 \\ {\bf i}_2 & 3 \le t \le 6 \end{array}$ & -1\\
\hline
${\bf i}_5$ & $ \begin{array}{cc} {\bf i}_{5+t} & 0 \le t \le 1 \\ {\bf i}_6 & 2 \le t \le 6 \end{array}$ & 1\\
\hline
${\bf i}_6$ & $ \begin{array}{cc} {\bf i}_{6-t} & 0 \le t \le 3 \\ {\bf i}_3 & 4 \le t \le 6 \end{array}$ & -1\\
\hline
\end{tabular}
\caption{Sample path for 1-dimensional chain}
\label{table:1d}
\end{table}
\end{center}

\begin{figure}
\includegraphics[width=3.5in]{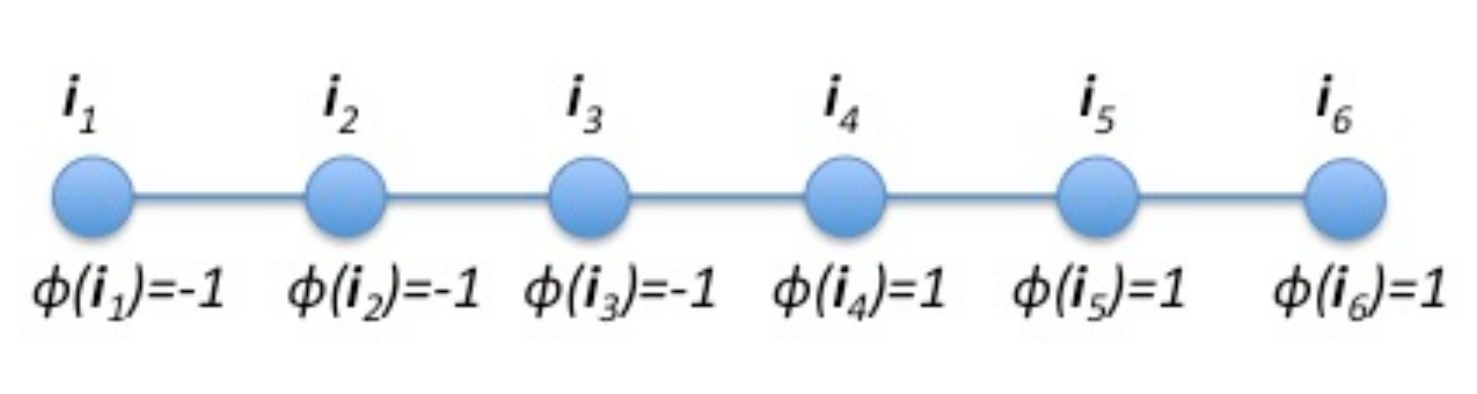}
\caption{1-dimensional chain}
\label{fig:1dchain}
\end{figure}

\noindent {\bf Example 2: 2-dimensional chain}

Consider a graph of the form ${\cal V}_1 \otimes  {\cal V}_2$ where $\|{\cal V}_1\| = N_1$ and $\|{\cal V}_2\| = N_2$ are even integers.  The vertices ${\bf i} \in {\cal V}_1 \otimes  {\cal V}_2$ have $i_1 \in \{1,\dots, N_1 \}$ and $i_2 \in \{1,\dots, N_2 \}$.  We have ${\bf i} \sim {\bf i}^\prime$ if ${i^\prime}_1  = i_1 +1$ and ${i^\prime}_2  = i_2$ or  if ${i^\prime}_1  = i_1$ and ${i^\prime}_2  = i_2 +1$.  Fig. \ref{fig:2dchain} shows an example with $N_1 = N_2 = 4$.  Given any non-vanishing $\phi({\bf i})$ satisfying $\sum_{{\bf i} \in {\cal V}} \phi({\bf i}) = 0$, we seek a lower bound on the expectation value of its energy, $E_{\phi}$.

Our strategy is to define a path $P^{(2)}$ to shuffle the $i_2$ values of vertices until $\sum_{i_1 = 1}^{N_1} sign\,\, \psi(P^{(2)}({\bf i},T^{(2)})) = 0$ for each fixed $i_2 \in \{1,\dots, N_2 \}$.  After the shuffling is complete, the problem can be solving using the method of Ex. 1 independently at each $i_2$.  The shuffling will use $T^{(2)} = N_2$ steps.

We specify $P^{(2)}$ in terms a function $j_2: \{1,\dots,N_1 \} \otimes \{1,\dots,N_2 \} \rightarrow \{1,\dots,N_2 \}$.   This function $j_2({\bf i})$ assigns a new $i_2$ coordinate to each vertex in the graph.  Given $j_2$,  we set $P^{(2)}({\bf i},t) = (i_1,i_2 + t\text{ sign }(j_2({\bf i}) - i_2))$ for $0 \le t \le |j_2({\bf i}) - i_2|$ and $P^{(2)}({\bf i},t) = (i_1,j_2({\bf i}))$ for $ |j_2({\bf i}) - i_2| \le t \le T^{(2)}$.  This $P^{(2)}$ carries each vertex ${\bf i}=(i_1,i_2)$ to $(i_1,j_2({\bf i}))$.

We now define $j_2({\bf i})$.  Condition (III) of Lem. \ref{lemma:path} requires that $P^{(2)}({\bf i},T^{(2)})$ be one-to-one, so we keep track of a pool of ``remaining'' vertices to ensure that the $(i_1,j_2({\bf i}))$ are distinct for distinct ${\bf i}=(i_1,i_2)$.  To start, the pool of remaining vertices is simply the whole set $R_0 = {\cal V}_1 \otimes  {\cal V}_2$.  We proceed inductively from $i_2$ to $i_2+1$.   Start with $i_2 = 1$.   Let ${\cal I}_{i_2}^{\cal P} = \{i_1|  (i_1,k_2) \in R_{i_2 - 1} \Rightarrow (i_1,k_2) \in {\cal P} \}$.  This is a set of $i_1$ values for which $j_2(i_1,i_2)$ can be chosen somewhat indiscriminately since any remaining $(i_1,j_2)$ will give a $\psi(i_1,j_2) \ge 0$ irrespective of the choice of $j_2$.  Similarly, define the set ${\cal I}_{i_2}^{\cal N} = \{i_1| (i_1,k_2) \in R_{i_2 - 1} \Rightarrow (i_1,k_2)  \in {\cal N} \}$.  Finally, ${\cal I}_{i_2} = \{1,\dots,N_1\} - {\cal I}_{i_2}^{\cal P} \cup {\cal I}_{i_2}^{\cal N}$ gives the $i_1$ values for which $j_2(i_1,i_2)$ should be chosen more deliberately to target $\psi(i_1,j_2) \ge 0$ or $\psi(i_1,j_2) \le 0$.    With the sets ${\cal I}_{i_2}^{\cal P}$, ${\cal I}_{i_2}^{\cal N}$, ${\cal I}_{i_2}$ established, we can define $j_2$.  For each $i_1 \in {\cal I}_{i_2}^{\cal P} \cup  {\cal I}_{i_2}^{\cal N}$, pick any $j_2(i_1,i_2)$ such that $(i_1,j_2) \in R_{i_2 - 1}$.  For $N_1/2  -  \|{\cal I}_{i_2}^{\cal N}\|$ elements $i_1$ of  ${\cal I}_{i_2}$, choose $j_2(i_1,i_2)$ such that $(i_1,j_2) \in  R_{i_2 - 1}$ and $\psi((i_1,j_2)) \le 0$  (i.e. $(i_1,j_2) \in  R_{i_2 - 1} \cap {\cal N}$).   For the remaining $N_1/2  -  \|{\cal I}_{i_2}^{\cal P}\|$ elements $i_1$ of ${\cal I}_{i_2}$, choose $j_2(i_1,i_2)$ such that $(i_1,j_2) \in R_{i_2 - 1}$ and $\psi(i_1,j_2) \ge 0$ (i.e. $(i_1,j_2) \in  R_{i_2 - 1} \cap  {\cal P}$).   Set $R_{i_2} = R_{i_2-1} \setminus \{(i_1,j_2(i_1,i_2))| i_1 \in \{1,\dots,N_1\}\}$ to update the pool of remaining vertices.  Increase $i_2$ by 1 and continue the induction until $j_2(i_1,i_2)$ is defined for all $i_2 \le  N_2$.

With $j_2(i_1,i_2)$ defined, and therefore $P^{(2)}$ as well, we have $\sum_{i_1 = 1}^N sign\,\, \psi(P^{(2)}((i_1,i_2),T^{(2)})) = 0$ for any fixed $i_2 \in \{1,\dots N_2 \}$.  At each fixed $i_2$, we can use the method of Ex. 1 to define $P^{(1)}_{i_2}(P^{(2)}((i_1,i_2),T^{(2)}),t)$ for $0 \le t \le T^{(1)} = N_1$.  Our overall map is then obtained by stringing $P^{(1)}$ and $P^{(2)}$ together, so that $P((i_1,i_2),t) = P^{(2)}((i_1,i_2),t)$ for $0 \le t \le T^{(2)}$  and $P((i_1,i_2),t) = P^{(1)}_{i_2}(P^{(2)}((i_1,i_2),T^{(2)}),t-T^{(2)})$ for $T^{(2)} \le t \le T^{(1)}+T^{(2)}$.  This map has $T = T^{(1)} + T^{(2)} = N_1+N_2$.  Because $P^{(2)}$ only uses edges that change $i_2$ and $P^{(1)}$ only uses edges that change $i_1$, we compute that $B = max\,\,\{N_1,N_2\}$.  Lem. \ref{lemma:path} implies that $E_{\phi} \ge 2/(2(N_1+N_2)+1)(2 max\,\, \{N_1,N_2\}+1)$.  This should be compared to the exact result  $2(1-\cos \pi/ max\,\, \{N_1,N_2\}) \approx \pi^2/( max\,\, \{N_1,N_2\})^2$.

We illustrate this procedure in the case $N_1 = N_2= 4$ shown in Fig. \ref{fig:2dchain}.  A sample function $\phi({\bf i})$ is shown in Fig. \ref{fig:2dphi} and listed in Table \ref{table:2d}.    The median value of $\phi({\bf i})$ is $\Phi=0.25$, which is used to define $\psi({\bf i}) = \phi({\bf i}) - \Phi$ displayed in Fig. \ref{fig:2dpsi}. One finds that ${\cal N} = \{(1,4),(2,1),(2,2),(2,3),(2,4),(3,1),(3,4),(4,1)\}$ and ${\cal P} = \{(1,1),(1,2),(1,3),(3,2),(3,3),(4,2),(4,3),(4,4)\}$.  Note that each has $N_1 N_2/2 = 8$ elements as expected since  $\Phi$ was chosen to be the median of $\phi({\bf i})$.

To define $j_2({\bf i})$, we start with $i_2 = 1$ and let $R_0$ include all $16$ vertices in the graph.  We find that ${\cal I}_{1}^{\cal P}$ is empty since there is no choice of $i_1$ that guarantees $\psi((i_1,k_2)) \ge 0$ for all $k_2$.  On the other hand, ${\cal I}_{1}^{\cal N} = \{2\}$ since $ \psi(2,k_2) \le 0$ for all $k_2$.  It follows that ${\cal I}_{1} = \{1,3,4\}$.  We can indiscriminately choose $j_2(2,1)=1$ since, for any choice we make, $ \psi(2,j_2(2,1)) \le 0$.  We choose $j_2(1,1) = 4$ and confirm that $\psi(1,j_2(1,1)) \le 0$.  Then we choose  $j_2(3,1) = 2$ so  that $\psi(3,j_2(3,1)) \ge 0$ and $j_2(4,1) =2$ so that  $\psi(4,j_2(4,1)) \ge 0$.  Based on our choice of $j_2(i_1,1)$, we now have $\sum_{i_1=1}^4 \text{ sign } \psi(i_1,j_2(i_1,1)) = 0$ as desired.  We used vertices $\{(i_1,j_2(i_1,i_2)) | i_1 \in \{1,\dots,N_1\}\} = \{(1,4),(2,1),(3,2),(4,2)\}$, so there are 12 remaining vertices in $R_1=\{(1,1),(1,2),(1,3),(2,2),(2,3),(2,4),(3,1),(3,3),(3,4),(4,1),(4,3),(4,4) \}$.

We move on to $i_2 = 2$.  Now, ${\cal I}_{2}^{\cal P} = \{1\}$ because $\psi(1,1)$, $\psi(1,2)$, and $\psi(1,3)$ are all non-negative while $\psi(1,4)$ need not concern us since we already used $(1,4)$ and thus $(1,4) \notin R_1$.  We find that ${\cal I}_{2}^{\cal N} = \{2\}$ because $\psi(2,2)$, $\psi(2,3)$, and $\psi(2,4)$ are all non-positive.  The remaining values of $i_1$ are in ${\cal I}_{2} = \{3,4\}$.  We can indiscriminately choose $j_2(1,2) = 1$ since $(1,1) \in R_1$ and $j_2(2,2) = 2$ since $(2,2) \in R_1$.  We choose $j_2(3,2) = 3$ and $j_2(4,2) = 1$.  Then, $\psi(1,j_2(1,2)) \ge 0$, $\psi(2,j_2(2,2)) \le 0$, $\psi(3,j_2(3,2)) \ge 0$, and $\psi(4,j_2(4,2)) \le 0$, which implies  $\sum_{i_1=1}^4 \text{ sign } \psi(i_1,j_2(i_1,2)) = 0$ as desired.  We have 8 remaining vertices in $R_2 = \{(1,2),(1,3),(2,3),(2,4),(3,1),(3,4),(4,3),(4,4) \}$.

We go to $i_2 = 3$.  Now, ${\cal I}_{3}^{\cal P} = \{1,4\}$, ${\cal I}_{3}^{\cal N} = \{2,3\}$, and ${\cal I}_{3}$ is empty.  We can indiscriminately choose $j_2(1,3) = 2$, $j_2(4,3) = 3$, $j_2(2,3) = 3$, and $j_2(3,3) = 1$.  Then, $\psi(1,j_2(1,3)) \ge 0$, $\psi(2,j_2(2,3)) \le 0$, $\psi(3,j_2(3,3)) \le 0$, and $\psi(4,j_2(4,3)) \ge 0$, which implies  $\sum_{i_1=1}^4 \text{ sign } \psi(i_1,j_2(i_1,3)) = 0$ as desired.  We have 4 remaining vertices in $R_3 = \{(1,3),(2,4),(3,4),(4,4) \}$.

The final value is $i_2=4$.  We see that ${\cal I}_{4}^{\cal P} = \{1,4\}$, ${\cal I}_{4}^{\cal N} = \{2,3\}$, and ${\cal I}_{4}$ is empty.  Given the remaining vertices in $R_3$, we have no choice remaining: $j_2(1,4) = 3$, $j_2(4,4) = 4$, $j_2(2,4) = 4$,  and $j_2(3,4) = 4$.  Then, $\psi(1,j_2(1,4)) \ge 0$, $\psi(2,j_2(2,4)) \le 0$, $\psi(3,j_2(3,4)) \le 0$, and $\psi(4,j_2(4,4)) \ge 0$, which implies  $\sum_{i_1=1}^4 \text{ sign } \psi(i_1,j_2(i_1,4)) = 0$.  Table \ref{table:2dj2} shows that the map from $(i_1,i_2)$ to $(i_1,j_2(i_1,i_2))$ is one-to-one.  Fig. \ref{fig:2dpsiP2} shows that shifts along the columns by $P^{(2)}$ lead to $\sum_{i_1=1}^4 \text{ sign } \psi(P^{(2)}((i_1,j_2(i_1,i_2)),T^{(2)})) = 0$ along each row.

Finally, we use the map $P^{(1)}$ to shift along each row.  Defining the total map $P({\bf i},T)$, that strings together $P^{(1)}$ and $P^{(2)}$  with  $T=T^{(1)}+T^{(2)}$, we thus ensure $\text{ sign } \psi(P({\bf i},T))$ alternates along each row.  Lem. \ref{lemma:path} then implies that $E_{\phi} \ge 2\mathcal{E}/(17)(9)$ while the exact answer is $2\mathcal{E}(1-\cos \pi/ 4) \approx \mathcal{E}\pi^2/16$.

\begin{figure}
\includegraphics[width=2.5in]{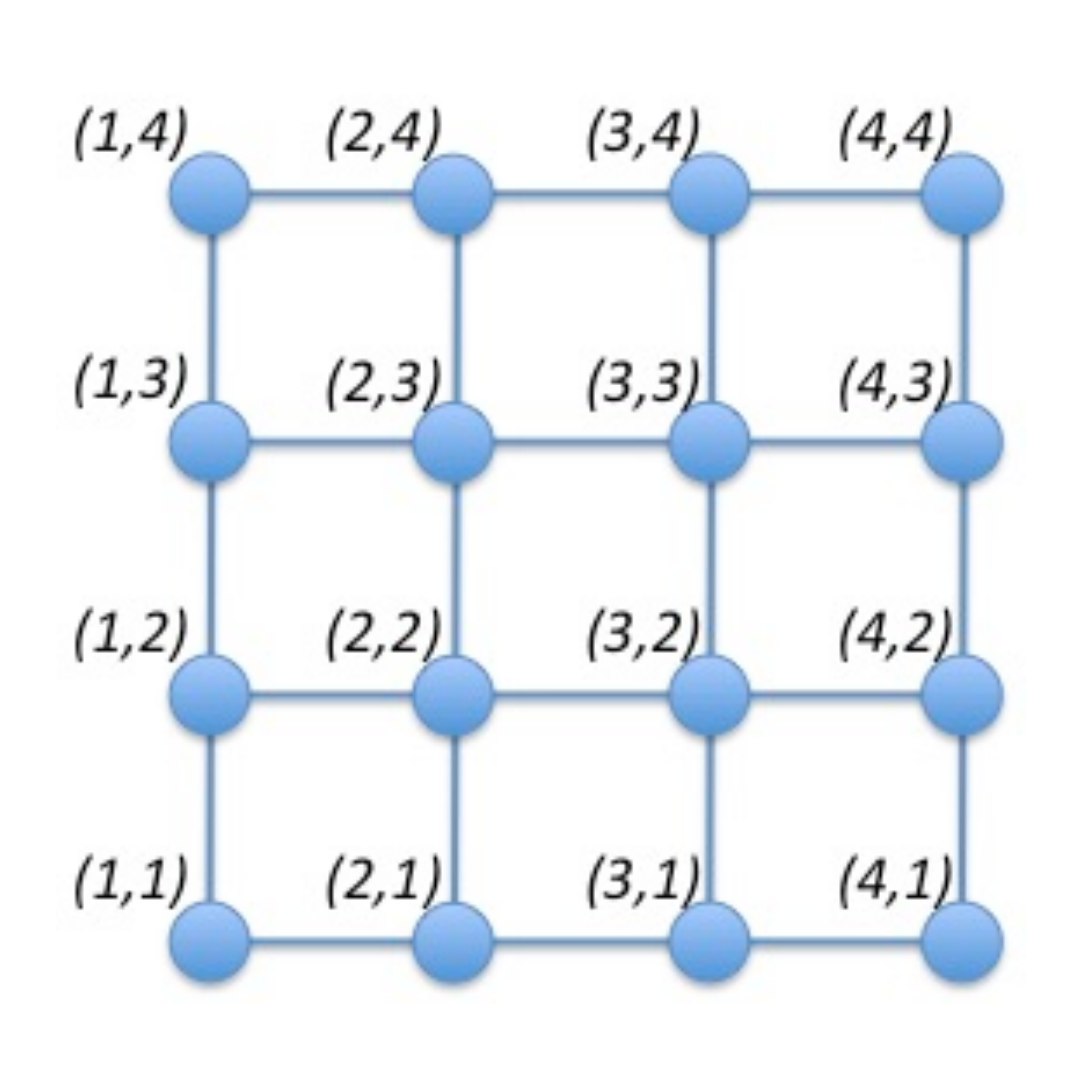}
\caption{Graph for a sample 2-dimensional chain with $N_1 = N_2 =4$.}
\label{fig:2dchain}
\end{figure}

\begin{figure}
\includegraphics[width=2.5in]{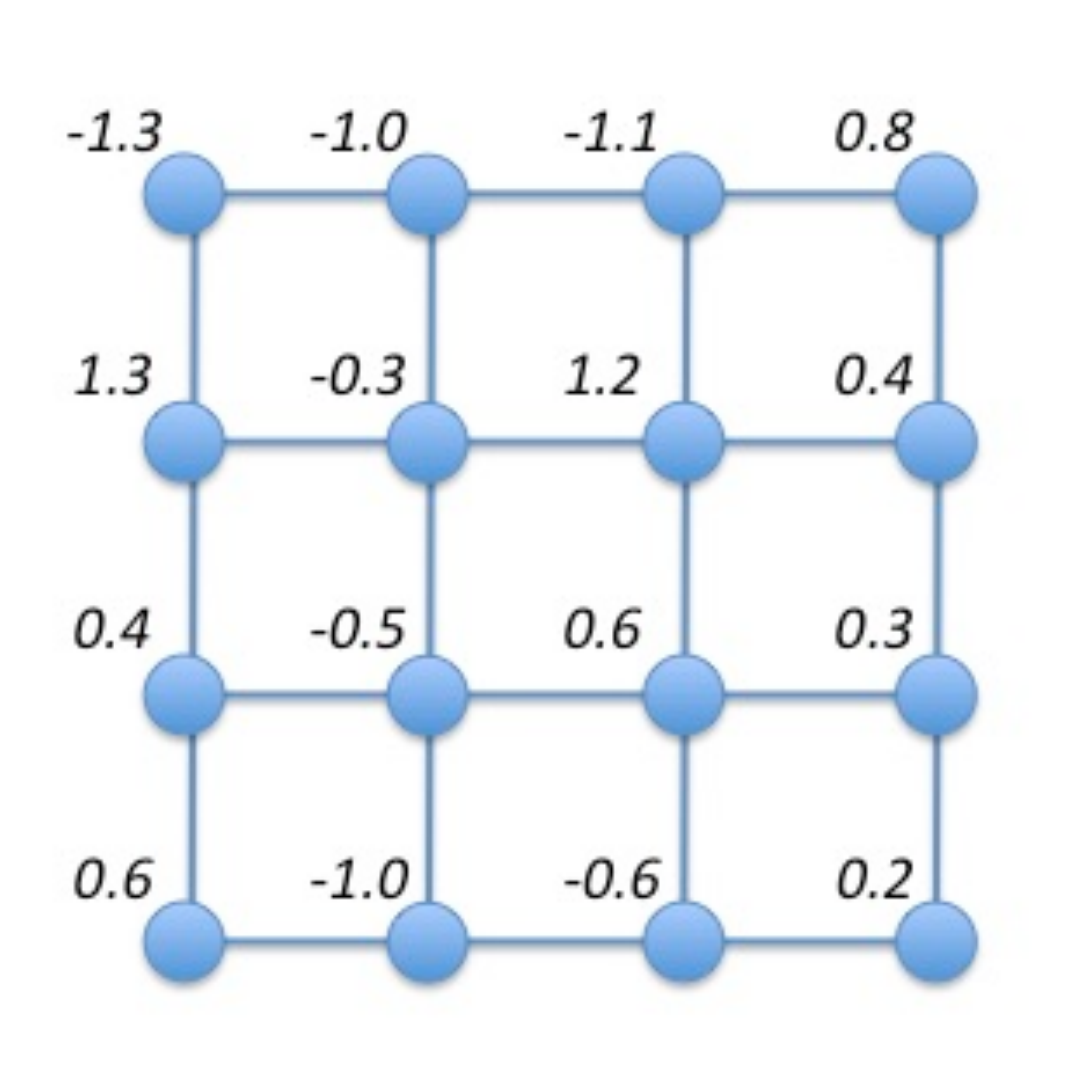}
\caption{Sample $\phi({\bf i})$ for 2-dimensional chain.}
\label{fig:2dphi}
\end{figure}

\begin{center}
\begin{table}
\begin{tabular}{|c|c|c|c|}
\hline
${\bf i}$ & $\phi({\bf i})$ & $\psi({\bf i})$ & $\text{ sign }\psi({\bf i})$\\
\hline
$(1,1)$ & 0.6 & 0.35& 1\\
\hline
$(2,1)$ & -1.0 & -1.25& -1\\
\hline
$(3,1)$ & -0.6 & -0.85& -1\\
\hline
$(4,1)$ & 0.2 & -0.05& -1\\
\hline
$(1,2)$ & 0.4 & 0.15& 1\\
\hline
$(2,2)$ & -0.5 & -0.75& -1\\
\hline
$(3,2)$ & 0.6 & 0.35& 1\\
\hline
$(4,2)$ & 0.3 & 0.05& 1\\
\hline
$(1,3)$ & 1.3 & 1.05& 1\\
\hline
$(2,3)$ & -0.3 & -0.55& -1\\
\hline
$(3,3)$ & 1.2 & 0.95& 1\\
\hline
$(4,3)$ & 0.4 & 0.15& 1\\
\hline
$(1,4)$ & -1.3 & -1.55& -1\\
\hline
$(2,4)$ & -1.0& -1.25& -1\\
\hline
$(3,4)$ & -1.1 & -1.35& -1\\
\hline
$(4,4)$ & 0.8 & 0.55& 1\\
\hline
\end{tabular}
\caption{Values of $\phi({\bf i})$ and $\psi({\bf i})$ for example 2-dimensional graph}
\label{table:2d}
\end{table}
\end{center}

\begin{figure}
\includegraphics[width=2.5in]{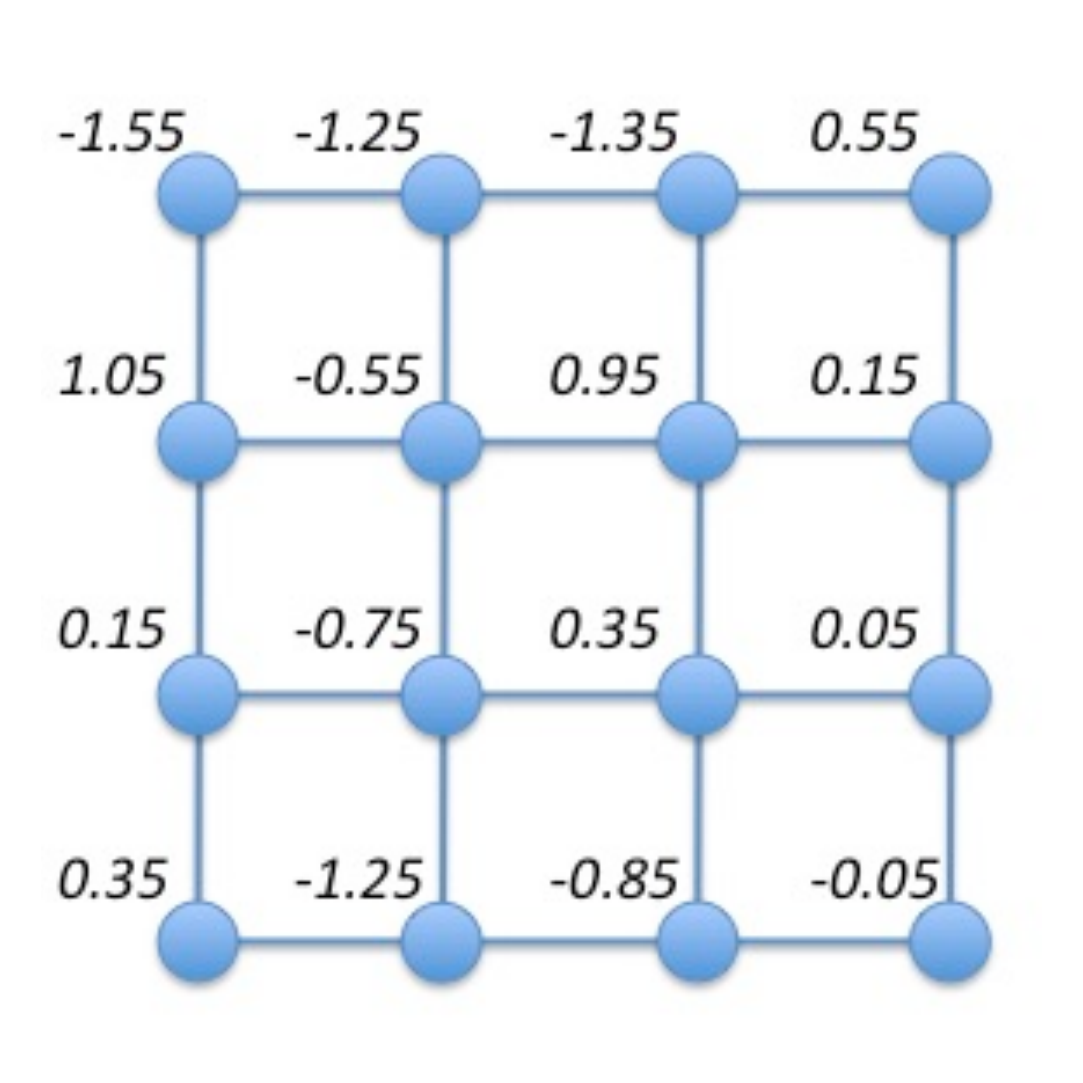}
\caption{Sample $\psi({\bf i})$ for 2-dimensional chain.}
\label{fig:2dpsi}
\end{figure}

\begin{center}
\begin{table}
\begin{tabular}{|c|c|c|}
\hline
$(i_1,i_2)$ & $(i_1,j_2(i_1,i_2))$ & $\text{ sign } \psi(i_1,j_2(i_1,i_2))$\\
\hline
$(1,1)$&$(1,4)$&-1\\
\hline
$(2,1)$&$(2,1)$&-1\\
\hline
$(3,1)$&$(3,2)$&1\\
\hline
$(4,1)$&$(4,2)$&1\\
\hline
$(1,2)$&$(1,1)$&1\\
\hline
$(2,2)$&$(2,2)$&-1\\
\hline
$(3,2)$&$(3,3)$&1\\
\hline
$(4,2)$&$(4,1)$&-1\\
\hline
$(1,3)$&$(1,2)$&1\\
\hline
$(2,3)$&$(2,3)$&-1\\
\hline
$(3,3)$&$(3,1)$&-1\\
\hline
$(4,3)$&$(4,3)$&1\\
\hline
$(1,4)$&$(1,3)$&1\\
\hline
$(2,4)$&$(2,4)$&-1\\
\hline
$(3,4)$&$(3,4)$&-1\\
\hline
$(4,4)$&$(4,4)$&1\\
\hline
\end{tabular}
\caption{Values of $j_2$ for example 2-dimensional graph}
\label{table:2dj2}
\end{table}
\end{center}

\begin{figure}
\includegraphics[width=2.5in]{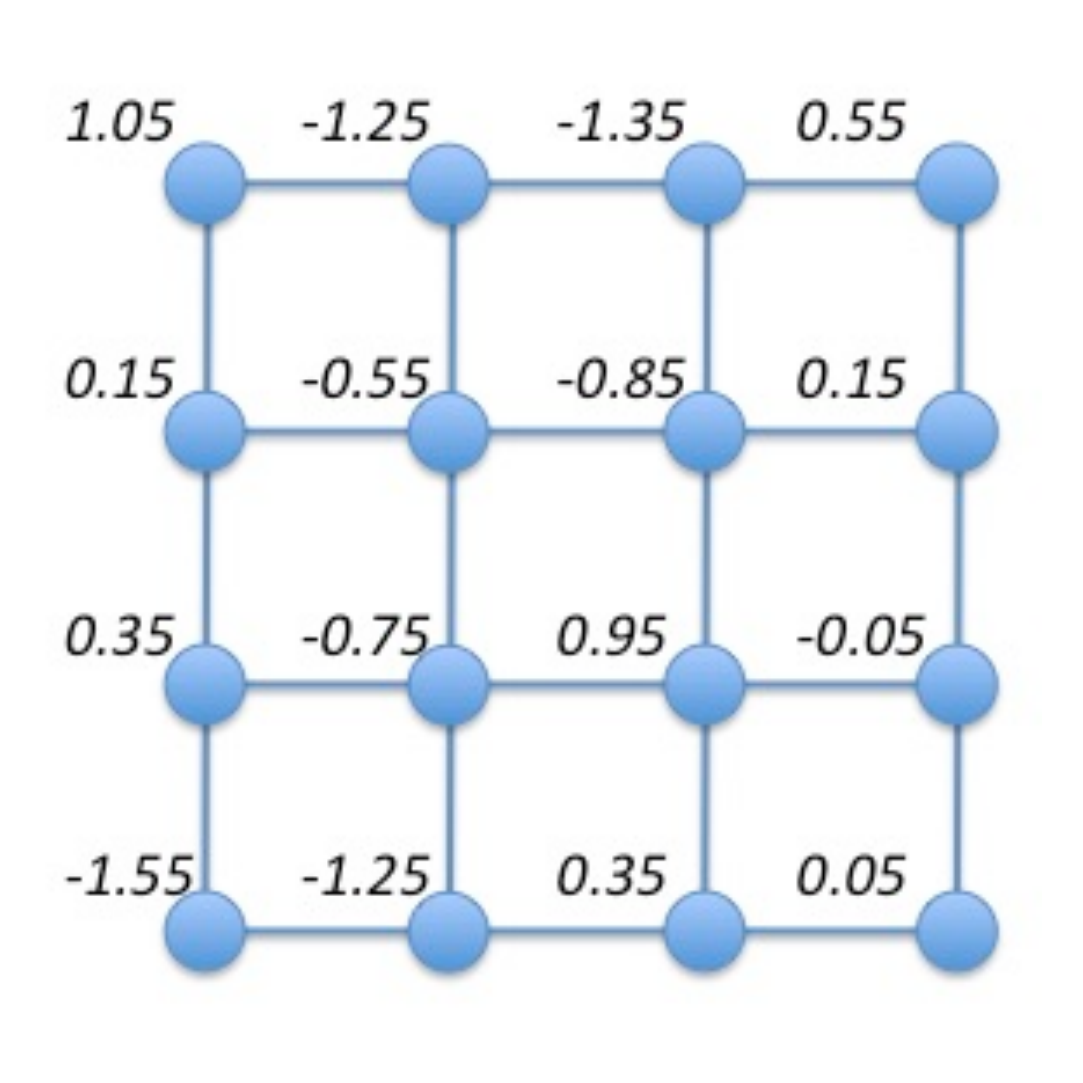}
\caption{Sample $\psi(P^{(2)}({\bf i},T^{(2)}))$ for 2-dimensional chain.}
\label{fig:2dpsiP2}
\end{figure}

\begin{figure}
\includegraphics[width=2.5in]{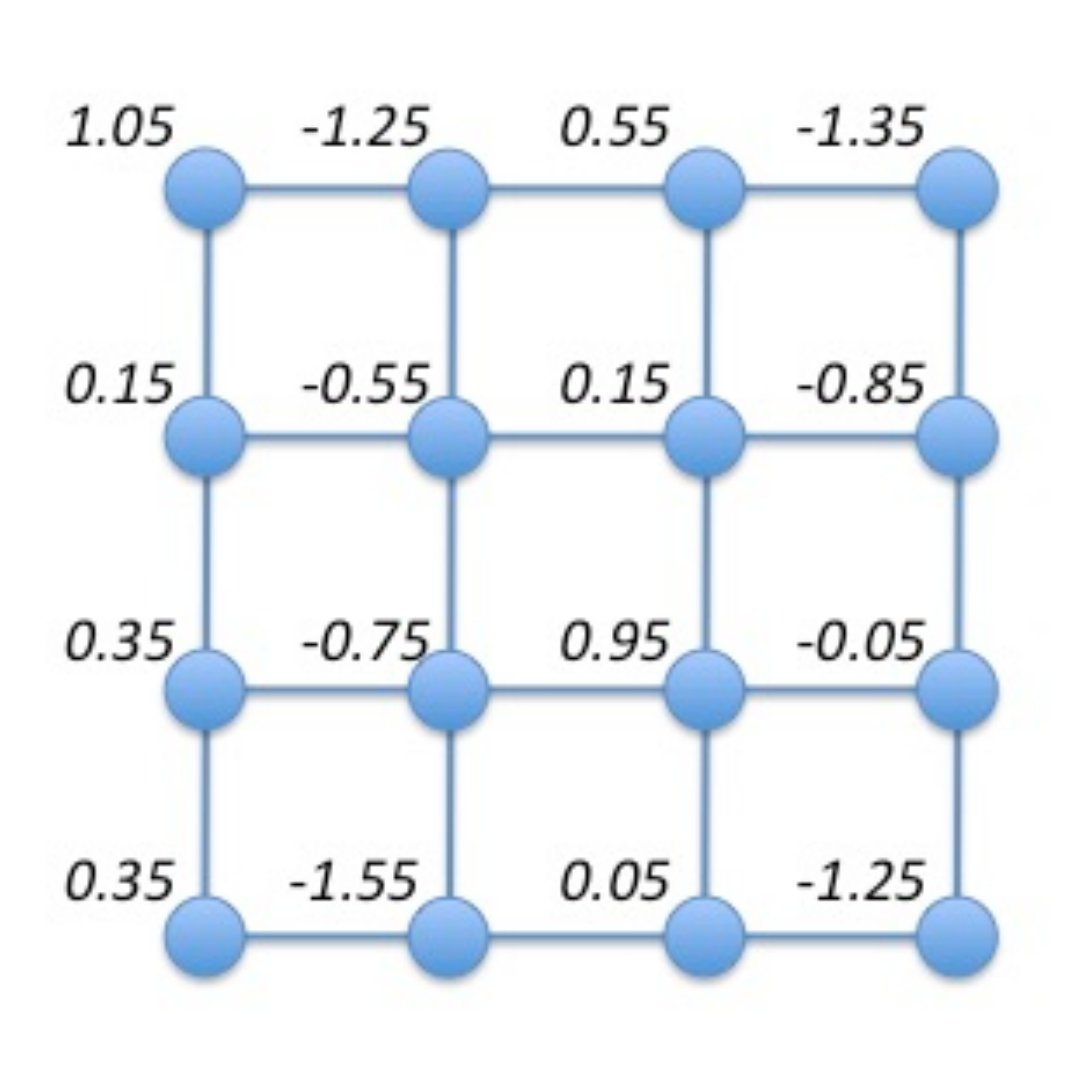}
\caption{Sample $\psi(P({\bf i},T))$ for 2-dimensional chain.}
\label{fig:2dpsiP}
\end{figure}

\noindent {\bf Example 3: M-dimensional chain}

Consider a graph of the form $ {\cal V}_1 \otimes \dots \otimes  {\cal V}_M$ defined analogously to the graphs of Ex. 1 and Ex. 2.  To bound $E_{\phi}$ for this graph,  we can reduce one dimension at a time, just as we did in Ex. 2.

Define $P^{(M)}$ with $T^{(M)} = N_M$ in terms of $j_M: \{1,\dots,N_1 \} \otimes \dots \otimes \{1,\dots,N_M \} \rightarrow \{1,\dots,N_M \}$.  Explicitly, set $P^{(M)}({\bf i},t) = (i_1,\dots,i_M + t\text{ sign }(j_M(i_1,\dots,i_M) - i_M))$ for $0 \le t \le |j_M(i_1,\dots,i_M) - i_M|$ and $P^{(M)}({\bf i},t) = (i_1,j_M(i_1,\dots,i_M))$ for $ |j_M(i_1,\dots,i_M) - i_M| \le t \le T^{(M)}$.

To define $j_M(i_1,\dots,i_M)$, let $R_0 = {\cal V}_1 \otimes \dots \otimes {\cal V}_M$ and start with $i_M = 1$.  Proceed inductively, letting ${\cal I}_{i_M}^{\cal P} = \{(i_1,\dots,i_{M-1})|  (i_1,\dots,i_{M-1},k_M) \in R_{i_M - 1} \Rightarrow (i_1,\dots,i_{M-1},k_M)  \in {\cal P} \}$, ${\cal I}_{i_M}^{\cal N} = \{(i_1,\dots,i_{M-1})| (i_1,\dots,i_{M-1},k_M) \in R_{i_M - 1} \Rightarrow (i_1,\dots,i_{M-1},k_M)  \in {\cal N} \}$, and ${\cal I}_{i_M} = \{1,\dots,N_1\} \otimes \dots \otimes  \{1,\dots,N_{M-1}\} - {\cal I}_{i_M}^{\cal P} \cup {\cal I}_{i_M}^{\cal N}$.    For each $(i_1,\dots,i_{M-1}) \in {\cal I}_{i_M}^{\cal P} \cup  {\cal I}_{i_M}^{\cal N}$, choose any $j_M(i_1,\dots,i_M)$ such that $(i_1,\dots,i_{M-1},j_M) \in R_{i_M - 1}$.  For $N_1 \dots N_{M-1}/2  -  \|{\cal I}_{i_M}^{\cal N}\|$ elements $(i_1,\dots,i_{M-1}) \in {\cal I}_{i_M}$, choose a specific $j_M(i_1,\dots,i_M)$ such that $(i_1,\dots,i_{M-1},j_M) \in {\cal N} \cap R_{i_M - 1}$.  For the remaining $N_1 \dots N_{M-1} /2  -  \|{\cal I}_{i_M}^{\cal P}\|$ elements $(i_1,\dots,i_{M-1}) \in {\cal I}_{i_M}$, choose a specific $j_M(i_1,\dots,i_M)$ such that $(i_1,\dots,i_{M-1},j_M) \in {\cal P} \cap R_{i_M - 1}$.   Set $R_{i_M} = R_{i_M-1} - \{(i_1,\dots,i_{M-1},j_M(i_1,\dots,i_M))|(i_1,\dots,i_{M-1}) \in \{1,\dots,N_1\} \otimes \dots \otimes  \{1,\dots,N_{M-1}\} \}$.  Increase $i_M$ by 1 and continue the induction until $j_M(i_1,\dots,i_M)$ is defined for all $i_M \le  N_M$.

After completing the definition of $P^{(M)}$, we have $\sum_{i_1 = 1}^{N_1} \dots \sum_{i_{M-1} = 1}^{N_{M-1}} sign\,\, \psi(P^{(M)}({\bf i},T^{(M)})) = 0$ for each fixed $i_M \in \{1,\dots, N_M \}$.  We then define $P^{(M-1)}_{i_M}(P^{(M)}({\bf i},T^{(M)}),t)$ for each $i_M \in  \{1,\dots,N_M \}$.  The definition progresses like that of $P^{(M)}$, with $M$ replaced everywhere by $M-1$.  Then we have $\sum_{i_1 = 1}^{N_1} \dots \sum_{i_{M-2} = 1}^{N_{M-2}} sign\,\, \psi(P^{(M-1)}_{i_M}(P^{(M)}({\bf i},T^{(M)}),T^{(M-1)})=0$ for each fixed $(i_{M-1},i_M)$.  We turn to $P^{(M-2)}_{i_{M-1},i_{M}}(P^{(M-1)}_{i_M}(P^{(M)}({\bf i},T^{(M)}),T^{(M-1)}),t)$, then to $P^{(M-3)}_{i_{M-2},i_{M-1},i_{M}}(P^{(M-2)}_{i_{M-1},i_{M}}(P^{(M-1)}_{i_M}(P^{(M)}({\bf i},T^{(M)}),T^{(M-1)}),T^{(M-2)}),t)$, and eventually to $P^{(2)}_{i_3,\dots,i_M}(P^{(3)}_{i_4,\dots,i_M}((\dots),T^{(3)}),t)$ satisfying $\sum_{i_1 = 1}^{N_1} sign\,\, \psi(P^{(2)}_{i_3,\dots,i_M}(P^{(3)}_{i_4,\dots,i_M}((\dots),T^{(3)}),T^{(2)})) = 0$ for each $(i_2,i_3,\dots,i_M)$.  The next step is to define $P^{(1)}_{i_2,\dots,i_M}(P^{(2)}_{i_3,\dots,i_M}((\dots),T^{(2)}),t)$ for each fixed $(i_2,\dots,i_M)$ by following the method of Ex. 1.  Finally, we string the series of maps $P^{(M)}({\bf i},t)$ to $P^{(1)}_{i_2,\dots,i_M}(P^{(2)}_{i_3,\dots,i_M}(\dots,T^{(2)}),t)$ together into
\[
P({\bf i},t) = \left\{
\begin{array}{lcc} 
P^{(M)}({\bf i},t) & & 0\le t \le T^{(M)}\\
P^{(M-1)}_{i_M}(P^{(M)}({\bf i},T^{(M)}),t-T^{(M)}) & & T^{(M)}\le t \le T^{(M)}+T^{(M-1)} \\
 \vdots &\hspace{0.1in} & \\
P^{(1)}_{i_2,\dots,i_M}(P^{(2)}_{i_3,\dots,i_M}((\dots),T^{(2)}),t-(T^{(M)}+\dots+T^{(2)})) & & T^{(M)}+\dots+T^{(2)}\le t \le T^{(M)}+\dots+T^{(1)} 
\end{array} \right.
\]
and set $T = T^{(M)}+\dots+T^{(1)} = N_1 +\dots + N_M$.  Noting that $B = max\,\, \{N_1,\dots,N_M\}$, we find $E_{\phi} \ge 2\mathcal{E}/(2N_1 +\dots + 2N_M+1)(2 max\,\, \{N_1,\dots,N_M\}+1) \ge 2\mathcal{E}/(2 M max\,\, \{N_1,\dots,N_M\}+1)(2 max\,\, \{N_1,\dots,N_M\}+1)$.  Note that there is an extra factor of $M$ in the denominator compared to the exact result $2\mathcal{E}(1-\cos \pi/ max\,\, \{N_1,\dots,N_M\}) \approx \pi^2\mathcal{E}/( max\,\, \{N_1,\dots,N_M\})^2$.

\noindent {\bf Example 4: M-dimensional chain with 2-qubit gates}
\begin{figure}
\includegraphics[width=3.5in]{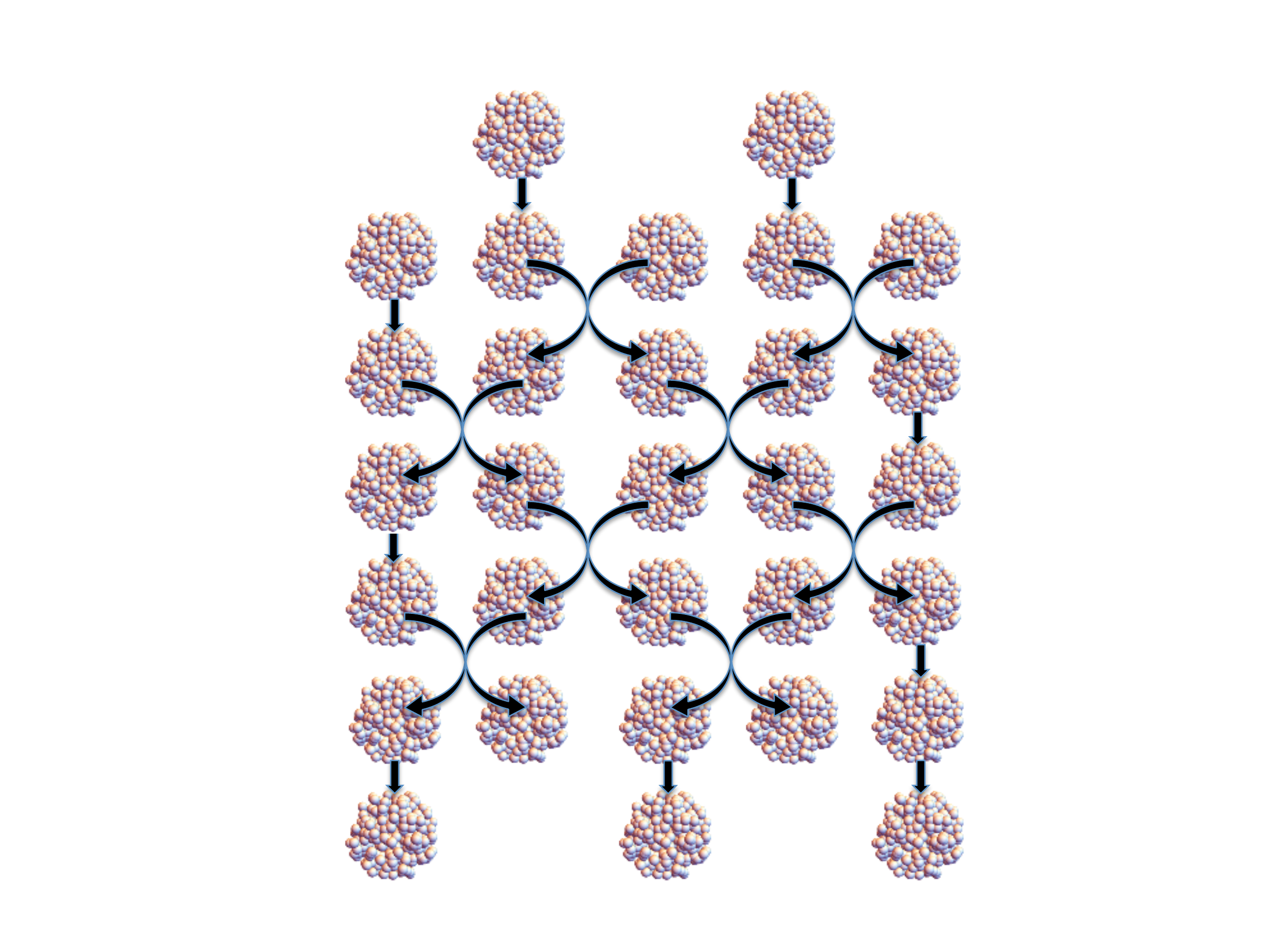}
\caption{Schematic physical apparatus for realizing the configuration of Ex. 4.}
\label{fig:nanocrystalarray1d}
\end{figure}

Consider the system depicted in Fig. \ref{fig:nanocrystalarray1d}, for which $M=5$ and $N=6$.  We study this system, generalized to arbitrary odd $M$ and even $N$.  The vertices ${\cal V}$ of its graph can be labelled $(k_1,\dots,k_M)$.   For odd $\alpha$, $k_\alpha$ is constrained to lie in the interval $\{1,\dots,N\}$, and, for even $\alpha$, $k_\alpha$ is constrained to lie in the interval $\{0,\dots,N-1\}$.  At every time step between $2$ and $N-1$, 2-qubit gates act, so there are many additional constraints on the indices due to the requirement that the penalty Hamiltonian (\ref{eq:penalty}) must vanish.

The graph has the following edges.  If $k_M$ is even, there is an edge between $(k_1,\dots,k_M)$ and $(k_1,\dots,k_M+1)$.  If $k_1$ and $k_2$ are even, there is an edge between $(k_1,k_2,\dots,k_M)$ and $(k_1+1,k_2+1,\dots,k_M)$.  If $k_3$ and $k_4$ are even, there is an edge between $(k_1,k_2,k_3,k_4,\dots,k_M)$ and $(k_1,k_2,k_3+1,k_4+1,\dots,k_M)$.  This pattern holds out to qubit $M-1$: if $k_{M-2}$ and $k_{M-1}$ are even, there is an edge between $(k_1,\dots,k_{M-2},k_{M-1},k_M)$ and $(k_1,\dots,k_{M-2}+1,k_{M-1}+1,k_M)$.    On the other hand, if $k_1$ is odd, there is an edge between $(k_1,\dots,k_M)$ and $(k_1+1,\dots,k_M)$.  If $k_2$ and $k_3$ are odd, there is an edge between $(k_1,k_2,k_3,\dots,k_M)$ and $(k_1,k_2+1,k_3+1,\dots,k_M)$.  If $k_4$ and $k_5$ are odd, there is an edge between $(k_1,k_2,k_3,k_4,k_5,\dots,k_M)$ and $(k_1,k_2,k_3,k_4+1,k_5+1,\dots,k_M)$.  This pattern holds out to qubit $M$: if $k_{M-1}$ and $k_M$ are odd, there is an edge between $(k_1,\dots,k_{M-1},k_M)$ and $(k_1,\dots,k_{M-1}+1,k_M+1)$.   Finally, there are some extra 1-qubit gates near $k_{\alpha}=1$ and $k_{\alpha} = N$.  For odd $\alpha$, there is an edge between $(k_1,\dots,k_\alpha=N-1,\dots,k_M)$ and $(k_1,\dots,k_\alpha=N,\dots,k_M)$.  For even $\alpha$, there is an edge between $(k_1,\dots,k_\alpha=0,\dots,k_M)$ and $(k_1,\dots,k_\alpha=1,\dots,k_M)$. 

To bound the gap of the system, it is useful to relabel the vertices $(k_1,\dots,k_M)$ by new indices $(i_1,\dots,i_M)$.  First, stipulate that $i_M \equiv k_M$.  Now, by definition, the penalty (\ref{eq:penalty}) vanishes on all the $(k_1,\dots,k_M)$.  Given a fixed value of $k_{\alpha}$,  one notices that $k_{\alpha-1}$ constrained to just 2 possible time step values.    We relabel these 2 values $i_{\alpha-1} = 0$ or $1$.  In particular, for $\alpha$ even, $k_{\alpha} = 0$ implies $k_{\alpha-1}=1$ or $2$ while  $k_{\alpha} > 0$ implies $k_{\alpha-1} = 2 \lfloor (k_\alpha+1)/2 \rfloor -1 $ or  $2 \lfloor (k_\alpha +1)/2 \rfloor$.  We define
\[
i_{\alpha-1} = \left\{ \begin{array}{cc} k_{\alpha-1}-2 &\,\,\,\, k_\alpha = 0 \\ k_{\alpha-1} - 2 \lfloor (k_\alpha+1)/2 \rfloor+1&\,\,\,\, k_\alpha >0   \end{array} \right. .
\]
For $\alpha$ odd,  $k_\alpha < N$ implies $k_{\alpha-1} = 2 \lfloor (k_\alpha)/2 \rfloor $ or $2\lfloor (k_\alpha)/2 \rfloor +1$ while $k_\alpha = N$ implies $k_{\alpha-1} =N-2$ or $N-1$.  We define
\[
i_{\alpha-1} = \left\{ \begin{array}{cc} k_{\alpha-1} -  2 \lfloor (k_\alpha)/2 \rfloor &\,\,\,\, k_\alpha < N \\ k_{\alpha-1}-(N-2) &\,\,\,\, k_\alpha = N \end{array} \right. .
\]
Since $i_{\alpha-1} = \{0,1\}$ and $i_M \in \{1, \dots, N \}$, there are $\|{\cal V}\| =  2^{M-1} N$ vertices $(i_1,\dots,i_M) \in {\cal V}$.

The graph has the following edges in terms of the relabelled coordinates $i_\alpha$.  The 2-qubit gates provide edges between $(i_1,\dots,i_{\alpha-1}=1,i_\alpha=0,\dots,i_M)$ and $(i_1,\dots,i_{\alpha-1}-1=0,i_\alpha+1=1,\dots,i_M)$.  The 1-qubit gates provide edges between $(i_1,\dots,i_{M-1},i_M)$ and $(i_1,\dots,i_{M-1},i_M+1)$ when $i_M$ is even and between $(i_1=0,\dots,i_{M-1},i_M)$ and $(i_1+1=1,\dots,i_{M-1},i_M)$.  There are additional 1-qubit gate edges at the beginning and end of the computation.

We wish to bound the energy $E_\phi$ for some $\phi({\bf i})$ orthogonal to the ground state, using Lem. \ref{lemma:path}.  It is therefore necessary to establish a map $P^{(M)}$.  First,  we define $j_M(i_1,\dots,i_M)$ so that, as in Ex. 3, $\sum_{i_1 = 0}^{1} \dots \sum_{i_{M-1} = 0}^{1} sign\,\, \psi(i_1,\dots,i_{M-1},j_M(i_1,\dots,i_M)) = 0$ for each fixed $i_M \in \{1,\dots, N\}$.  In contrast to Ex. 3, $(i_1,\dots,i_{M-1},i_M + t\text{ sign }(j_M(i_1,\dots,i_{M-1}) - i_M))$ is an invalid choice here for $P^{(M)}$ because $i_M$ cannot in general advance independently of $i_{M-1}$: there are constraints on the coordinates.  Taking these constraints into account, we need to choreograph the path $Q^{(M)}((i_1,\dots,i_{M}),t)$ that takes $(i_1,\dots,i_{M-1},i_M)$ to $(i_1,\dots,i_{M-1},i_M+1)$.  We will then define $P^{(M)}$ in terms of $Q^{(M)}$.

Start by setting $Q^{(M)}((i_1,\dots,i_{M}),t=0) = (i_1,\dots,i_{M})$.  If $i_M \text{ mod } 2 \equiv 0$ or $i_M = N-1$, there is an edge between $(i_1,\dots,i_{M-1},i_M)$ and $(i_1,\dots,i_{M-1},i_M+1)$ via a 1-qubit gate.  We can then define $Q^{(M)}((i_1,\dots,i_{M}),t) = (i_1,\dots,i_{M}+1)$ for $1 \le t \le 2M$, and $Q^{(M)}$ is successfully specified.  If $i_M \text{ mod } 2 \equiv 1$, the situation is more complicated since qubit $M$ can only advance in tandem with qubit $M-1$.  We distinguish 2 cases.

(1) If $i_{M-1} = 0$, set $\alpha = M-1$ and $\tau =0$, and perform a ``downward sweep'' of $\alpha$ as follows.  If $(i_1(\tau),\dots,i_{\alpha}(\tau),\dots,i_M(\tau)) \equiv Q^{(M)}((i_1,\dots,i_{M}),\tau)$ has $i_{\alpha}(\tau) = 0$, determine if there is a direct edge from $(i_1(\tau),\dots,i_{\alpha}(\tau) = 0,\dots,i_M(\tau))$ to $(i_1(\tau),\dots,i_{\alpha}(\tau)+1 = 1,\dots,i_M(\tau))$ due to a 1-qubit gate.  If such a direct edge is available, set $Q^{(M)}((i_1,\dots,i_M),\tau+1) = (i_1(\tau),\dots,i_{\alpha}(\tau)+1 = 1,\dots,i_M(\tau))$, increment $\tau$, and halt the downward sweep of $\alpha$.   If $i_{\alpha}(\tau) = 0$ but no direct edge is available, examine $i_{\alpha -1}(\tau)$.  If $i_{\alpha-1}(\tau) = 1$, then there is an edge from $(i_1(\tau),\dots,i_{\alpha-1}(\tau)=1,i_{\alpha}(\tau) = 0,\dots,i_{M}(\tau))$ to $(i_1(\tau),\dots,i_{\alpha-1}(\tau)-1=0,i_{\alpha}(\tau)+1 = 1,\dots,i_{M}(\tau))$ due to a 2-qubit gate.  Set $Q^{(M)}((i_1,\dots,i_{M}),\tau+1) = (i_1(\tau),\dots,i_{\alpha-1}(\tau)-1=0,i_{\alpha}(\tau) +1 = 1,\dots,i_{M}(\tau))$.  Increment $\tau$, decrement $\alpha$, and continue the downward sweep.  If the downward sweep reaches $\alpha=1$, we will have $i_1(\tau) = 0$, and there will always be a direct edge from $(i_1(\tau)=0,i_2(\tau),\dots,i_M(\tau))$ to $(i_1(\tau)+1=1,i_2(\tau),\dots,i_M)$ due to a 1-qubit gate.  Thus, the downward sweep will inevitably halt.  When the downward sweep halts at some $\alpha = A$, one finds that inevitably $Q^{(M)}((i_1,\dots,i_{M}),\tau) = (i_1,\dots,i_{A-1},1,i_{A},i_{A+1},\dots,i_{M-2},i_{M})$, with $i_{A}$, $i_{A+1}$, $\dots$, $i_{M-2}$ each translated one slot over compared to the initial vertex $(i_1,\dots,i_{A-1},i_{A},i_{A+1},\dots,i_{M-2},i_{M-1},i_{M})$.

We then begin an ``upward sweep.''  If $i_A = 1$, we note that $(i_1,\dots,i_{A-1},1,i_{A} ,i_{A+1},\dots,i_{M-2},i_{M})=(i_1,\dots,i_{A-1},i_A,1,i_{A+1} ,\dots,i_{M-2},i_{M})$.   If $i_A = 0$, then there is an edge from $(i_1,\dots,i_{A-1},1,i_{A} ,i_{A+1},\dots,i_{M-2},i_{M})$ to $(i_1,\dots,i_{A-1},i_{A},1,i_{A+1},\dots,i_{M-2},i_{M})$ due to a 2-qubit gate.  Thus, we increment $\tau$ and set $Q^{(M)}((i_1,\dots,i_{M}),\tau) = (i_1,\dots,i_{A-1},i_A,1,i_{A+1} ,\dots,i_{M-2},i_{M})$.  For either value of $i_A$, we are left with $Q^{(M)}((i_1,\dots,i_{M}),\tau) = (i_1,\dots,i_{A-1},i_{A},1,i_{A+1},\dots,i_{M-2},i_{M})$.  Repeat this process for $\alpha = A+1$, moving the extra 1 to the right to obtain  $Q^{(M)}((i_1,\dots,i_{M}),\tau) = (i_1,\dots,i_{A-1},i_{A},i_{A+1},1,\dots,i_{M-2},i_{M})$.  Continue the upward sweep of $\alpha$ until $Q^{(M)}((i_1,\dots,i_{M}),\tau) = (i_1,\dots,i_{M-2},1,i_{M})$.  Increment $\tau$ and make the final advance on the edge from $(i_1,\dots,i_{M-2},1,i_{M})$ to $(i_1,\dots,i_{M-2},0,i_{M}+1)$: $Q^{(M)}((i_1,\dots,i_{M}),\tau) = (i_1,i_2,i_3,\dots,i_{M-2},i_{M-1} = 0,i_{M}+1)$.  At this point, $\tau$ is at most $2M$.  If $\tau < 2M$, fix $Q^{(M)}((i_1,\dots,i_{M}),\tau) = (i_1,i_2,i_3,\dots,i_{M-2},i_{M-1},i_{M}+1)$ for all remaining $\tau$ up to $2M$.  This completes the definition of $Q^{(M)}$ when $i_{M-1} = 0$.

(2) If $i_{M-1} = 1$, set $\tau = 1$ and define $Q^{(M)}((i_1,\dots,1,i_{M}),\tau) = (i_1,\dots,i_{M-2},0,i_{M}+1)$.  Then, follow the downward sweep and upward sweep procedures described in case (1).  Just omit the final advance since $i_M$ has already been increased to $i_M+1$.

This definition of $Q^{(M)}((i_1,\dots,i_{M}),t)$ for $0 \le t \le 2M$ ensures that $Q^{(M)}((i_1,\dots,i_{M}),0) = (i_1,\dots,i_{M})$ and $Q^{(M)}((i_1,\dots,i_{M}),2M) = (i_1,\dots,i_{M}+1)$ and that $Q^{(M)}$ advances one edge at a time.  A graphical illustration of $Q^{(M)}((0,0,1,0,3),t)$ appears in Fig. \ref{fig:Qdef}, for the case $M=5$ and $N=6$ appearing in Fig. \ref{fig:nanocrystalarray1d}.

\begin{figure}[htp]
\subfloat[Configuration of particles with $i_1 = 0$, $i_2 = 0$, $i_3=1$, $i_4=0$, $i_5 = 3$, corresponding to point $Q^{(5)}((0,0,1,0,3),0) = (0,0,1,0,3)$.]{%
  \includegraphics[width=3.5in]{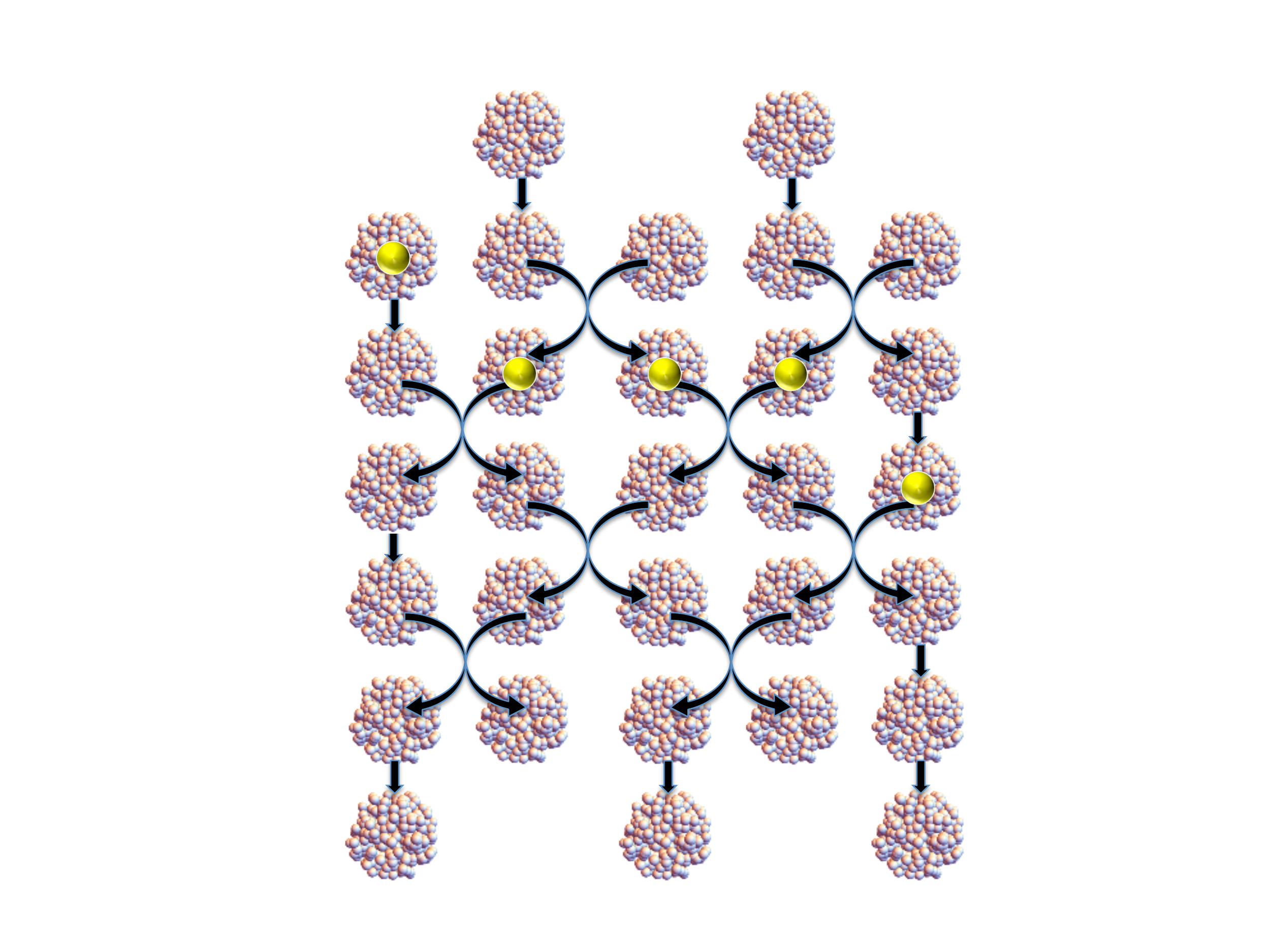}%
}
\subfloat[Configuration of particles with $i_1 = 0$, $i_2 = 0$, $i_3=0$, $i_4=1$, $i_5 = 3$, corresponding to point $Q^{(5)}((0,0,1,0,3),1) = (0,0,0,1,3)$.]{%
  \includegraphics[width=3.5in]{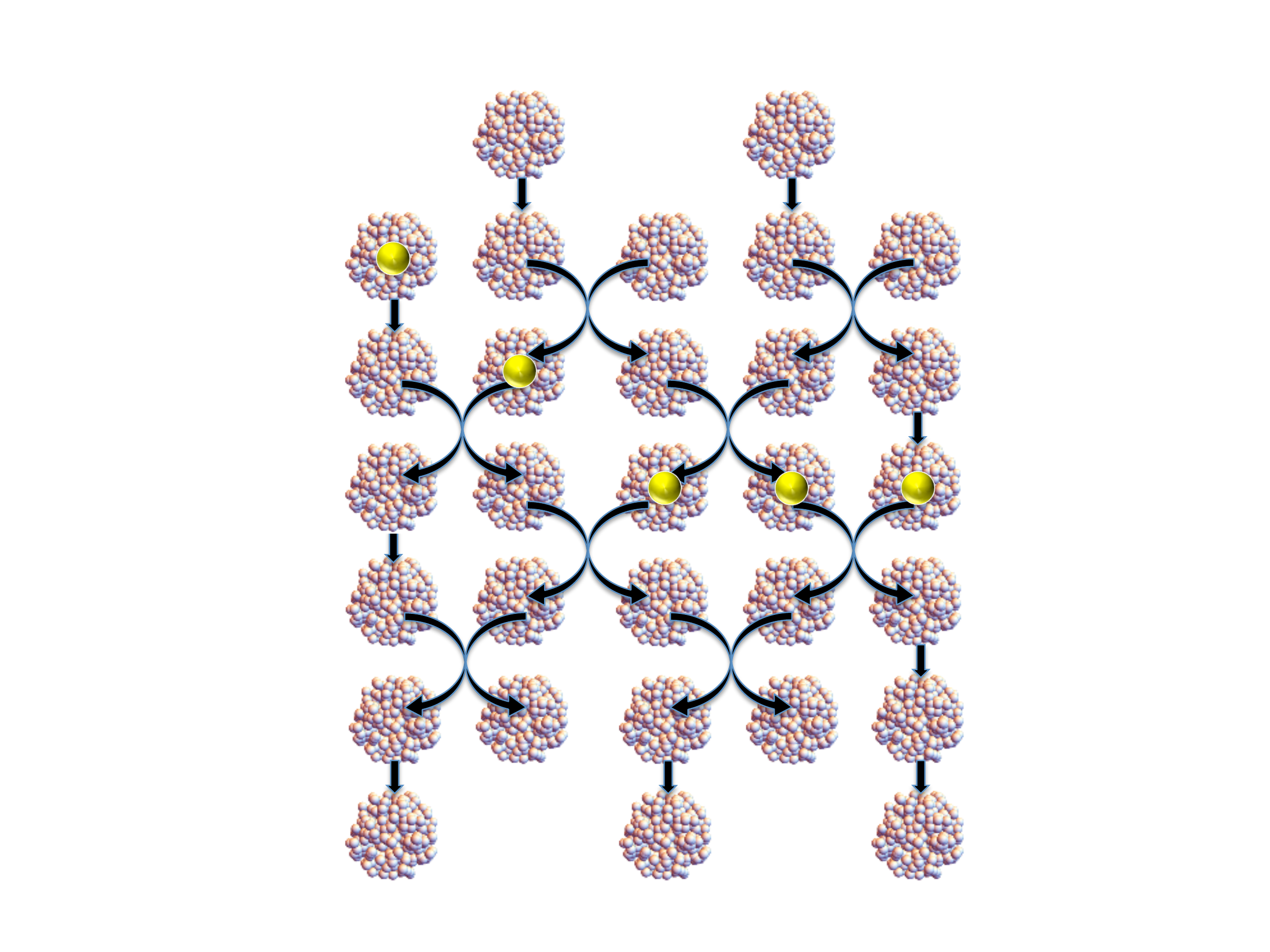}%
}

\subfloat[Configuration of particles after downward sweep with $i_1 = 1$, $i_2 = 0$, $i_3=0$, $i_4=1$, $i_5 = 3$, corresponding to point $Q^{(5)}((0,0,1,0,3),2) = (1,0,0,1,3)$.]{%
  \includegraphics[width=3.5in]{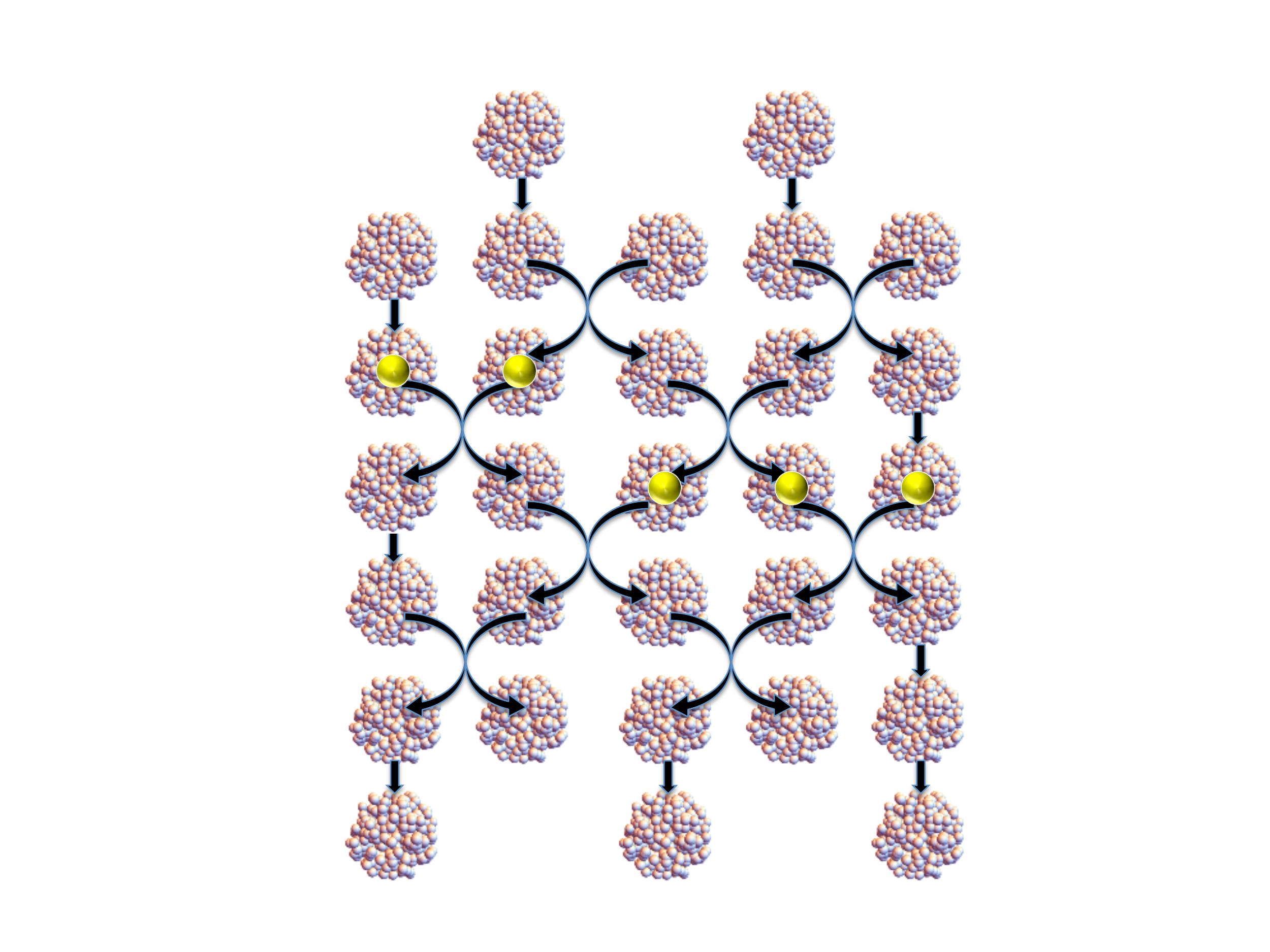}%
}
\subfloat[Configuration of particles with $i_1 = 0$, $i_2 = 1$, $i_3=0$, $i_4=1$, $i_5 = 3$, corresponding to point $Q^{(5)}((0,0,1,0,3),3) = (0,1,0,1,3)$.]{%
  \includegraphics[width=3.5in]{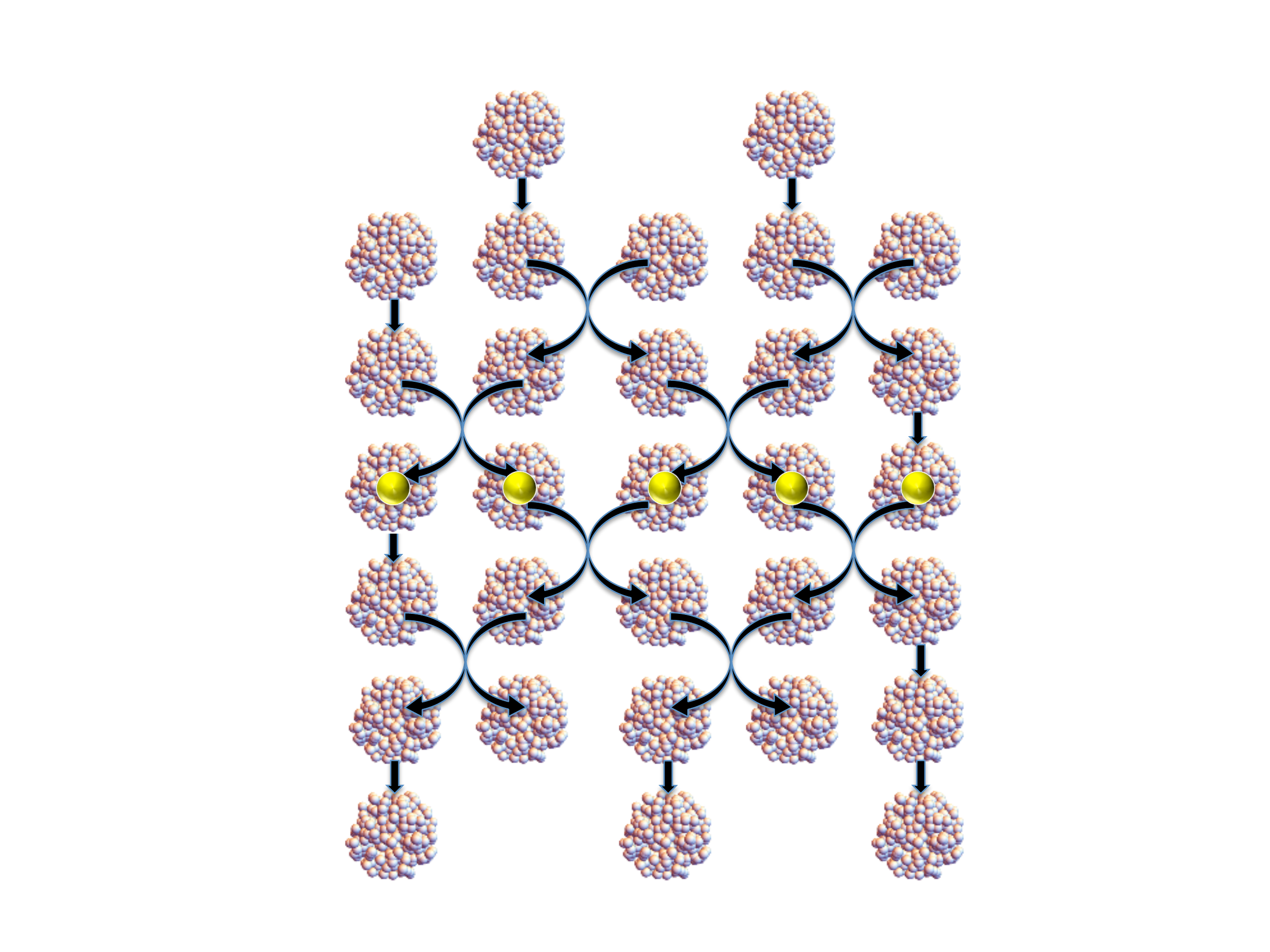}%
}

\subfloat[Configuration of particles with $i_1 = 0$, $i_2 = 0$, $i_3=1$, $i_4=1$, $i_5 = 3$, corresponding to point $Q^{(5)}((0,0,1,0,3),4) = (0,0,1,1,3)$.]{%
  \includegraphics[width=3.5in]{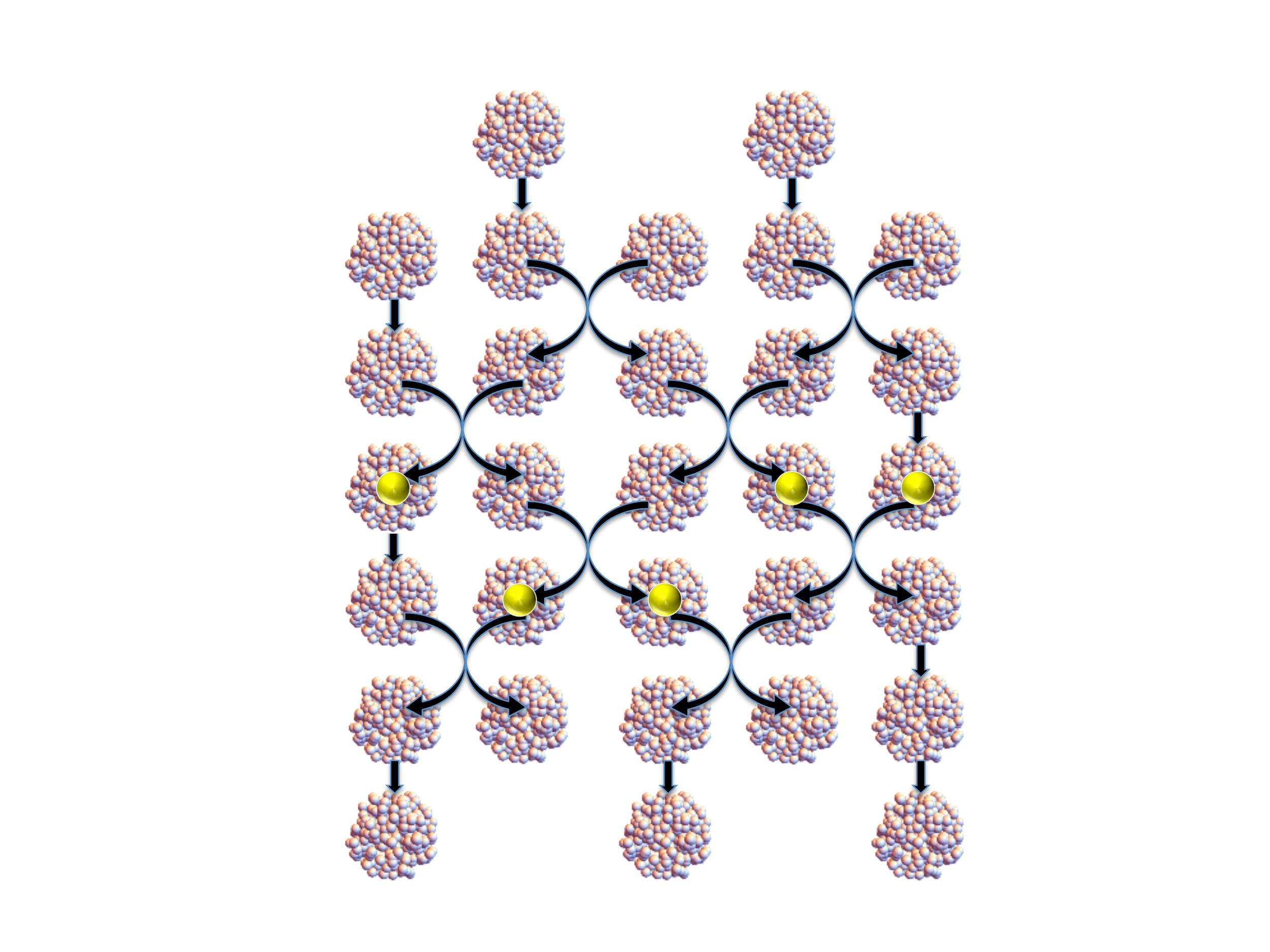}%
}
\subfloat[Configuration of particles after upward sweep with $i_1 = 0$, $i_2 = 0$, $i_3=1$, $i_4=0$, $i_5 = 4$, corresponding to point $Q^{(5)}((0,0,1,0,3),5) = (0,0,1,0,4)$ in the graph.  Value of $i_5$ has advanced by 1 while all other coordinates $i_\alpha$ have returned to their initial values.]{%
  \includegraphics[width=3.5in]{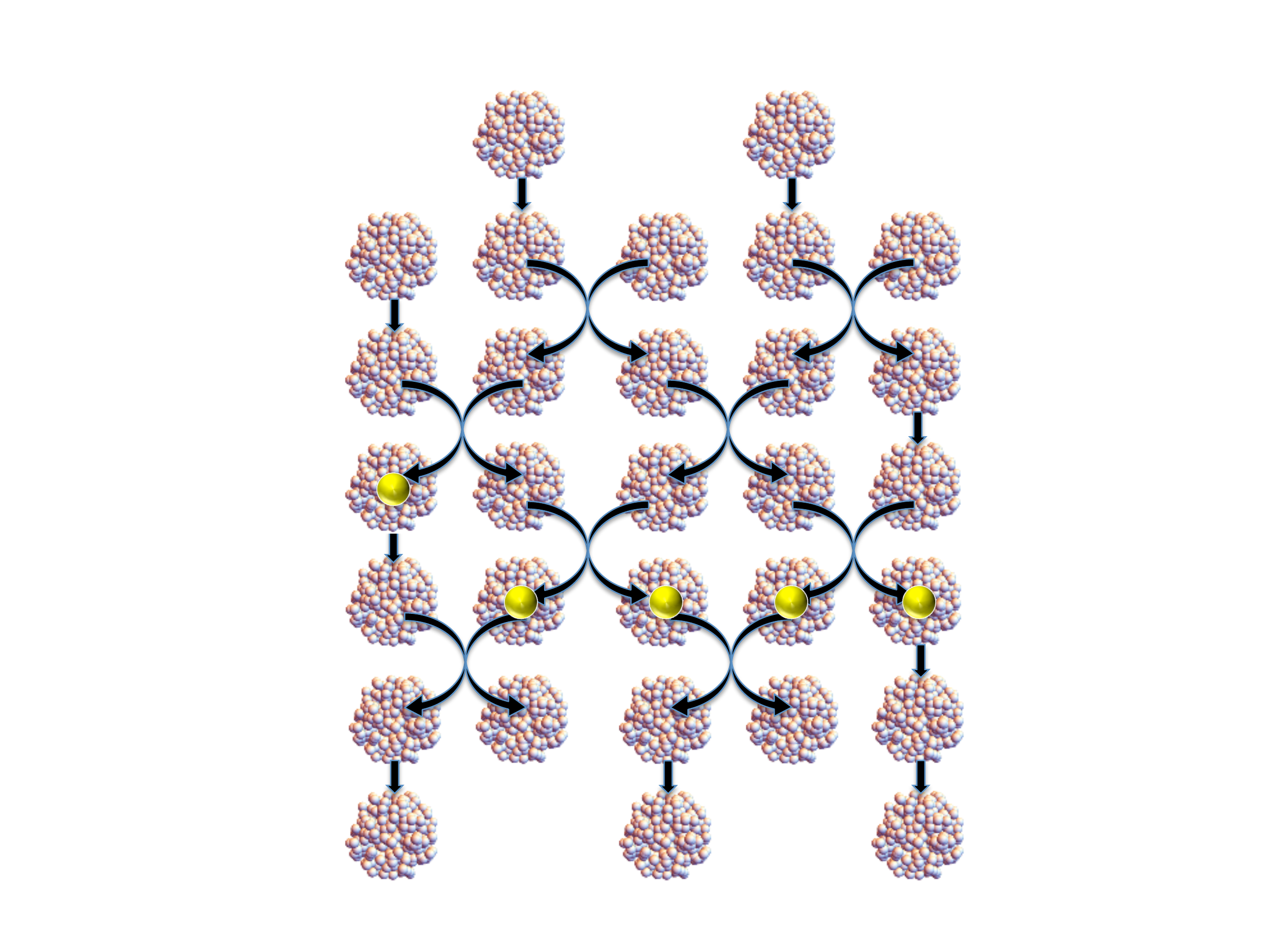}%
}
\caption{$Q^{(5)}((0,0,1,0,3),t)$.  Note $Q^{(5)}((0,0,1,0,3),t) = Q^{(5)}((0,0,1,0,3),5)$ for $t=6,\dots,10$.}
\label{fig:Qdef}
\end{figure}

We will want to increment $i_M$ by more than 1 and to decrement $i_M$ as well.  So, we extend this definition of $Q^{(M)}((i_1,\dots,i_{M}),t)$ to allow $t > 2M$ and $t < 0$ by stitching multiple $Q^{(M)}$ maps together:

\[
Q^{(M)}((i_1,\dots,i_{M}),t)= \left\{
\begin{array}{lcc} 
Q^{(M)}((i_1,\dots,i_{M}+r),t-2Mr) & & 2Mr \le t \le 2M(r+1)\\
Q^{(M)}((i_1,\dots,i_{M}-r),2Mr-\vert t \vert) & & -2Mr\le t \le -2M(r-1)\\
\end{array} \right. .
\]

Using this definition of $Q^{(M)}$, we fix $T^{(M)} = 2 M N$ and write $P^{(M)}((i_1,\dots,i_{M}),t) = Q^{(M)}((i_1,\dots,i_{M}),t\text{ sign }(j_M(i_1,\dots,i_{M-1}) - i_M))$ for $0 \le t \le 2M \vert j_M(i_1,\dots,i_M) - i_M\vert $ and $P^{(M)}((i_1,\dots,i_{M}),t) = (i_1,\dots,j_M(i_1,\dots,i_M))$ for $2M \vert j_M(i_1,\dots,i_M) - i_M\vert \le t \le T^{(M)}$.   As desired, this definition of $P^{(M)}$ implies $\sum_{i_1 = 0}^{1} \dots \sum_{i_{M-1} = 0}^{1} sign\,\, \psi(P^{(M)}((i_1,\dots,i_M),T^{(M)})) = 0$ for each fixed $i_M \in \{1,\dots, N \}$.

With $P^{(M)}$ specified, we proceed by analogy to Ex. 3 and define $P^{(M-1)}_{i_M}(P^{(M)}((i_1,\dots,i_M),T^{(M)}),t)$ for each $i_M \in  \{1,\dots,N \}$.  The definition progresses like that of $P^{(M)}$, with $M$ replaced everywhere by $M-1$.  However, since $i_{M-1} \in \{0,1\}$, we have $T^{(M-1)} = 2(M-1)$ and need not extend the definition of $Q^{(M-1)}((i_1,\dots,i_{M}),t)$ beyond $-2(M-1) \le t \le 2(M-1)$.  After $P^{(M-1)}$ comes $P^{(M-2)}$ and so on.  Following Ex. 3, we eventually define $P((i_1,\dots,i_M),t)$ and $T = T^{(M)}+\dots+T^{(1)} = 2M N + 2(M-1) +\dots+2(1) = M(2 N+ M-1)$.

To apply Lem. \ref{lemma:path}, we need the value of $B$.  Fixing our attention on a given edge, we first ask how many times $Q^{(M)}((i_1,\dots,i_{M}),t)$ uses this edge for $0\le t\le 2M$.  If $i_M=N-1$ or $i_M \text{ mod } 2 \equiv 0$, $Q^{(M)}$ uses the edge between $(i_1,\dots,i_{M})$ and $(i_1,\dots,i_{M}+1)$.  This edge is not used for any other starting vertex.  If $i_M \text{ mod } 2 \equiv 1$, things are more involved.  If $i_{M-1} = 1$, $Q^{(M)}((i_1,\dots,i_{M}),t)$ starts with the edge between $(i_1,\dots,i_{M-1},i_{M})$ and $(i_1,\dots,0,i_{M}+i_{M-1})$.  This edge is used by at most 2 starting vertices: the first edge of the ``downward sweep'' used by the starting vertex $i_1,\dots,i_{M-1}=1,i_{M}$ and the final edge of the ``upward sweep'' used by the starting vertex $i_1,\dots,i_{M-1}=0,i_{M}$.  In general, the downward sweep  involves edges between $(i_1,\dots,i_{\alpha-1}=1,i_{\alpha} = 0,\dots,i_{M})$ and $(i_1,\dots,i_{\alpha-1}-1=0,i_{\alpha} +1 = 1,\dots,i_{M})$. During the downward sweep, each edge can be determined using only knowledge of the previous edge used.  The downward sweep completes in $(i_1,\dots,i_A,1,i_{A+1},\dots,i_{M}+i_{M-1})$, from which the starting vertex $(i_1,\dots,i_{M})$ can be uniquely extracted since $i_M \text{ mod } 2 \equiv 0$.  Since each edge of the downward sweep cascades to a unique end vertex from which the starting vertex $(i_1,\dots,i_{M})$ can be uniquely identified, two distinct starting vertices cannot use the same edge during their downward sweeps.  After the downward sweep, $Q^{(M)}$ uses an upward sweep starting from $(i_1,\dots,i_A,1,i_{A+1},\dots,i_{M}+i_{M-1})$.  Each edge used in the upward sweep can be determined using only knowledge of the previous edge used.  The upward sweep completes in $(i_1,\dots,i_{M}+1)$, which is clearly a unique function of the original starting vertex $(i_1,\dots,i_{M})$.  Thus, 2 distinct starting vertices cannot use the same edge during their upward sweeps.  We conclude that a given edge can be used by at most 2 starting vertices -- it can be used in the downward sweep of one starting vertex and the upward sweep of a second starting vertex.   This is true for all $t$ up to $2M$.  It is possible that the same edge will be used 2 more times whenever $t$ increases by another $2M$.  Since $T^{(M)} = 2 MN$,  we see that $P^{(M)}$ can use a given edge at most $2N$ times.  For $\alpha < M$,  there are $2$ values of $i_\alpha$ rather than $N$ values, so $P^{(\alpha)}$ can use an edge at most $2(2)$ times.  It follows that $B \le 2 (N)+2 (2+\dots+2) = 2(N+2M-2)$.  Thus, Lem. \ref{lemma:path} implies that $E_{\phi} \ge 2\mathcal{E}/(2M(2N+M-1)+1)(4(N+2M-2)+1)$.

\end{document}